\documentclass[12pt]{article}
\usepackage{rascal-lst}
\usepackage[hyphens]{url}
\usepackage{hyperref}
\usepackage{breakurl}
\usepackage[T1]{fontenc}
\usepackage[]{microtype}
\usepackage[margin=1in]{geometry}
\usepackage[]{bbm}
\usepackage{amsmath,amsthm,amssymb,amsfonts,mathtools}
\usepackage{thmtools}
\usepackage{thm-restate}
\usepackage{stmaryrd}
\usepackage[]{semantic}
\usepackage[]{listings}
\usepackage[]{xcolor}
\usepackage[]{galois}
\usepackage{dutchcal}
\usepackage{booktabs}
\usepackage[]{cancel} 
\usepackage[]{graphicx}
\usepackage{mathtools}
\usepackage[]{cleveref}
\usepackage[]{textgreek}
\usepackage{sourcecodepro}
\usepackage[round]{natbib}
\usepackage[inline]{enumitem}
\usepackage{tcolorbox}
\setlistdepth{20}
\renewlist{itemize}{itemize}{20}

\setlist[itemize]{label=$\cdot$}

\setlist[itemize,1]{label=\textbullet}
\setlist[itemize,2]{label=--}
\setlist[itemize,3]{label=*}
\setlist[itemize,4]{label=$\cdot$}
\setlist[itemize,5]{label=\textbullet}
\setlist[itemize,6]{label=--}
\setlist[itemize,7]{label=*}

\declaretheorem[name=Lemma]{lemma}
\declaretheorem[name=Example,style=remark]{example}
\declaretheorem[name=Remark,style=remark]{remark}

\xdefinecolor{SolarizedBlue}{HTML}{268bd2}
\xdefinecolor{SolarizedCyan}{HTML}{2aa198}
\xdefinecolor{SolarizedBase1}{HTML}{93a1a1}
\xdefinecolor{SolarizedBase0}{HTML}{839496}
\xdefinecolor{SolarizedBase01}{HTML}{586e75}

\newtcbox{\highlightbox}[1][gray]{on line,
arc=0pt,outer arc=0pt,colback=#1!10!white,colframe=#1!50!black,
boxsep=0pt,left=1pt,right=1pt,top=.5pt,bottom=.5pt,
boxrule=0pt,bottomrule=.1pt,toprule=.1pt}

\allowdisplaybreaks

\newcommand\namedset[1]{\ensuremath{\mathrm{#1}}}

\newcommand{\ttuple}[1]{\ensuremath{\langle#1\rangle}}
\newcommand{\auxiliary}[1]{\textrm{#1}}
\newcommand{\esequence}[1]{\underline{\ensuremath{\mathstrut #1}}}
\newcommand{\keyword}[1]{\text{\upshape \textsf{#1}}}
\newcommand{\rascaltype}[1]{\keyword{#1}}
\newcommand{\typing}[2]{#1 : #2}
\newcommand{\subtyping}[2]{#1 <: #2}
\newcommand{\notsubtyping}[2]{#1 \not<: #2}
\newcommand{\reconstruct}[3]{\textbf{recons } #1 \textbf{ using } #2
  \textbf{ to } #3}
\newcommand{\match}[4]{#3 |- #1 \overset{?}{\coloneqq} #2 \xRightarrow[\textrm{match}]{} #4}
\newcommand{\matchall}[6]{#3 |- #1 \overset{?}{\coloneqq} #2 ~|~ #4 \xRightarrow[\textrm{match$\star$}]{#5} #6}
\newcommand{\evalexpr}[4]{#1; #2 \xRightarrow[\mathrm{expr}]{} #3; #4}
\newcommand{\evalexprfuel}[5]{#1; #2 \xRightarrow[\mathrm{expr}]{}^{#5} #3; #4}
\newcommand{\evalexprstar}[4]{#1; #2 \xRightarrow[\mathrm{expr{\star}}]{} #3; #4}
\newcommand{\evalexprstarfuel}[5]{#1; #2 \xRightarrow[\mathrm{expr{\star}}]{}^{#5} #3; #4}
\newcommand{\evaleach}[5]{#1; #2; #3 \xRightarrow[\mathrm{each}]{} #4; #5}
\newcommand{\evaleachfuel}[6]{#1; #2; #3 \xRightarrow[\mathrm{each}]{}^{#6} #4; #5}
\newcommand{\evalgenexpr}[4]{#1; #2 \xRightarrow[\mathrm{gexpr}]{} #3; #4}
\newcommand{\evalgenexprfuel}[5]{#1; #2 \xRightarrow[\mathrm{gexpr}]{}^{#5} #3; #4}
\newcommand{\evalcase}[5]{#1; #2; #3 \xRightarrow[\mathrm{case}]{} #4 ; #5}
\newcommand{\evalcasefuel}[6]{#1; #2; #3 \xRightarrow[\mathrm{case}]{}^{#6} #4 ; #5}
\newcommand{\evalcases}[5]{#1; #2; #3 \xRightarrow[\mathrm{cases}]{} #4 ; #5}
\newcommand{\evalcasesfuel}[6]{#1; #2; #3 \xRightarrow[\mathrm{cases}]{}^{#6} #4 ; #5}
\newcommand{\evalvisit}[6]{#2; #3; #4 \xRightarrow[\mathrm{visit}]{#1} #5 ; #6}
\newcommand{\evalvisitfuel}[7]{#2; #3; #4 \xRightarrow[\mathrm{visit}]{#1}^{#7} #5 ; #6}
\newcommand{\evaltdvisit}[6]{#1; #2; #3 \xRightarrow[\mathrm{td-visit}]{#4} #5 ; #6}
\newcommand{\evaltdvisitfuel}[7]{#1; #2; #3 \xRightarrow[\mathrm{td-visit}]{#4}^{#7} #5 ; #6}
\newcommand{\evaltdvisitstar}[6]{#1; #2; #3 \xRightarrow[\mathrm{td-visit{\star}}]{#4} #5 ; #6}
\newcommand{\evaltdvisitstarfuel}[7]{#1; #2; #3 \xRightarrow[\mathrm{td-visit{\star}}]{#4}^{#7} #5 ; #6}
\newcommand{\evalbuvisit}[6]{#1; #2; #3 \xRightarrow[\mathrm{bu-visit}]{#4} #5 ; #6}
\newcommand{\evalbuvisitfuel}[7]{#1; #2; #3 \xRightarrow[\mathrm{bu-visit}]{#4}^{#7} #5 ; #6}
\newcommand{\evalbuvisitstar}[6]{#1; #2; #3 \xRightarrow[\mathrm{bu-visit{\star}}]{#4} #5 ; #6}
\newcommand{\evalbuvisitstarfuel}[7]{#1; #2; #3
  \xRightarrow[\mathrm{bu-visit{\star}}]{#4}^{#7} #5 ; #6}
\newcommand{\sem}[1]{\llbracket#1\rrbracket}

\xdefinecolor{SolarizedBlue}{HTML}{268bd2}
\xdefinecolor{SolarizedCyan}{HTML}{2aa198}
\xdefinecolor{SolarizedBase1}{HTML}{93a1a1}
\xdefinecolor{SolarizedBase0}{HTML}{839496}
\xdefinecolor{SolarizedBase01}{HTML}{586e75}

\begin{document}

\title{The Formal Semantics of Rascal Light}
\author{Ahmad Salim Al-Sibahi}
\maketitle

\section{Introduction}
\label{sec:introduction}
Rascal\,\citep{DBLP:conf/scam/KlintSV09,Klint2011} is a high-level transformation language that aims to simplify software
language engineering tasks like defining program syntax, analyzing and transforming programs,
and performing code generation. The language
provides several features including built-in collections
(lists, sets, maps), algebraic data-types, powerful pattern matching operations
with backtracking, and
high-level traversals supporting multiple
strategies.

Interaction between different language features can be difficult to comprehend,
since most features are semantically rich.
My goal is to provide a well-defined formal semantics for a large subset of Rascal called
Rascal Light---in the spirit of (core) Caml Light\,\citep{WEB:Leroy97},
Clight\,\citep{DBLP:journals/jar/BlazyL09}, and Middleweight Java\,\citep{PittsAM:impccj}---suitable
for developing formal techniques, e.g., type systems and static analyses.

\subsection*{Scope}
Rascal Light aims to model a realistic set of high-level transformation
language features, by capturing a large subset of the Rascal operational
language. Rascal Light targets being practically usable, i.e., it should be possible to translate many realistic pure Rascal programs to this subset by an
expert programmer without losing the high level of abstraction.
The following Rascal features are captured in Rascal Light:%
\begin{itemize}
\item Fundamental definitions including algebraic data-types, functions with type parameters, and global variables.
\item Basic expressions, including variable assignment, definition and access,
  exceptions, collections (including sets, lists and maps), $\keyword{if}$-expressions,
  $\keyword{switch}$-statements, $\keyword{for}$-loops and $\keyword{while}$-loops including control flow operators.
\item Powerful pattern matching operations including type patterns, non-linear
  pattern matching, (sub)collection patterns, and descendant pattern
  matching.
\item Backtracking using the $\keyword{fail}$ operator, including roll-back of
  state.
\item Traversals using generic $\keyword{visit}$-expressions, supporting the
  different kinds of available strategies: $\keyword{bottom-up}$($\keyword{-break}$),
  $\keyword{top-down}$($\keyword{-break}$), $\keyword{innermost}$, and $\keyword{outermost}$.
\item Fixed-point iteration using the $\keyword{solve}$-loop.
\end{itemize}

The following Rascal features are considered out of scope:%
\begin{itemize}
\item Concrete syntax declaration and literals, string interpolation, regular
  expression matching, date literals and path literals
\item The standard library including Input/Output and the Java foreign function interface
\item The module system, include modular extension of language elements such as datatypes and functions
\item Advanced type system features, like parametric polymorphism, the numerical hierarchy and type inference
\end{itemize}
In Rascal, Boolean expressions can contain pattern matching subexpressions and
backtracking is affected by the various Boolean operators (conjunction,
disjunction, implication). In Rascal Light, backtracking is more restricted,
which I believe heavily simplifies the semantics while losing little
expressiveness in practice; most programs could be rewritten to support the
required backtracking at the cost of increased verbosity.

\subsection*{Method}
I have used the language
documentation,\footnote{\url{http://tutor.rascal-mpl.org/Rascal/Rascal.html}} and the open source implementation of
Rascal,\footnote{\url{https://github.com/usethesource/rascal}} to derive the
formalism.
The syntax is primarily based on the \textmu{}Rascal syntax description\,\citep{WEBCITE:muRascal17}, but altered
to focus on the high-level features.
The semantics is based on the description of individual language features in the
documentation. In case of under-specification or ambiguity, I used small test
programs to check what the expected behavior is.
I thank the Rascal developers for our personal correspondence which further
clarified ambiguity and limitations of the semantics compared to Rascal.

To discover possible issues, the semantics has been implemented as a prototype
and tested against a series of Rascal programs.
To strengthen the correctness claims, I have proven a series of
theorems of interest in \Cref{sec:theorems}, since as
\cite{DBLP:books/daglib/0069232} states:

\begin{quote}
  \emph{The robustness of the semantics depends upon theorems}
  \begin{flushright}
    --- The Definition of Standard ML
  \end{flushright}
\end{quote}

\subsection*{Notation}
A sequence of $e$s is represented by either $\esequence{e}$ or more
explicitly $e_1, \dots, e_{\mathrm{n}}$, the empty sequence is denoted by $\varepsilon$ and the
concatenation of sequences $\esequence{e_1}$ and $\esequence{e_2}$ is
denoted $\esequence{e_1}, \esequence{e_2}$. We overload notation in an intuitive
manner for operations on sequences, so given a sequence $\esequence{v}$, $v_i$
denotes the $i$th element in the sequence, and $\esequence{v : t}$ denotes the
sequence $v_1 : t_1, \dots, v_{\mathrm{n}} : t_{\mathrm{n}}$.

\subsection*{Overview}
The abstract syntax of Rascal Light is presented in \Cref{sec:syntax},
and the dynamic (operational) semantics is presented in \Cref{sec:semantics}.
Finally, \Cref{sec:theorems} presents relevant theorems regarding the soundness of the
semantics.

\section{Abstract Syntax}
\label{sec:syntax}
Rascal Light programs are organized into modules that consist of definitions of
global variables, functions and algebraic data types. I will in the rest of this
report assume that modules
are well-formed: top-level definitions and function parameters have unique
names, all function calls,
constructors and datatypes used have corresponding definitions, and all
variables are well-scoped with no shadowing. To maintain a clean presentation, I will not write the definition
environment explicitly, but mention the necessary definitions as premises when required.

Rascal Light has three kinds of definitions: global variables, functions and
algebraic data-types. Global variables are typed and are initialized with a
target expression at module loading time. Functions have a return type, a unique
name, a list of typed uniquely named parameters, and have an expression as a
body. Algebraic data-types have unique names and declare a set of possible alternative constructors, each taking a list of typed fields as arguments.

{\small\begin{alignat*}{3}
  d &\Coloneqq \keyword{global} \; t \; y = e & (\textit{Global variables}) \\
  &\quad \mid \keyword{fun} \; t \; f(\esequence{t \; x}) = e & (\textit{Function Definition}) \\
  &\quad \mid \keyword{data} \; \mathit{at} = k_1(\esequence{t \; x}) | \dots |
    k_{\mathrm{n}}(\esequence{t \; x})  &\quad (\textit{Algebraic Data Type Definition})
\end{alignat*}}

\noindent{}The operational part of Rascal consists of syntactically distinct categories of
statements and expressions; I have chosen to collapse the two
categories in Rascal Light since statements also
have result values that can be used inside expressions.
Most constructs are standard imperative ones, such as blocks,
assignment, branching and loops (including control operations);
I have chosen to use  $\keyword{local}$-$\keyword{in}$ for representing
blocks in Rascal Light instead of curly braces like in Rascal, to
distinguish them from set literal expressions; blocks contain locally declared variables and a sequence
of expressions to evaluate.

Notable Rascal-specific constructs include a generalized
version of the $\keyword{switch}$-expression that supports rich pattern matching over both basic values and algebraic data types, the $\keyword{visit}$-expression which
allows traversing data types using various
strategies (\Cref{exmp:expsimpl}), and the $\keyword{solve}$-expression which continuously evaluates
its body expression until target variables reach a fixed-point in assigned
values (\Cref{exmp:solve}). The $\keyword{fail}$ control operator allows backtracking
inside $\keyword{switch}$ and $\keyword{visit}$ statements to try other possible
matches (\Cref{exmp:fail}).

\vspace{-1.5em}
{\small\allowdisplaybreaks\begin{alignat*}{3}
  e &\Coloneqq vb & (\textit{Basic Values }, vb \in \left\{1,\mathtt{``foo"}, 3.14, \ldots \right\}) \\
  &\quad \mid x & (\textit{Variables}, x, y \in \namedset{Var}) \\
  &\quad \mid \ominus \; e & (\textit{Unary Operations}, \ominus \in \left\{-, \ldots \right\}) \\
  &\quad \mid e_1 \oplus e_2 & (\textit{Binary Operations}, \oplus \in \left\{+, -, \times, \ldots \right\}) \\
  &\quad \mid k (\esequence{e}) & \quad (\textit{Constructor Applications}, k \in \namedset{Constructor}) \\
  &\quad \mid [\esequence{e}] & (\textit{List Expressions}) \\
  &\quad \mid \{ \esequence{e} \} & (\textit{Set Expressions}) \\
  &\quad \mid ( \esequence{e : e'} ) & (\textit{Map Expressions}) \\
  &\quad \mid e_1 [e_2] & (\textit{Map Lookup}) \\
  &\quad \mid e_1 [e_2 = e_3] & (\textit{Map Update}) \\
  &\quad \mid f ( \esequence{e} ) & (\textit{Function Calls}, f \in \namedset{Function}) \\
  &\quad \mid \keyword{return} \; e & (\textit{Return Expressions}) \\
  &\quad \mid x = e & (\textit{Variable Update Assignments}) \\
  &\quad \mid \keyword{if} \; e \; \keyword{then} \; e_1 \; \keyword{else} \; e_2 & (\textit{Conditionals})  \\
  &\quad \mid \keyword{switch} \; e \; \keyword{do} \; \esequence{\mathit{cs}} & (\textit{Case Matching}) \\
  &\quad \mid \mathit{st} \; \keyword{visit} \; e \; \keyword{do} \; \esequence{\mathit{cs}}  & (\textit{Deep Traversal})  \\
  &\quad \mid \keyword{break} \mid \keyword{continue} \mid \keyword{fail} & (\textit{Control Operations}) \\
  &\quad \mid \keyword{local} \; \esequence{t \; x} \; \keyword{in} \;
  \esequence{e} \; \keyword{end} & (\textit{Blocks}) \\
  &\quad \mid \keyword{for} \; g \; e & (\textit{Iteration}) \\
  &\quad \mid \keyword{while} \; e \; e & (\textit{While Expressions}) \\
  &\quad \mid \keyword{solve} \; \esequence{x} \; e & (\textit{Fixedpoint Expressions}) \\
  &\quad \mid \keyword{throw} \; e & (\textit{Exception Throwing}) \\
  &\quad \mid \keyword{try} \; e_1 \; \keyword{catch} \; x => e_2 & (\textit{Exception Catching}) \\
  &\quad \mid \keyword{try} \; e_1 \; \keyword{finally} \; e_2 &
  (\textit{Exception Finalization}) \\
  \mathit{cs} &\Coloneqq \keyword{case} \; p => e &\quad (\textit{Case})
\end{alignat*}}
\begin{example}[Expression simplifier]\label{exmp:expsimpl} An expression simplifier can use visit to
  completely simplify all sub-expressions no matter where they are in the input expression:
\begin{lstlisting}[language=Rascal]
  data Expr = intlit(int v) | plus(Expr lop, Expr rop) | ...;

  Expr simplify(Expr e) = 
    bottom-up visit(e) {
      case plus(intlit(0), y) => y
      case plus(x, intlit(0)) => x
    };
\end{lstlisting}
\end{example}

\begin{example}[Lattice fixed-point]\label{exmp:solve} The Kleene fixed-point of a continuous
  function in a lattice for a domain $Val$ with functions for the bottom element and
  least-upper bound can be computed using solve:
\begin{lstlisting}[language=Rascal]
  Val fix(Fun f) = {
     Val v = bottom();
     solve(v) {
        v = lub(v, apply(f, v))
     };
  }
\end{lstlisting}
\end{example}

\begin{example}[Knapsack problem]\label{exmp:fail} The knapsack problem concerns
  finding a subset of items with greatest value under a specified maximum
  weight. The function below uses backtracking to (slowly) find the optimal solution.
  \begin{lstlisting}[language=Rascal]
    set[Item] slowknapsack(set[Item] items, int maxWeight) = {
      set[Item] res = {};
      solve(res) {
        switch(items) {
          case {*xs, *ys}:
          if (sumweights(xs) > maxWeight
              || sumvalues(xs) < sumvalues(res)) fail;
          else res = xs;
        };
        res
      };
    }
  \end{lstlisting}
\end{example}

\noindent{}There are two kinds of \textit{generator expressions}:
\textit{enumerating assignment} where the range variable is
assigned to each element in the result of a collection-producing expression, and
\textit{matching assignment} that produces all possible assignments that match
target patterns (defined later).
{\small\begin{alignat*}{3}
  g &\Coloneqq x <- e &\quad (\textit{Enumerating Assignment})  \\
  &\quad \mid  p \coloneqq e &\quad (\textit{Matching Assignment})
\end{alignat*}}

\noindent{}The $\keyword{visit}$-expression supports various strategies that determine the
order a particular value is traversed w.r.t.\ its contained values. The
$\keyword{top-down}$ strategy
traverses the value itself before contained values, and conversely
$\keyword{bottom-up}$ traverses contained values before itself. The $\keyword{break}$
versions stop the traversal at first successful match, and the
$\keyword{outermost}$ and $\keyword{innermost}$ respectively apply the $\keyword{top-down}$
and $\keyword{bottom-up}$ until a fixed-point is reached.
{\small\begin{alignat*}{3}
  \mathit{st} &\Coloneqq \keyword{top-down} & (\textit{Preorder Traversal}) \\
  &\quad \mid \keyword{bottom-up} & (\textit{Postorder Traversal}) \\
  &\quad \mid \keyword{top-down-break} & (\textit{First-match Preorder Traversal}) \\
  &\quad \mid \keyword{bottom-up-break} &\quad (\textit{First-match Postorder Traversal}) \\
  &\quad \mid \keyword{outermost} & (\textit{Fixedpoint Preorder Traversal}) \\
  &\quad \mid \keyword{innermost} & (\textit{Fixedpoint Postorder Traversal}) 
\end{alignat*}}

\noindent{}Like Rascal, Rascal Light has a rich pattern language that not only includes matching on
basic values and constructors, but also powerful matching inside collections and
descendant patterns that allow matching arbitrarily deeply contained values.

{\small\begin{alignat*}{3}
  p &\Coloneqq vb &(\textit{Basic Value Patterns}) \\
  &\quad \mid x       &(\textit{Variable Patterns}) \\
  &\quad \mid k(\esequence{p}) &\quad (\textit{Deconstructor Patterns}) \\
  &\quad \mid t \; x: p  &(\textit{Typed Labelled Patterns}) \\
  &\quad \mid [\esequence{{{\star}p}}] &(\textit{List Patterns}) \\
  &\quad \mid \{\esequence{{{\star}p}}\} &(\textit{Set Patterns}) \\
  &\quad \mid {!p} &(\textit{Negation pattern}) \\
  &\quad \mid /p & (\textit{Descendant Patterns})
\end{alignat*}}

\noindent{}Patterns inside collections can be either ordinary patterns or 
\emph{star patterns} that match a subcollection with arbitrary number of elements.
{\small\begin{alignat*}{3}
{{\star}p} &\Coloneqq p & (\textit{Ordinary Pattern}) \\
  &\quad \mid {\star}x &\quad (\textit{Star Pattern})
\end{alignat*}}

\noindent{}Rascal Light programs expressions evaluate to values, which either are basic values,
constructor values, collections or the undefined value ($\blacksquare$).
{\small\begin{alignat*}{3}
  v &\Coloneqq  vb &(\textit{Basic Values}) \\
  &\quad \mid k (\esequence{v}) & \quad (\textit{Constructor Values}, k \in \namedset{Constructor}) \\
  &\quad \mid [\esequence{v}] & (\textit{List Values}) \\
  &\quad \mid \{ \esequence{v} \} & (\textit{Set Values}) \\
  &\quad \mid ( \esequence{v : v'} ) & (\textit{Map Values}) \\
  &\quad \mid \blacksquare & (\textit{Undefined Value})
\end{alignat*}}

\section{Semantics}
\label{sec:semantics}
I present a formal development of the dynamic aspects of Rascal Light, using a
natural semantics specification. Natural (big-step) semantics\,\citep{DBLP:conf/stacs/Kahn87} is particularly suitable for
Rascal, because it closely mimics semantics of an interpreter and the high-level features---exceptions, backtracking,
traversals---introduce a rich control-flow that depends not only on the
structure of the program but also on the provided input.
There is no concurrency or interleaving in Rascal, and so there is no need for a
more fine-grained operational semantics like (small-step) structural operational semantics (SOS)\,\citep{DBLP:journals/jlp/Plotkin04a}.

\subsection*{Value Typing}
\label{sub:syntax}
Rascal (Light) is strongly typed and so all values are typed. The types are fairly straight-forward in the sense that
most values have a canonical type and there is a subtyping hierarchy (explained
shortly) with a
bottom type $\rascaltype{void}$ and a top type $\rascaltype{value}$.
{\small\begin{alignat*}{3}
  t &\Coloneqq  \mathit{tb} &(\textit{Base Types, $\mathit{tb} \in \left\{ \namedset{Int}, \namedset{Rational}, \namedset{String}, \ldots \right\}$}) \\
  &\quad \mid  \mathit{at} &\quad (\textit{Algebraic Data Types, $\mathit{at} \in \namedset{DataType}$}) \\
  &\quad \mid  \rascaltype{set}\ttuple{t} & (\textit{Sets}) \\
  &\quad \mid  \rascaltype{list}\ttuple{t} & (\textit{Lists}) \\
  &\quad \mid  \rascaltype{map}\ttuple{t_1,t_2} & (\textit{Maps}) \\
  &\quad \mid  \rascaltype{void} & (\textit{Bottom Type}) \\
  &\quad \mid  \rascaltype{value} & (\textit{Top Type})
\end{alignat*}}

\noindent{}I provide a typing judgment for values
of form \highlightbox{$\typing{v}{t}$}, which states that value $v$ has type $t$.
Basic values are typed by their defining basic type, the undefined value
($\blacksquare$) has the bottom type $\rascaltype{void}$, and constructor values are
typed by their corresponding data type definition assuming the contained values
are well-typed. The type of collections is determined by the contained values,
and so a least upper bound operator is defined in types---following the subtype
ordering---which is used to infer a precise type
for the value parameters.

{\small
\[
  \inference[\textsc{T-Basic}]{vb \in \sem{\mathit{tb}}}{\typing{vb}{\mathit{tb}}}
  \quad
  \inference[\textsc{T-Bot}]{}{\typing{\blacksquare}{\rascaltype{void}}}
\]

\[
  \inference[\textsc{T-Constructor}]{\keyword{data} \; \mathit{at} = \dots ~|~
    k(\esequence{t}) ~|~ \dots &
    \esequence{\typing{v}{t'}} & \esequence{\subtyping{t'}{t}}}{\typing{k(v_1, \dots, v_{\mathrm{n}})}{\mathit{at}}}
\]

\[
  \inference[\textsc{T-Set}]{\esequence{\typing{v}{t}}}{\typing{\{\esequence{v}\}}{\rascaltype{set}\ttuple{\bigsqcup\esequence{t}}}}
  \quad
  \inference[\textsc{T-List}]{\esequence{\typing{v}{t}}}{\typing{[\esequence{v}]}{\rascaltype{list}\ttuple{\bigsqcup\esequence{t}}}}
  \quad
  \inference[\textsc{T-Map}]{\esequence{\typing{v}{t}} & \esequence{\typing{v'}{t'}}}{\typing{(\esequence{v : v'})}{\rascaltype{map}\ttuple{\bigsqcup\esequence{t},\bigsqcup\esequence{t'}}}}
\]
}

\noindent{}The subtyping relation has form
\highlightbox{$\subtyping{t}{t'}$}, stating t is a subtype of $t'$. We let
\highlightbox{$\notsubtyping{t}{t'}$} denote the negated form where none of the
given cases below matches.
subtyping is reflexive, so every type is a sub-type of itself;
$\rascaltype{void}$ and $\rascaltype{value}$ act as bottom type and top type respectively.
{\small
\[
  \inference[\textsc{ST-Refl}]{}{\subtyping{t}{t}}
  \quad
  \inference[\textsc{ST-Void}]{}{\subtyping{\rascaltype{void}}{t}}
  \quad
  \inference[\textsc{ST-Value}]{}{\subtyping{t}{\rascaltype{value}}}
\]}

\noindent{}Collections are covariant in their type parameters, which is safe since all
values are immutable.

{\small
  \[\def\arraystretch{3}
    \begin{tabular}{cc}
  \inference[\textsc{ST-List}]{\subtyping{t}{t'}}{\subtyping{\rascaltype{list}\ttuple{t}}{\rascaltype{list}\ttuple{t'}}}
      &
  \inference[\textsc{ST-Set}]{\subtyping{t}{t'}}{\subtyping{\rascaltype{set}\ttuple{t}}{\rascaltype{set}\ttuple{t'}}}
  \\
         \multicolumn{2}{c}{\inference[\textsc{ST-Map}]{\subtyping{t_{\mathrm{key}}}{t_{\mathrm{key}}'} &
    \subtyping{t_{\mathrm{val}}}{t_{\mathrm{val}}'}}{\subtyping{\mathrm{map}\ttuple{t_{\mathrm{key}},
      t_{\mathrm{val}}}}{\rascaltype{map}\ttuple{t_{\mathrm{key}}', t_{\mathrm{val}}'}}}}
\end{tabular}\]}

\noindent{}The least upper bound on types is defined as follows:
{\small\begin{gather*}
  \bigsqcup \varepsilon = \rascaltype{void} \qquad \bigsqcup t, \esequence{t'} = t
  \sqcup \bigsqcup\esequence{t'} \\
  t_1 \sqcup t_2 =
  \begin{cases}
    t_1 & \textbf{if } t_2 = \rascaltype{void} \vee t_1 = t_2 \\
    t_2 & \textbf{if } t_1 = \rascaltype{void} \\
    \rascaltype{list}\ttuple{t_1' \sqcup t_2'} & \textbf{if } t_1 =
    \rascaltype{list}\ttuple{t_1'} \wedge  t_2 = \rascaltype{list}\ttuple{t_2'} \\
        \rascaltype{set}\ttuple{t_1' \sqcup t_2'} & \textbf{if } t_1 =
    \rascaltype{set}\ttuple{t_1'} \wedge  t_2 = \rascaltype{set}\ttuple{t_2'} \\
      \rascaltype{map}\ttuple{t_1' \sqcup t_2', t_1'' \sqcup t_2''} & \textbf{if } t_1 =
    \rascaltype{map}\ttuple{t_1',t_1''} \wedge  t_2 = \rascaltype{map}\ttuple{t_2',t_2''} \\
    \rascaltype{value} & \textbf{otherwise}
  \end{cases}
\end{gather*}}
\subsection*{Expression Evaluation}
The main judgment for Rascal Light expressions has the form
\highlightbox{$\evalexpr{e}{\sigma}{\mathit{vres}}{\sigma'}$}, where the expression $e$ is
evaluated in an initial store $\sigma \in \namedset{Var} -> \namedset{Val}$---mapping variables to values---returning a result
$\mathit{vres}$ and updated store $\sigma'$ as a side-effect.
The result $\mathit{vres}$ is either an ordinary $\keyword{success} \; v$ signifying
successful execution or an exceptional result $\mathit{exres}$.
{\small\begin{alignat*}{3}
  \mathit{vres} &\Coloneqq \keyword{success} \; v \mid \mathit{exres}
\end{alignat*}}

\noindent{}An exceptional result is either a control operation ($\keyword{break}$,
$\keyword{continue}$, $\keyword{fail}$), an $\keyword{error}$ that happened
during execution, a thrown exception $\keyword{throw} \; v$ or an early return value
$\keyword{return} \; v$. The difference between $\keyword{success} \; v$ and
$\keyword{return} \; v$ is that the latter should propagate directly from
sub-expressions to the surrounding function call boundary, while the value
in the former can be used further in intermediate computations (\Cref{exmp:earlyreturn}).
{\small\begin{alignat*}{3}
  \mathit{exres} &\Coloneqq\keyword{return} \; v \mid
  \keyword{throw} \; v \mid \keyword{break} \mid \keyword{continue} \mid
  \keyword{fail} \mid \keyword{error}
\end{alignat*}}

\begin{example}[Early Return]\label{exmp:earlyreturn} The function that calculates producs, uses the
  early return functionality to short-circuit the rest of the calculation when
  a factor in the provided list is zero. If this branch is hit during execution,
  evaluating the expression produces $\keyword{return} \; 0$ as result, which is
  then propagated directly to function call boundary, skipping the rest of the
  loop and sequenced expressions.
  \begin{lstlisting}[language=Rascal]
    int prod(list[int] xs) {
      int res = 1;
      for (x <- xs) {
        if (x == 0) return 0;
        else res *= x;
      };
      res
    }
  \end{lstlisting}
\end{example}

\noindent{}Basic values evaluate to their semantic equivalent result values without any side effect (\textsc{E-Val}).
{\small\[
  \inference[\textsc{E-Val}]{}{\evalexpr{\mathit{vb}}{\sigma}{\keyword{success} \; \mathit{vb}}{\sigma}}
\]}

\noindent{}A variable evaluates to the value it is assigned to in the store if available
(\textsc{E-Var-Sucs}), and otherwise result in an error (\textsc{E-Var-Err}).%
{\small\[
  \inference[\textsc{E-Var-Sucs}]{x \in \textrm{dom }
    \sigma}{\evalexpr{x}{\sigma}{\keyword{success} \; \sigma(x)}{\sigma}}
  \quad
  \inference[\textsc{E-Var-Err}]{x \notin \textrm{dom }
    \sigma}{\evalexpr{x}{\sigma}{\keyword{error}}{\sigma}}
\]}

\noindent{}Unary expressions evaluate their operands, applying possible
side-effects; if successful the corresponding semantic unary operator $\sem{\ominus}$ is applied on the result value
(\textsc{E-Un-Sucs}), and otherwise it propagates the exceptional result
(\textsc{E-Un-Exc}).
{\small
\[
  \inference[\textsc{E-Un-Sucs}]{\evalexpr{e}{\sigma}{\keyword{success} \; v}{\sigma'}}{\evalexpr{\ominus \; e}{\sigma}{\sem{\ominus}(v)}{\sigma'}}
  \quad
  \inference[\textsc{E-Un-Exc}]{\evalexpr{e}{\sigma}{\mathit{exres}}{\sigma'}}{\evalexpr{\ominus \; e}{\sigma}{\mathit{exres}}{\sigma'}}
\]}
\begin{example}[Unary operator semantics]
  \textrm{$\sem{-}(3)$ will evaluate to $\keyword{success} \; {-3}$, while $\sem{-}(\{\})$
  will evaluate to $\keyword{error}$}.
\end{example}

\noindent{}Evaluating binary expressions is similar to unary
expressions, requiring both operands to evaluate successfully
(\textsc{E-Bin-Sucs}) to apply the corresponding semantic
binary operator $\sem{\oplus}$; otherwise the exceptional results of
the operands are propagated in order from left (\textsc{E-Bin-Exc1}) to right (\textsc{E-Bin-Exc2}).%
{\small
\[
  \inference[\textsc{E-Bin-Sucs}]{\evalexpr{e_1}{\sigma}{\keyword{success} \; v_1}{\sigma''} &
    \evalexpr{e_2}{\sigma''}{\keyword{success} \;
      v_2}{\sigma'}}{\evalexpr{e_1 \oplus e_2}{\sigma}{\sem{\oplus}(v_1,v_2)}{\sigma'}}
\]

\[
  \inference[\textsc{E-Bin-Exc1}]{\evalexpr{e_1}{\sigma}{\mathit{exres}_1}{\sigma''}}{\evalexpr{e_1 \oplus e_2}{\sigma}{\mathit{exres}_1}{\sigma'}}
\]

\[
  \inference[\textsc{E-Bin-Exc2}]{\evalexpr{e_1}{\sigma}{\keyword{success} \; v_1}{\sigma''} &
    \evalexpr{e_2}{\sigma''}{\mathit{exres}_2}{\sigma'}}{\evalexpr{e_1 \oplus e_2}{\sigma}{\mathit{exres}_2}{\sigma'}}
\]}

\noindent{}Constructor expressions evaluate their arguments first, and if they all
successfully evaluate to values, then check whether the types of values match
those expected in the declaration. If the result values have the right types and
are not $\blacksquare$, a
constructor value is constructed (\textsc{E-Cons-Sucs}), and otherwise a (type)
error is produced (\textsc{E-Cons-Err}). In case any of the arguments
has an exceptional result the evaluation of the rest of the arguments halts and
the exceptional result is propagated (\textsc{E-Cons-Exc}).%
{\small\[
  \inference[\textsc{E-Cons-Sucs}]{\keyword{data} \; \mathit{at} = \dots ~|~
    k(\esequence{t}) ~|~ \dots  &
    \evalexprstar{\esequence{e}}{\sigma}{\keyword{success} \;
      \esequence{v}}{\sigma'} \\
    \esequence{\typing{v}{t'}} & \esequence{v \neq \blacksquare} & \esequence{\subtyping{t'}{t}}}{\evalexpr{k(\esequence{e})}{\sigma}{\keyword{success}
      \; k(\esequence{v})}{\sigma'}}
\]

\[
  \inference[\textsc{E-Cons-Err}]{\keyword{data} \; \mathit{at} = \dots ~|~
    k(\esequence{t}) ~|~ \dots  &
    \evalexprstar{\esequence{e}}{\sigma}{\keyword{success} \;
      \esequence{v}}{\sigma'} \\
    \esequence{\typing{v}{t'}} & \exists i. v_i = \blacksquare \vee \notsubtyping{t'_i}{t_i} }{\evalexpr{k(\esequence{e})}{\sigma}{\keyword{error}}{\sigma'}}
\]

\[
  \inference[\textsc{E-Cons-Exc}]{\evalexprstar{\esequence{e}}{\sigma}{\mathit{exres}}{\sigma'}}{\evalexpr{k(\esequence{e})}{\sigma}{\mathit{exres}}{\sigma'}}
\]}

\noindent{}Evaluating list expressions also requires evaluating all subexpressions to a
series of values; because of sequencing and necessity of early propagation of exceptional
results, evaluation of series of subexpressions is done using a mutually
recursive sequence evaluation judgment (see page~\pageref{par:exprseq}). If the evaluation is successful then a list value is
constructed (\textsc{E-List-Sucs}), unless any value is undefined ($\blacksquare$) in which
case we produce an error (\textsc{E-List-Err}), and otherwise the exceptional result is
propagated (\textsc{E-List-Exc}).

{\small\[
  \inference[\textsc{E-List-Sucs}]{\evalexprstar{\esequence{e}}{\sigma}{\keyword{success} \;
      \esequence{v}}{\sigma'} & \esequence{v \neq \blacksquare}}{\evalexpr{[\esequence{e}]}{\sigma}{\keyword{success}
      \; [\esequence{v}]}{\sigma'}}
\]
\[
  \inference[\textsc{E-List-Err}]{\evalexprstar{\esequence{e}}{\sigma}{\keyword{success} \;
      \esequence{v}}{\sigma'} & \exists i. v_i = \blacksquare}{\evalexpr{[\esequence{e}]}{\sigma}{\keyword{error}}{\sigma'}}
  \quad
  \inference[\textsc{E-List-Exc}]{\evalexprstar{\esequence{e}}{\sigma}{\mathit{exres}}{\sigma'}}{\evalexpr{[\esequence{e}]}{\sigma}{\mathit{exres}}{\sigma'}}
\]}

\noindent{}Set expression evaluation mirror the one for lists, except that values are
constructed using a set constructor (\textsc{E-Set-Sucs}), which may reorder
values and ensures that there are no duplicates. If any contained value was undefined
($\blacksquare$) then an error is produced instead (\textsc{E-Set-Err}), and exceptional results are
propagated (\textsc{E-Set-Exc}).

{\small\[
  \inference[\textsc{E-Set-Sucs}]{\evalexprstar{\esequence{e}}{\sigma}{\keyword{success} \;
      \esequence{v}}{\sigma'} & \esequence{v \neq \blacksquare}}{\evalexpr{\{\esequence{e}\}}{\sigma}{\keyword{success}
      \; \{\esequence{v}\}}{\sigma'}}
\]
\[
  \inference[\textsc{E-Set-Err}]{\evalexprstar{\esequence{e}}{\sigma}{\keyword{success} \;
      \esequence{v}}{\sigma'} & \exists i. v_i = \blacksquare}{\evalexpr{\{\esequence{e}\}}{\sigma}{\keyword{error}}{\sigma'}}
\quad
  \inference[\textsc{E-Set-Exc}]{\evalexprstar{\esequence{e}}{\sigma}{\mathit{exres}}{\sigma'}}{\evalexpr{\{\esequence{e}\}}{\sigma}{\mathit{exres}}{\sigma'}}
\]}

\noindent{}Map expressions evaluate their keys and values in the declaration sequence, and
if successful construct a map (\textsc{E-Map-Sucs}). Similarly to other
collection expressions, errors are
produced if any value is undefined (\textsc{E-Map-Err}) and exceptional results are propagated (\textsc{E-Map-Exc}).%
\[
  \inference[\textsc{E-Map-Sucs}]{\evalexprstar{\esequence{e, e'}}{\sigma}{\keyword{success} \;
      \esequence{v, v'}}{\sigma'} & \esequence{v \neq \blacksquare}
    & \esequence{v' \neq \blacksquare}}{\evalexpr{(\esequence{e : e'})}{\sigma}{\keyword{success}
      \; (\esequence{v : v'})}{\sigma'}}
\]
\[
  \inference[\textsc{E-Map-Err}]{\evalexprstar{\esequence{e,e'}}{\sigma}{\keyword{success} \;
      \esequence{v, v'}}{\sigma'} & \exists. v_i = \blacksquare \vee v_i' = \blacksquare}{\evalexpr{(\esequence{e : e'})}{\sigma}{\keyword{error}}{\sigma'}}
\]
\[
  \inference[\textsc{E-Map-Exc}]{\evalexprstar{\esequence{e, e'}}{\sigma}{\mathit{exres}}{\sigma'}}{\evalexpr{(\esequence{e : e'})}{\sigma}{\mathit{exres}}{\sigma'}}
\]

\noindent{}Lookup expressions require evaluating the outer expression to a map---otherwise
producing an error (\textsc{E-Lookup-Err})---and the index expression to a value. If the
index is an existing key then the corresponding value is produced as result
(\textsc{E-Lookup-Sucs}) and otherwise the $\textrm{nokey}$ exception is thrown
(\textsc{E-Lookup-NoKey}); here, I assume that {\small$\keyword{data} \; \textrm{NoKey} =
    \textrm{nokey}(\textrm{value} \; \textrm{key})$} is a built-in data-type definition. Exceptional results are propagated from sub-terms
(\textsc{E-Lookup-Exc1}, \textsc{E-Lookup-Exc2}).%
{\small\[
  \inference[\textsc{E-Lookup-Sucs}]{\evalexpr{e_1}{\sigma}{\keyword{success}
      \; (\dots, v : v', \dots)}{\sigma''} \\
    \evalexpr{e_2}{\sigma''}{\keyword{success} \; v}{\sigma'}}{\evalexpr{e_1[e_2]}{\sigma}{\keyword{success}
    \; v'}{\sigma'}}
\]

\[
  \inference[\textsc{E-Lookup-NoKey}]{\evalexpr{e_1}{\sigma}{\keyword{success}
      \; (\esequence{v : v'})}{\sigma''} \\
    \evalexpr{e_2}{\sigma''}{\keyword{success} \; v''}{\sigma'} & \forall i. v'' \neq v_i}{\evalexpr{e_1[e_2]}{\sigma}{\keyword{throw}
    \; \auxiliary{nokey}(v'')}{\sigma'}}
\]

\[
  \inference[\textsc{E-Lookup-Err}]{\evalexpr{e_1}{\sigma}{\keyword{success} \;
      v}{\sigma'} & v \neq (\esequence{v' : v''})}{\evalexpr{e_1[e_2]}{\sigma}{\keyword{error}}{\sigma'}}
\]

\[
  \inference[\textsc{E-Lookup-Exc1}]{\evalexpr{e_1}{\sigma}{\mathit{exres}}{\sigma'}}{\evalexpr{e_1[e_2]}{\sigma}{\mathit{exres}}{\sigma'}}
\]

\[
  \inference[\textsc{E-Lookup-Exc2}]{\evalexpr{e_1}{\sigma}{\keyword{success} \;
      (\esequence{v : v'})}{\sigma''} & \evalexpr{e_2}{\sigma''}{\mathit{exres}}{\sigma'}}{\evalexpr{e_1[e_2]}{\sigma}{\mathit{exres}}{\sigma'}}
\]}

\noindent{}Map update expressions also require the outer expression to evaluate to a
map---otherwise producing an error (\textsc{E-Update-Err1})---and the index and target
expressions to evaluate to values.
On succesful evaluation of both index and target value the map is updated,
overriding the old value of the corresponding index if necessary
(\textsc{E-Update-Sucs}) unless the index or target value is equal to $\blacksquare$ in
which case an error is produced (\textsc{E-Update-Err2}).
Finally, exceptional results are propagated left-to-right if necessary
(\textsc{E-Update-Exc1}, \textsc{E-Update-Exc2}, \textsc{E-Update-Exc3}).%

{\small\[
  \inference[\textsc{E-Update-Sucs}]{\evalexpr{e_1}{\sigma}{\keyword{success}
      \; (\esequence{v : v'})}{\sigma'''} &
    \evalexpr{e_2}{\sigma'''}{\keyword{success} \; v''}{\sigma''} \\
    \evalexpr{e_3}{\sigma''}{\keyword{success} \; v'''}{\sigma'} & v'' \neq
    \blacksquare & v''' \neq \blacksquare
  }{\evalexpr{e_1[e_2 = e_3]}{\sigma}{\keyword{success}
    \; (\esequence{v : v'}, v'' : v''')}{\sigma'}}
\]

\[
  \inference[\textsc{E-Update-Err1}]{\evalexpr{e_1}{\sigma}{\keyword{success}
      \; v}{\sigma'} & v \neq (\esequence{v' : v''})}{\evalexpr{e_1[e_2 = e_3]}{\sigma}{\keyword{error}}{\sigma'}}
\]

\[
  \inference[\textsc{E-Update-Err2}]{\evalexpr{e_1}{\sigma}{\keyword{success}
      \; (\esequence{v : v'})}{\sigma'''} &
    \evalexpr{e_2}{\sigma'''}{\keyword{success} \; v''}{\sigma''} \\
    \evalexpr{e_3}{\sigma''}{\keyword{success} \; v'''}{\sigma'} & v'' =
    \blacksquare \vee v''' = \blacksquare
  }{\evalexpr{e_1[e_2 = e_3]}{\sigma}{\keyword{error}}{\sigma'}}
\]

\[
  \inference[\textsc{E-Update-Exc1}]{\evalexpr{e_1}{\sigma}{\mathit{exres}}{\sigma'}}{\evalexpr{e_1[e_2
      = e_3]}{\sigma}{\mathit{exres}}{\sigma'}}
\]

\[
  \inference[\textsc{E-Update-Exc2}]{\evalexpr{e_1}{\sigma}{\keyword{success} \;
      (\esequence{v : v'})}{\sigma''} &
    \evalexpr{e_2}{\sigma''}{\mathit{exres}}{\sigma'}}{\evalexpr{e_1[e_2 = e_3]}{\sigma}{\mathit{exres}}{\sigma'}}
\]

\[
  \inference[\textsc{E-Update-Exc3}]{\evalexpr{e_1}{\sigma}{\keyword{success} \;
      (\esequence{v : v'})}{\sigma''} \\
    \evalexpr{e_2}{\sigma'''}{\keyword{success} \; v''}{\sigma''} &
    \evalexpr{e_3}{\sigma''}{\mathit{exres}}{\sigma'}}{\evalexpr{e_1[e_2 = e_3]}{\sigma}{\mathit{exres}}{\sigma'}}
\]}

\noindent{}Evaluation of function calls is more elaborate.
Function definitions are statically scoped, and the semantics is eager, so
arguments are evaluated using call-by-value.
The initial step is thus to evaluate all the arguments to values if
possible---propagating the exceptional result otherwise (\textsc{E-Call-Arg-Exc})---and
then check whether the values have the right type (otherwise producing an error, \textsc{E-Call-Arg-Err}). The evaluation
proceeds by evaluating the body of the function with a fresh store that
contains the values of global variables and the parameters bound to their
respective argument values.
There are then four cases:
\begin{enumerate}
\item If the body successfully evaluates to a correctly typed value, then that
  value is provided as the result (\textsc{E-Call-Sucs}).
  \item If the body evaluates to a value that does not have the expected type,
    then it produces an error (\textsc{E-Call-Res-Err1}).
  \item If the result is a thrown exception or error, then it is propagated
    (\textsc{E-Call-Res-Exc}).
    \item Otherwise if the result is a control operator, then an error is
      produced (\textsc{E-Call-Res-Err2}).
\end{enumerate}
In all cases the resulting store of executing the body
is discarded---since local assignments fall out of scope---except the global
variable values which are added to the store that was there before
the function call.

{\small\[
  \inference[\textsc{E-Call-Sucs}]{\esequence{\keyword{global} \; t_y \; y} &
    \keyword{fun} \; t' \; f(\esequence{t \; x}) = e' \\
    \evalexprstar{\esequence{e}}{\sigma}{\keyword{success} \; \esequence{v}}{\sigma''} &
    \esequence{\typing{v}{t''}} &
    \esequence{\subtyping{t''}{t}} \\
    \evalexpr{[\esequence{y \mapsto \sigma''(y)},
      \esequence{x \mapsto
      v}]}{e'}{\mathit{vres}}{\sigma'} \\ \mathit{vres} =
    \keyword{return} \; v' \vee \mathit{vres} = \keyword{success} \; v' &
    \typing{v'}{t'''} & \subtyping{t'''}{t'}
  }{\evalexpr{f(\esequence{e})}{\sigma}{\keyword{success} \; v'}{\sigma''[\esequence{y \mapsto \sigma'(y)}]}}
\]

\[
  \inference[\textsc{E-Call-Arg-Err}]{\keyword{fun} \; t' \; f(\esequence{t \; x}) = e' &
    \evalexprstar{\esequence{e}}{\sigma}{\keyword{success} \; \esequence{v}}{\sigma'} \\
    \esequence{\typing{v}{t''}} & \notsubtyping{t''_i}{t_i} }{\evalexpr{f(\esequence{e})}{\sigma}{\keyword{error}}{\sigma'}}
\]

\[
  \inference[\textsc{E-Call-Arg-Exc}]{
    \evalexprstar{\esequence{e}}{\sigma}{\mathit{exres}}{\sigma'}}{\evalexpr{f(\esequence{e})}{\sigma}{\mathit{exres}}{\sigma'}}
\]

\[
  \inference[\textsc{E-Call-Res-Exc}]{\esequence{\keyword{global} \; t_y \; y} &
    \keyword{fun} \; t' \; f(\esequence{t \; x}) = e' \\
    \evalexprstar{\esequence{e}}{\sigma}{\keyword{success} \; \esequence{v}}{\sigma''} &
    \esequence{\typing{v}{t''}} & \esequence{\subtyping{t''}{t}} \\
    \evalexpr{[\esequence{y \mapsto \sigma''(y)},
      \esequence{x \mapsto
      v}]}{e'}{\mathit{exres}}{\sigma'} &\mathit{exres} =
    \keyword{throw} \; v'
  }{\evalexpr{f(\esequence{e})}{\sigma}{\mathit{exres}}{\sigma''[\esequence{y \mapsto \sigma'(y)}]}}
\]

\[
  \inference[\textsc{E-Call-Res-Err1}]{\esequence{\keyword{global} \; t_y \; y} &
    \keyword{fun} \; t' \; f(\esequence{t \; x}) = e' \\
    \evalexprstar{\esequence{e}}{\sigma}{\keyword{success} \; \esequence{v}}{\sigma''} &
    \esequence{\typing{v}{t''}} & \esequence{\subtyping{t''}{t}} \\
    \evalexpr{[\esequence{y \mapsto \sigma''(y)},
      \esequence{x \mapsto
      v}]}{e'}{\mathit{vres}}{\sigma'} \\ \mathit{vres} =
    \keyword{return} \; v' \vee \mathit{vres} = \keyword{success} \; v' &
    \typing{v'}{t'''} & \notsubtyping{t'''}{t'}
  }{\evalexpr{f(\esequence{e})}{\sigma}{\keyword{error}}{\sigma''[\esequence{y \mapsto \sigma'(y)}]}}
\]

\[
  \inference[\textsc{E-Call-Res-Err2}]{\esequence{\keyword{global} \; t_y \; y} &
    \keyword{fun} \; t' \; f(\esequence{t \; x}) = e' \\
    \evalexprstar{\esequence{e}}{\sigma}{\keyword{success} \; \esequence{v}}{\sigma''} &
    \esequence{\typing{v}{t''}} & \esequence{\subtyping{t''}{t}} \\
    \evalexpr{[\esequence{y \mapsto \sigma''(y)},
      \esequence{x \mapsto
      v}]}{e'}{\mathit{exres}}{\sigma'} \\ \mathit{exres} \in
  \{\keyword{break}, \keyword{continue}, \keyword{fail}, \keyword{error} \}
  }{\evalexpr{f(\esequence{e})}{\sigma}{\keyword{error}}{\sigma''[\esequence{y \mapsto \sigma'(y)}]}}
\]}

\noindent{}The $\keyword{return}$-expression evaluates its argument expression first, and if
it successfully produces a value then the result would be an early return with
that value (\textsc{E-Ret-Sucs}); recall that early returns are treated as
exceptional values and so propagated through most evaluation rules, except at
function call boundaries (rules \textsc{E-Call-Sucs} and \textsc{E-Call-Res-Err1}). Otherwise, the exceptional result is propagated (\textsc{E-Ret-Exc}).

{\small\[
  \inference[\textsc{E-Ret-Sucs}]{\evalexpr{e}{\sigma}{\keyword{success} \; v}{\sigma'}}{\evalexpr{\keyword{return} \;
      e}{\sigma}{\keyword{return} \; v}{\sigma'}}
  \quad
  \inference[\textsc{E-Ret-Exc}]{\evalexpr{e}{\sigma}{\mathit{exres}}{\sigma'}}{\evalexpr{\keyword{return} \;
      e}{\sigma}{\mathit{exres}}{\sigma'}}
\]}

\noindent{}In Rascal Light, variables must be declared before being
assigned, and declarations are unique since shadowing is disallowed. Evaluating
an assignment proceeds by evaluating the right-hand side
expression---propagating exceptional results (\textsc{E-Asgn-Exc})---and then
checking whether the produced value is compatible with the declared type.
If it is compatible then the store is updated
(\textsc{E-Asgn-Sucs}), and otherwise an error is produced (\textsc{E-Asgn-Err}).

{\small
\[
  \inference[\textsc{E-Asgn-Sucs}]{\keyword{local} \; t \; x \vee
    \keyword{global} \; t \; x & \evalexpr{e}{\sigma}{\keyword{success}
      \; v}{\sigma'} \\ \typing{v}{t'} & \subtyping{t'}{t} }{\evalexpr{x = e}{\sigma}{\keyword{success} \;
      v}{\sigma'[x \mapsto v]}}
\]

\[
  \inference[\textsc{E-Asgn-Err}]{\keyword{local} \; t \; x \vee
    \keyword{global} \; t \; x & \evalexpr{e}{\sigma}{\keyword{success}
      \; v}{\sigma'} \\ \typing{v}{t'} & \notsubtyping{t'}{t}}{\evalexpr{x = e}{\sigma}{\keyword{error}}{\sigma'}}
\]

\[
  \inference[\textsc{E-Asgn-Exc}]{\evalexpr{e}{\sigma}{\mathit{exres}}{\sigma'}}{\evalexpr{x = e}{\sigma}{\mathit{exres}}{\sigma'}}
\]}

\noindent{}The $\keyword{if}$-expression works like other languages: the
$\keyword{then}$-branch is evaluated if the condition is true (\textsc{E-If-True}), otherwise the
$\keyword{else}$-branch is evaluated (\textsc{E-If-False}). If the conditional
produces a non-Boolean
value then an error is raised (\textsc{E-If-Err}) and otherwise exceptional
results are propagated (\textsc{E-If-Exc}).

{\small\[
  \inference[\textsc{E-If-True}]{\evalexpr{e_{\mathrm{cond}}}{\sigma}{\keyword{success} \;
      \auxiliary{true}()}{\sigma''} & \evalexpr{e_1}{\sigma''}{\mathit{vres}_1}{\sigma'}}{\evalexpr{\keyword{if} \; e_{\mathrm{cond}}
      \; \keyword{then} \; e_1 \; \keyword{else} \; e_2}{\sigma}{\mathit{vres}_1}{\sigma'}}
\]

\[
  \inference[\textsc{E-If-False}]{\evalexpr{e_{\mathrm{cond}}}{\sigma}{\keyword{success} \;
      \auxiliary{false}()}{\sigma''} & \evalexpr{e_2}{\sigma''}{\mathit{vres}_2}{\sigma'}}{\evalexpr{\keyword{if} \; e_{\mathrm{cond}}
      \; \keyword{then} \; e_1 \; \keyword{else} \; e_2}{\sigma}{\mathit{vres}_2}{\sigma'}}
\]

\[
  \inference[\textsc{E-If-Err}]{\evalexpr{e_{\mathrm{cond}}}{\sigma}{\keyword{success} \;
      v}{\sigma'} & v \neq \auxiliary{true}() & v \neq \auxiliary{false}()}{\evalexpr{\keyword{if} \; e_{\mathrm{cond}}
      \; \keyword{then} \; e_1 \; \keyword{else} \; e_2}{\sigma}{\keyword{error}}{\sigma'}}
\]

\[
  \inference[\textsc{E-If-Exc}]{\evalexpr{e_{\mathrm{cond}}}{\sigma}{\mathit{exres}}{\sigma'}}{\evalexpr{\keyword{if} \; e_{\mathrm{cond}}
      \; \keyword{then} \; e_1 \; \keyword{else} \; e_2}{\sigma}{\mathit{exres}}{\sigma'}}
\]}

\noindent{}The $\keyword{switch}$-expression initially evaluates the scrutinee expression ($e$),
and then proceeds to execute the cases (discussed on page~\pageref{par:cases}) on the result value (\textsc{E-Switch-Sucs}). The
evaluation of cases is allowed to \keyword{fail}, in which case the evaluation
is successful and has the special value $\blacksquare$ (\textsc{E-Switch-Fail}); other
exceptional results are propagated as usual (\textsc{E-Switch-Exc1}, \textsc{E-Switch-Exc2}).

{\small\[
  \inference[\textsc{E-Switch-Sucs}]{\evalexpr{e}{\sigma}{\keyword{success} \;
      v}{\sigma''} & \evalcases{\esequence{\mathit{cs}}}{v}{\sigma''}{\keyword{success} \; v'}{\sigma'}}{\evalexpr{\keyword{switch} \; e \;
      \esequence{\mathit{cs}}}{\sigma}{\keyword{success} \; v'}{\sigma'}}
\]

\[
  \inference[\textsc{E-Switch-Fail}]{\evalexpr{e}{\sigma}{\keyword{success} \;
      v}{\sigma''} & \evalcases{\esequence{\mathit{cs}}}{v}{\sigma''}{\keyword{fail}}{\sigma'}}{\evalexpr{\keyword{switch} \; e \;
      \esequence{\mathit{cs}}}{\sigma}{\keyword{success} \; \blacksquare}{\sigma'}}
\]

\[
  \inference[\textsc{E-Switch-Exc1}]{\evalexpr{e}{\sigma}{\mathit{exres}}{\sigma'}}{\evalexpr{\keyword{switch} \; e \;
      \esequence{\mathit{cs}}}{\sigma}{\mathit{exres}}{\sigma'}}
\]

\[
  \inference[\textsc{E-Switch-Exc2}]{\evalexpr{e}{\sigma}{\keyword{success} \;
      v}{\sigma''} &
    \evalcases{\esequence{\mathit{cs}}}{v}{\sigma''}{\mathit{exres}}{\sigma'} &
    \mathit{exres} \neq \keyword{fail}}{\evalexpr{\keyword{switch} \; e \;
      \esequence{\mathit{cs}}}{\sigma}{\mathit{exres}}{\sigma'}}
\]}

\noindent{}The $\keyword{visit}$-expression has similar evaluation cases to
$\keyword{switch}$ (\textsc{E-Visit-Sucs}, \textsc{E-Visit-Fail},
\textsc{E-Visit-Exc1}, \textsc{E-Visit-Exc2}), except that the cases are evaluated using the
$\textrm{visit}$ relation that traverses the produced value of target expression
given the provided strategy.

{\small\[
  \inference[\textsc{E-Visit-Sucs}]{\evalexpr{e}{\sigma}{\keyword{success} \;
      v}{\sigma''} &
    \evalvisit{\mathit{st}}{\esequence{\mathit{cs}}}{v}{\sigma''}{\keyword{success}
    \; v'}{\sigma'}}{\evalexpr{
    \mathit{st} \; \keyword{visit} \; e \;
    \esequence{\mathit{cs}}}{\sigma}{\keyword{success} \; v'}{\sigma'}}
\]

\[
  \inference[\textsc{E-Visit-Fail}]{\evalexpr{e}{\sigma}{\keyword{success} \;
      v}{\sigma''} &
    \evalvisit{\mathit{st}}{\esequence{\mathit{cs}}}{v}{\sigma''}{\keyword{fail}}{\sigma'}}{\evalexpr{
    \mathit{st} \; \keyword{visit} \; e \;
    \esequence{\mathit{cs}}}{\sigma}{\keyword{success} \; v}{\sigma'}}
\]

\[
  \inference[\textsc{E-Visit-Exc1}]{\evalexpr{e}{\sigma}{\mathit{exres}}{\sigma'}}{\evalexpr{
      \mathit{st} \; \keyword{visit} \; e \; \esequence{\mathit{cs}}}{\sigma}{\mathit{exres}}{\sigma'}}
\]

\[
  \inference[\textsc{E-Visit-Exc2}]{\evalexpr{e}{\sigma}{\keyword{success} \;
      v}{\sigma''} &
    \evalvisit{\mathit{st}}{\esequence{\mathit{cs}}}{v}{\sigma''}{\mathit{exres}}{\sigma'}
    &
    \mathit{exres} \neq \keyword{fail}}{\evalexpr{
    \mathit{st} \; \keyword{visit} \; e \;
    \esequence{\mathit{cs}}}{\sigma}{\mathit{exres}}{\sigma'}}
\]}

\noindent{}The control operations $\keyword{break}$, $\keyword{continue}$ and
$\keyword{fail}$ evaluate to themselves without any side-effects
(\textsc{E-Break}, \textsc{E-Continue}, \textsc{E-Fail}).

{\small\[
  \inference[\textsc{E-Break}]{}{\evalexpr{\keyword{break}}{\sigma}{\keyword{break}}{\sigma}} \quad
  \inference[\textsc{E-Fail}]{}{\evalexpr{\keyword{fail}}{\sigma}{\keyword{fail}}{\sigma}}
\]

\[
  \inference[\textsc{E-Continue}]{}{\evalexpr{\keyword{continue}}{\sigma}{\keyword{continue}}{\sigma}}
\]}

\noindent{}Blocks allow evaluating inner expressions using a local declaration of
variables, which are then afterwards removed from the resulting store
(\textsc{E-Block-Sucs}, \textsc{E-Block-Exc}).
Recall, that we consider an implicit definition environment based on scoping,
and so the local declarations in the block will be implicitly available in the
evaluation of the body subexpression sequence.

{\small\[
  \inference[\textsc{E-Block-Sucs}]{\evalexprstar{\esequence{e}}{\sigma}{\keyword{success}
    \; \esequence{v}}{\sigma'}}{\evalexpr{\keyword{local} \; \esequence{t
        \; x} \; \keyword{in} \; \esequence{e} \;
      \keyword{end}}{\sigma}{\keyword{success} \; \auxiliary{last}(\esequence{v})}{(\sigma' \setminus \esequence{x})}}
\]
\\
\[
  \inference[\textsc{E-Block-Exc}]{\evalexprstar{\esequence{e}}{\sigma}{\mathit{exres}}{\sigma'}}{\evalexpr{\keyword{local} \; \esequence{t
        \; x} \; \keyword{in} \; \esequence{e} \;
      \keyword{end}}{\sigma}{\mathit{exres}}{(\sigma' \setminus \esequence{x})}}
\]
}

\noindent{}The auxiliary function $\auxiliary{last}$ is here used to extract the last
element in the sequence (or return $\blacksquare$ if empty).
{\small\begin{align*}
\auxiliary{last}(v_1, \dots, v_{\mathrm{n}}, v') &= v' \\
\auxiliary{last}(\varepsilon) &= \blacksquare
\end{align*}}

\noindent{}The $\keyword{for}$-loop evaluates target generator expression to
a set of possible \textit{environments} that represent possible assignments of
variables to values---propagating exceptions if necessary
(\textsc{E-For-Exc})---and then it iterates over each possible assignment using the
$\textrm{each}$-relation (\textsc{E-For-Sucs}).

{\small\[
  \inference[\textsc{E-For-Sucs}]{\evalgenexpr{g}{\sigma}{\keyword{success} \;
      \esequence{\rho}}{\sigma''} \\
    \evaleach{e}{\esequence{\rho}}{\sigma''}{\mathit{vres}}{\sigma'}
  }{\evalexpr{\keyword{for} \; g \; e}{\sigma}{\mathit{vres}}{\sigma'}}
  \quad
  \inference[\textsc{E-For-Exc}]{\evalgenexpr{g}{\sigma}{\mathit{exres}}{\sigma'}}{\evalexpr{\keyword{for} \; g \; e}{\sigma}{\mathit{exres}}{\sigma'}}
\]}

\noindent{}The evaluation of $\keyword{while}$-loops is analogous to other imperative languages with
control operations, in that the body of the while loop is continuously executed until
the target condition does not hold (\textsc{E-While-False}). If the body successfully finishes with an
value or $\keyword{continue}$ then iteration continues (\textsc{E-While-True-Sucs}), if the body finishes
with $\keyword{break}$ the iteration stops with value $\blacksquare$
(\textsc{E-While-True-Break}), if the conditional evaluates to a non-Boolean
value it errors out (\textsc{E-While-Err}) and otherwise if
another kind of exceptional result is produced then it is propagated
(\textsc{E-While-Exc1}, \textsc{E-While-Exc2}).

{\small\[
  \inference[\textsc{E-While-False}]{\evalexpr{e_{\mathrm{cond}}}{\sigma}{\keyword{success} \;
      \auxiliary{false}()}{\sigma'}}{\evalexpr{\keyword{while}
      \; e_{\mathrm{cond}} \; e}{\sigma}{\keyword{success} \; \blacksquare}{\sigma'}}
\]

\[
  \inference[\textsc{E-While-True-Sucs}]{\evalexpr{e_{\mathrm{cond}}}{\sigma}{\keyword{success} \;
      \auxiliary{true}()}{\sigma''} &
    \evalexpr{e}{\sigma''}{\mathit{vres}}{\sigma'''} \\
    \mathit{vres} = \keyword{success} \; v \vee \mathit{vres} = \keyword{continue} \\
   \evalexpr{\keyword{while}
     \; e_{\mathrm{cond}} \; e}{\sigma'''}{\mathit{vres}'}{\sigma'} 
  }{\evalexpr{\keyword{while}
      \; e_{\mathrm{cond}} \; e}{\sigma}{\mathit{vres}'}{\sigma'}}
\]

\[
  \inference[\textsc{E-While-True-Break}]{\evalexpr{e_{\mathrm{cond}}}{\sigma}{\keyword{success} \;
      \auxiliary{true}()}{\sigma''} & \evalexpr{e}{\sigma''}{\keyword{break}}{\sigma'}}{\evalexpr{\keyword{while}
      \; e_{\mathrm{cond}} \; e}{\sigma}{\keyword{success} \; \blacksquare}{\sigma'}}
\]

\[
  \inference[\textsc{E-While-Exc1}]{\evalexpr{e_{\mathrm{cond}}}{\sigma}{\mathit{exres}}{\sigma'}}{\evalexpr{\keyword{while} \; e_{\mathrm{cond}}
      \; e}{\sigma}{\mathit{exres}}{\sigma'}}
\]

\[
  \inference[\textsc{E-While-Exc2}]{\evalexpr{e_{\mathrm{cond}}}{\sigma}{\keyword{success} \;
      \auxiliary{true}()}{\sigma''} &
    \evalexpr{e}{\sigma''}{\mathit{exres}}{\sigma'} \\
    \mathit{exres} \in \{ \keyword{throw} \; v, \keyword{return} \; v, \keyword{fail}, \keyword{error} \}}{\evalexpr{\keyword{while}
      \; e_{\mathrm{cond}} \; e}{\sigma}{\mathit{exres}}{\sigma'}}
\]

\[
  \inference[\textsc{E-While-Err}]{\evalexpr{e_{\mathrm{cond}}}{\sigma}{\keyword{success} \;
      v}{\sigma'} & v \neq \auxiliary{true}() & v \neq
    \auxiliary{false}()}{\evalexpr{\keyword{while} \; e_{\mathrm{cond}} \; e}{\sigma}{\keyword{error}}{\sigma'}}
\]}

\noindent{}The $\keyword{solve}$-loop keeps evaluating the body expression until the provided
variables reach a fixed-point. Initially, the body expression is evaluated and
then the values of target variables is compared from before and after iteration;
if the values are equal after an iteration, then evaluation stops (\textsc{E-Solve-Eq})
and otherwise the iteration continues (\textsc{E-Solve-Neq}). If any of the variables do
not have a value assigned, an error is produced (\textsc{E-Solve-Err}), and otherwise if an exceptional result is produced, it is propagated (\textsc{E-Solve-Exc}). %
{\small\[
  \inference[\textsc{E-Solve-Eq}]{\evalexpr{e}{\sigma}{\keyword{success} \;
      v}{\sigma'} & \esequence{x} \subseteq \textrm{dom } \sigma \cap
    \textrm{dom } \sigma' & \esequence{\sigma(x) = \sigma'(x)} }{\evalexpr{\keyword{solve} \;
      \esequence{x} \; e}{\sigma}{\keyword{success} \; v}{\sigma'}}
\]

\[
  \inference[\textsc{E-Solve-Neq}]{\evalexpr{e}{\sigma}{\keyword{success} \;
      v}{\sigma''}& \esequence{x} \subseteq \textrm{dom } \sigma \cap 
    \textrm{dom } \sigma'' \\ \exists i. \sigma(x_i) \neq
    \sigma''(x_i)  &
    \evalexpr{\keyword{solve} \;
      \esequence{x} \; e}{\sigma''}{\mathit{vres}}{\sigma'}
        }{\evalexpr{\keyword{solve} \;
            \esequence{x} \; e}{\sigma}{\mathit{vres}}{\sigma'}}
\]
\[
  \inference[\textsc{E-Solve-Exc}]{\evalexpr{e}{\sigma}{\mathit{exres}}{\sigma'}}{\evalexpr{\keyword{solve} \;
      \esequence{x} \; e}{\sigma}{\mathit{exres}}{\sigma'}}
\]

\[
  \inference[\textsc{E-Solve-Err}]{\evalexpr{e}{\sigma}{\keyword{success} \;
      v}{\sigma'} & x_i \not\in \textrm{dom } \sigma \cap
    \textrm{dom } \sigma'}{\evalexpr{\keyword{solve} \;
      \esequence{x} \; e}{\sigma}{\keyword{error}}{\sigma'}}
\]}

\noindent{}The $\keyword{throw}$-expression, evaluates its inner expression
first---propagating exceptional results if necessary (\textsc{E-Thr-Exc})---and then produces a $\keyword{throw}$ result with
result value (\textsc{E-Thr-Sucs}).

{\small
\[
  \inference[\textsc{E-Thr-Sucs}]{\evalexpr{e}{\sigma}{\keyword{success} \; v}{\sigma'}}{\evalexpr{\keyword{throw} \;
      e}{\sigma}{\keyword{throw} \; v}{\sigma'}}
  \quad
  \inference[\textsc{E-Thr-Exc}]{\evalexpr{e}{\sigma}{\mathit{exres}}{\sigma'}}{\evalexpr{\keyword{throw} \;
      e}{\sigma}{\mathit{exres}}{\sigma'}}
\]}

\noindent{}The $\keyword{try}$-$\keyword{finally}$ expression executes the
$\keyword{try}$-body first and then the $\keyword{finally}$-body.
If the $\keyword{finally}$-body produces an exceptional result during execution then that
result is propagated (\textsc{E-Fin-Exc}) and otherwise the $\keyword{try}$-body
result value is used (\textsc{E-Fin-Sucs}).

{\small\[
  \inference[\textsc{E-Fin-Sucs}]{\evalexpr{e_1}{\sigma}{\mathit{vres}_1}{\sigma''}
    &
       \evalexpr{e_2}{\sigma''}{\keyword{success} \; v_2}{\sigma'}}{\evalexpr{\keyword{try}
      \; e_1 \; \keyword{finally} \; e_2}{\sigma}{\mathit{vres}_1}{\sigma'}}
\]

\[
  \inference[\textsc{E-Fin-Exc}]{\evalexpr{e_1}{\sigma}{\mathit{vres}_1}{\sigma''} &
    \evalexpr{e_2}{\sigma''}{\mathit{exres}_2}{\sigma'}}{\evalexpr{\keyword{try}
      \; e_1 \; \keyword{finally} \; e_2}{\sigma}{\mathit{exres}_2}{\sigma'}}
\]}

\noindent{}The $\keyword{try}$-$\keyword{catch}$ expression evaluates the
$\keyword{try}$-body and if it produces a thrown value, then it binds the value
in the body of \keyword{catch} and continues evaluation (\textsc{E-Try-Catch}).
For all other results, it simply propagates them without evaluating the $\keyword{catch}$-body (\textsc{E-Try-Ord}).%
{\small\[
  \inference[\textsc{E-Try-Catch}]{\evalexpr{e_1}{\sigma}{\keyword{throw} \;
      v_1}{\sigma''} & \evalexpr{e_2}{\sigma''[x \mapsto
      v_1]}{\mathit{vres}_2}{\sigma'}}{\evalexpr{\keyword{try} \; e_1 \;
      \keyword{catch} \; x => e_2}{\sigma}{\mathit{vres}_2}{\left( \sigma'
        \setminus x \right)}}
\]

\[
  \inference[\textsc{E-Try-Ord}]{\evalexpr{e_1}{\sigma}{\mathit{vres}_1}{\sigma'}
    & \mathit{vres}_1 \neq \keyword{throw} \; v_1}{\evalexpr{\keyword{try} \; e_1 \; \keyword{catch} \; x =>
      e_2}{\sigma}{\mathit{vres}_1}{\sigma'}}
\]}

\paragraph{Expression Sequences}\label{par:exprseq}
Evaluating a sequence of expressions proceeds by evaluating each expression,
combining the results if successful (\textsc{ES-Emp}, \textsc{ES-More}) and
otherwise propagating the first exceptional result encountered
(\textsc{ES-Exc1}, \textsc{ES-Exc2}).

{\small
\[
  \inference[\textsc{ES-More}]{\evalexpr{e}{\sigma}{\keyword{success} \; v}{\sigma''} & \evalexprstar{
      \esequence{e'}}{\sigma''}{\keyword{success} \;\esequence{v'}}{\sigma'}}{\evalexprstar{e,
      \esequence{e'}}{\sigma}{\keyword{success} \; v, \esequence{v'}}{\sigma'}}
\]

\[
  \inference[\textsc{ES-Emp}]{}{\evalexprstar{\keyword{success} \; \varepsilon}{\sigma}{\varepsilon}{\sigma}}
  \quad
  \inference[\textsc{ES-Exc1}]{\evalexpr{e}{\sigma}{\mathit{exres}}{\sigma'} }{\evalexprstar{e,
      \esequence{e'}}{\sigma}{\mathit{exres}}{\sigma'}}
\]

\[
  \inference[\textsc{ES-Exc2}]{\evalexpr{e}{\sigma}{\keyword{success} \; v}{\sigma''} & \evalexprstar{
      \esequence{e'}}{\sigma''}{\mathit{exres}}{\sigma'}}{\evalexprstar{e,
      \esequence{e'}}{\sigma}{\mathit{exres}}{\sigma'}}
\]}

\paragraph{Cases}\label{par:cases}
The evaluation relation for evaluating a series of cases has the form
\highlightbox{$\evalcases{\esequence{\mathit{cs}}}{v}{\sigma}{\mathit{vres}}{\sigma'}$},
and intuitively proceeds by sequentially evaluating each case (in
$\esequence{\mathit{cs}}$) against value $v$ until one of them
produces a non-$\keyword{fail}$ result. For each case, the first step is to
match the given value against target pattern and then evaluate the target
expression under the set of possible matches; if the evaluation of the target
expression produces a $\keyword{fail}$ as result, the rest of the cases are
evaluated in a restored initial state (\textsc{ECS-More-Fail}) and otherwise the result is
propagated (\textsc{ECS-More-Ord}). If all possible cases are exhausted, the
result is $\keyword{fail}$ (\textsc{ECS-Emp}).%
{\small
\[
  \inference[\textsc{ECS-Emp}]{}{\evalcases{\varepsilon}{v}{\sigma}{\keyword{fail}}{\sigma}}
\]

\[
  \inference[\textsc{ECS-More-Fail}]{\match{p}{v}{\sigma}{\esequence{\rho}}
    & \evalcase{\esequence{\rho}}{e}{\sigma}{\keyword{fail}}{\sigma''} \\
    \evalcases{\esequence{\mathit{cs}}}{v}{\sigma}{\mathit{vres}}{\sigma'}
  }{\evalcases{\keyword{case} \;
    p => e, \esequence{\mathit{cs}}}{v}{\sigma}{\mathit{vres}}{\sigma'}}
\]

\[
  \inference[\textsc{ECS-More-Ord}]{\match{p}{v}{\sigma}{\esequence{\rho}}
    & \evalcase{\esequence{\rho}}{e}{\sigma}{\mathit{vres}}{\sigma'} \\ \mathit{vres} \neq \keyword{fail}}{\evalcases{\keyword{case} \;
    p => e, \esequence{\mathit{cs}}}{v}{\sigma}{\mathit{vres}}{\sigma'}}
\]}

\noindent{}Evaluating a single case---with relation
\highlightbox{$\evalcase{\esequence{\rho}}{e}{\sigma}{\mathit{vres}}{\sigma'}$}
---requires trying each possible binding (in $\esequence{\rho}$) sequentially,
producing $\keyword{fail}$ if no binding is available (\textsc{EC-Emp}). If evaluating target
expression produces a non-$\keyword{fail}$ value then it is propagated
(\textsc{EC-More-Ord}), otherwise the rest of the possible bindings are tried
in a restored initial state (\textsc{EC-More-Fail}).

{\small
\[
  \inference[\textsc{EC-More-Fail}]{\evalexpr{e}{\sigma\rho}{\keyword{fail}}{\sigma''}
    & \evalcase{\esequence{\rho'}}{e}{\sigma}{\mathit{vres}}{\sigma'}}{\evalcase{\rho,
      \esequence{\rho'}}{e}{\sigma}{\mathit{vres}}{\sigma'}}
\]

\[
  \inference[\textsc{EC-Emp}]{}{\evalcase{\varepsilon}{e}{\sigma}{\keyword{fail}}{\sigma}}
  \quad
  \inference[\textsc{EC-More-Ord}]{\evalexpr{e}{\sigma\rho}{\mathit{vres}}{\sigma'}
    & \mathit{vres} \neq \keyword{fail}}{\evalcase{\rho,
      \esequence{\rho'}}{e}{\sigma}{\mathit{vres}}{(\sigma' \setminus
      \mathrm{dom} \; \rho})}
\]}

\paragraph{Traversals}
One of the key features of Rascal is $\keyword{visit}$-expressions which provide
generic traversals over data values and collections, allowing for multiple
strategies to determine the traversal ordering and halting conditions.
In a traditional object-oriented language or functional language, transforming
large structures is cumbersome and requires a great amount of boilerplate,
requiring a function for each type of datatype, where each function must
deconstruct the input data, applying target changes, recursively calling the right
traversal functions for traversal of further contained data and reconstructing the data with new values.
Precisely, the first-class handling of these aspects makes Rascal particularly
suitable as a high-level transformation language.

The main traversal relation
\highlightbox{$\evalvisit{\mathit{st}}{\esequence{\mathit{cs}}}{v}{\sigma}{\mathit{vres}}{\sigma'}$}
delegates execution to the correct strategy-dependent traversal, and performs
fixed-point calculation if necessary. For the
$\keyword{top-down-break}$ and $\keyword{top-down}$ strategies it uses the
top-down traversal relation
\highlightbox{$\evaltdvisit{\esequence{\mathit{cs}}}{v}{\sigma}{\mathit{br}}{\mathit{vres}}{\sigma'}$}
specifying $\keyword{break}$ and no $\keyword{no-break}$ as breaking strategies
respectively (\textsc{EV-TD} and \textsc{EV-TDB}); this works analogously with
the $\keyword{bottom-up-break}$  (\textsc{EV-BU} and \textsc{EV-BUB}) and
$\keyword{bottom-up}$ strategies using the
\highlightbox{$\evalbuvisit{\esequence{\mathit{cs}}}{v}{\sigma}{\mathit{br}}{\mathit{vres}}{\sigma'}$}
relation.
{\small\[
  \inference[\textsc{EV-TD}]{
    \evaltdvisit{\esequence{\mathit{cs}}}{v}{\sigma}{\keyword{no-break}}{\mathit{vres}}{\sigma'}}{\evalvisit{\keyword{top-down}}{\esequence{\mathit{cs}}}{v}{\sigma}{\mathit{vres}}{\sigma'}}
  \quad
  \inference[\textsc{EV-TDB}]{
    \evaltdvisit{\esequence{\mathit{cs}}}{v}{\sigma}{\keyword{break}}{\mathit{vres}}{\sigma'}}{\evalvisit{\keyword{top-down-break}}{\esequence{\mathit{cs}}}{v}{\sigma}{\mathit{vres}}{\sigma'}}
\]
\[
  \inference[\textsc{EV-BU}]{
    \evalbuvisit{\esequence{\mathit{cs}}}{v}{\sigma}{\keyword{no-break}}{\mathit{vres}}{\sigma'}}{\evalvisit{\keyword{bottom-up}}{\esequence{\mathit{cs}}}{v}{\sigma}{\mathit{vres}}{\sigma'}}
  \quad
  \inference[\textsc{EV-BUB}]{
    \evalbuvisit{\esequence{\mathit{cs}}}{v}{\sigma}{\keyword{break}}{\mathit{vres}}{\sigma'}}{\evalvisit{\keyword{bottom-up-break}}{\esequence{\mathit{cs}}}{v}{\sigma}{\mathit{vres}}{\sigma'}}
\]}

\noindent{}The $\keyword{innermost}$ streatgy evaluates the $\keyword{bottom-up}$ traversal
as long as it produces a resulting value not equal to the one from the previous iteration (\textsc{EV-IM-Neq}), returning the result value when a fixed-point is reached
(\textsc{EV-IM-Eq}); if any exceptional result happens during evaluation it will
be propagated (\textsc{EV-IM-Exc}). Analogous evaluation steps happens with
$\keyword{outermost}$ streatgy and $\keyword{top-down}$ traversal
(\textsc{EV-OM-Neq}, \textsc{EV-OM-Eq}, \textsc{EV-OM-Exc}).

{\small
\[
  \inference[\textsc{EV-IM-Eq}]{
    \evalbuvisit{\esequence{\mathit{cs}}}{v}{\sigma}{\keyword{no-break}}{\keyword{success} \;
      v}{\sigma'}}{\evalvisit{\keyword{innermost}}{\esequence{\mathit{cs}}}{v}{\sigma}{\keyword{success}
    \; v}{\sigma'}}
\quad
  \inference[\textsc{EV-IM-Exc}]{
    \evalbuvisit{\esequence{\mathit{cs}}}{v}{\sigma}{\keyword{no-break}}{\mathit{exres}}{\sigma'}}{\evalvisit{\keyword{innermost}}{\esequence{\mathit{cs}}}{v}{\sigma}{\mathit{exres}}{\sigma'}}
\]}

{\small
\[
  \inference[\textsc{EV-IM-Neq}]{
    \evalbuvisit{\esequence{\mathit{cs}}}{v}{\sigma}{\keyword{no-break}}{\keyword{success} \;
      v'}{\sigma''} & v \neq v'\\
    \evalvisit{\keyword{innermost}}{\esequence{\mathit{cs}}}{v'}{\sigma''}{\mathit{vres}}{\sigma'}}{\evalvisit{\keyword{innermost}}{\esequence{\mathit{cs}}}{v}{\sigma}{\mathit{vres}}{\sigma'}}
\]

\[
  \inference[\textsc{EV-OM-Eq}]{
    \evaltdvisit{\esequence{\mathit{cs}}}{v}{\sigma}{\keyword{no-break}}{\keyword{success} \;
      v}{\sigma'}}{\evalvisit{\keyword{outermost}}{\esequence{\mathit{cs}}}{v}{\sigma}{\keyword{success}
    \; v}{\sigma'}}
\quad
  \inference[\textsc{EV-OM-Exc}]{
    \evaltdvisit{\esequence{\mathit{cs}}}{v}{\sigma}{\keyword{no-break}}{\mathit{exres}}{\sigma'}}{\evalvisit{\keyword{outermost}}{\esequence{\mathit{cs}}}{v}{\sigma}{\mathit{exres}}{\sigma'}}
\]

\[
  \inference[\textsc{EV-OM-Neq}]{
    \evaltdvisit{\esequence{\mathit{cs}}}{v}{\sigma}{\keyword{no-break}}{\keyword{success} \;
      v'}{\sigma''} & v \neq v'\\
    \evalvisit{\keyword{outermost}}{\esequence{\mathit{cs}}}{v'}{\sigma''}{\mathit{vres}}{\sigma'}}{\evalvisit{\keyword{outermost}}{\esequence{\mathit{cs}}}{v}{\sigma}{\mathit{vres}}{\sigma'}}
\]
}

\noindent{}The top-down traversal strategy starts by executing all cases on the target
value, \emph{scrutinee}, applying possible replacements and effects to produce
an intermediate result value; the traversal then continues on the sequence of contained values of the
this intermediate result, finally reconstructing
a new output containing the possible replacement values obtained
(\textsc{ETV-Ord-Sucs1}, \textsc{ETV-Ord-Sucs2}). If using the
$\keyword{break}$ strategy, the traversal will stop at the first value that
produces a successful result (\textsc{ETV-Break-Sucs});
otherwise, if any sub-result produces a non-$\keyword{fail}$ exceptional result
it is propagated (\textsc{ETV-Exc1}, \textsc{ETV-Exc2}).
{\small
\[
  \inference[\textsc{ETV-Break-Sucs}]{\evalcases{\esequence{\mathit{cs}}}{v}{\sigma}{\keyword{success}
      \; v'}{\sigma'} & \mathit{br} = \keyword{break}
  }{\evaltdvisit{\esequence{\mathit{cs}}}{v}{\sigma}{\mathit{br}}{\keyword{success}
    \; v'}{\sigma'}}
\]
\[
  \inference[\textsc{ETV-Ord-Sucs1}]{\evalcases{\esequence{\mathit{cs}}}{v}{\sigma}{\mathit{vfres}}{\sigma''}
    & \mathit{br} \neq \keyword{break} \vee \mathit{vfres} = \keyword{fail} \\ v'' = \textrm{if-fail}(\mathit{vfres}, v) 
    & \esequence{v'''} = \textrm{children}(v'') \\
    \evaltdvisitstar{\esequence{cs}}{\esequence{v'''}}{\sigma''}{\mathit{br}}{\keyword{fail}}{\sigma'}
  }{\evaltdvisit{\esequence{\mathit{cs}}}{v}{\sigma}{\mathit{br}}{\mathit{vfres}}{\sigma'}}
\]
\[
  \inference[\textsc{ETV-Ord-Sucs2}]{\evalcases{\esequence{\mathit{cs}}}{v}{\sigma}{\mathit{vfres}}{\sigma''}
    & \mathit{br} \neq \keyword{break} \vee \mathit{vfres} = \keyword{fail} \\ v'' = \textrm{if-fail}(\mathit{vfres}, v) 
    & \esequence{v'''} = \textrm{children}(v'') \\
    \evaltdvisitstar{\esequence{cs}}{\esequence{v'''}}{\sigma''}{\mathit{br}}{\keyword{success}
      \; \esequence{v''''}}{\sigma'} &
    \reconstruct{v''}{\esequence{v''''}}{\mathit{rcres}}}{\evaltdvisit{\esequence{\mathit{cs}}}{v}{\sigma}{\mathit{br}}{\mathit{rcres}}{\sigma'}}
\]
\[
  \inference[\textsc{ETV-Exc1}]{\evalcases{\esequence{\mathit{cs}}}{v}{\sigma}{\mathit{exres}}{\sigma'}
  & \mathit{exres} \neq \keyword{fail}}{\evaltdvisit{\esequence{\mathit{cs}}}{v}{\sigma}{\mathit{br}}{\mathit{exres}}{\sigma'}}
\]

\[
  \inference[\textsc{ETV-Exc2}]{\evalcases{\esequence{\mathit{cs}}}{v}{\sigma}{\mathit{vfres}}{\sigma''}
    & \mathit{br} \neq \keyword{break} \vee \mathit{vfres} = \keyword{fail} \\ v'' = \textrm{if-fail}(\mathit{vfres}, v) & \esequence{v'''} = \textrm{children}(v'')
    \\
    \evaltdvisitstar{\esequence{cs}}{\esequence{v'''}}{\sigma''}{\mathit{br}}{\mathit{exres}}{\sigma'}
    & \mathit{exres} \neq \keyword{fail}}{\evaltdvisit{\esequence{\mathit{cs}}}{v}{\sigma}{\mathit{br}}{\mathit{exres}}{\sigma'}}
\]}

\noindent{}Evaluating a sequence of top-down traversals, requires executing a top-down
traversal for each element, failing if the input sequence is
empty (\textsc{ETVS-Emp}) and otherwise combining the results
(\textsc{ETVS-More}).
If the $\keyword{break}$ strategy is used, then the iteration will instead stop
at first succesful result (\textsc{ETVS-Break}), and any non-\keyword{fail} exceptional result
is propagated
(\textsc{ETVS-Exc1}, \textsc{ETVS-Exc2}).

{\small\[
  \inference[\textsc{ETVS-Emp}]{}{\evaltdvisitstar{\esequence{cs}}{\varepsilon}{\sigma}{\mathit{br}}{\keyword{fail}}{\sigma}}
\]

\[
  \inference[\textsc{ETVS-Break}]{\evaltdvisit{\esequence{cs}}{v}{\sigma}{\mathit{br}}{\keyword{success} \; v''}{\sigma'}
    & \mathit{br} =
    \keyword{break}}{\evaltdvisitstar{\esequence{cs}}{v,\esequence{v'}}{\sigma}{\mathit{br}}{\keyword{success}
      \; v'', \esequence{v'}}{\sigma'}}
\]

\[
  \inference[\textsc{ETVS-More}]{\evaltdvisit{\esequence{cs}}{v}{\sigma}{\mathit{br}}{\mathit{vfres}}{\sigma''}
    & \mathit{br} \neq \keyword{break} \vee \mathit{vfres} = \keyword{fail}
    \\
    \evaltdvisitstar{\esequence{cs}}{\esequence{v'}}{\sigma''}{\mathit{br}}{\mathit{vfres}{\star}'}{\sigma'}}{\evaltdvisitstar{\esequence{cs}}{v,\esequence{v'}}{\sigma}{\mathit{br}}{\textrm{vcombine}(\mathit{vfres},
      \mathit{vfres}{\star}', v, \esequence{v'})}{\sigma'}}
\]

\[
  \inference[\textsc{ETVS-Exc1}]{\evaltdvisit{\esequence{cs}}{v}{\sigma}{\mathit{br}}{\mathit{exres}}{\sigma'}
    & \mathit{exres} \neq \keyword{fail}}{\evaltdvisitstar{\esequence{cs}}{v,\esequence{v'}}{\sigma}{\mathit{br}}{\mathit{exres}}{\sigma'}}
\]

\[
  \inference[\textsc{ETVS-Exc2}]{\evaltdvisit{\esequence{cs}}{v}{\sigma}{\mathit{br}}{\mathit{vfres}}{\sigma''}
    & \mathit{br} \neq \keyword{break} \vee \mathit{vfres} = \keyword{fail} 
    \\ \evaltdvisitstar{\esequence{cs}}{\esequence{v'}}{\sigma''}{\mathit{br}}{\mathit{exres}}{\sigma'}
    & \mathit{exres} \neq \keyword{fail}}{\evaltdvisitstar{\esequence{cs}}{v,\esequence{v'}}{\sigma}{\mathit{br}}{\mathit{exres}}{\sigma'}}
\]}

\noindent{}Bottom-up traversals work analogously to top-down traversals, except that
traversal of children and reconstruction happens before traversing the final
reconstructed value (\textsc{EBU-Break-Sucs}, \textsc{EBU-No-Break-Sucs},
\textsc{EBU-Fail-Sucs}). The analogy also holds with propagation of exceptional
results and errors (\textsc{EBU-Exc}, \textsc{EBU-No-Break-Err}, \textsc{EBU-No-Break-Exc}).
{\small\[
  \inference[\textsc{EBU-Break-Sucs}]{ \esequence{v''} =
    \textrm{children}(v) &
    \evalbuvisitstar{\esequence{cs}}{\esequence{v''}}{\sigma}{\mathit{br}}{\keyword{success}
      \; \esequence{v'}}{\sigma'} \\
    \mathit{br} = \keyword{break} &
    \reconstruct{v}{\esequence{v'}}{\mathit{rcres}}}{\evalbuvisit{\esequence{\mathit{cs}}}{v}{\sigma}{\mathit{br}}{\mathit{rcres}}{\sigma'}}
\]

\[
  \inference[\textsc{EBU-No-Break-Sucs}]{ \esequence{v''} =
    \textrm{children}(v) &
    \evalbuvisitstar{\esequence{cs}}{\esequence{v''}}{\sigma}{\mathit{br}}{\keyword{success}
      \; \esequence{v'''}}{\sigma''} \\
\mathit{br} = \keyword{no-break} &
\reconstruct{v}{\esequence{v'''}}{\keyword{success} \; v'} \\
\evalcases{\esequence{\mathit{cs}}}{v'}{\sigma''}{\mathit{vfres}'}{\sigma'} \\
}{\evalbuvisit{\esequence{\mathit{cs}}}{v}{\sigma}{\mathit{br}}{\keyword{success}
    \; \textrm{if-fail}(\mathit{vfres}',
        v')}{\sigma'}}
\]

\[
  \inference[\textsc{EBU-Fail-Sucs}]{ \esequence{v''} =
    \textrm{children}(v) &
    \evalbuvisitstar{\esequence{cs}}{\esequence{v''}}{\sigma}{\mathit{br}}{\keyword{fail}}{\sigma''}\\
\evalcases{\esequence{\mathit{cs}}}{v}{\sigma''}{\mathit{vres}}{\sigma'}}{\evalbuvisit{\esequence{\mathit{cs}}}{v}{\sigma}{\mathit{br}}{\mathit{vres}}{\sigma'}}
\]

\[
  \inference[\textsc{EBU-Exc}]{ \esequence{v'} =
    \textrm{children}(v) &
    \evalbuvisitstar{\esequence{cs}}{\esequence{v'}}{\sigma}{\mathit{br}}{\mathit{exres}}{\sigma'}
    & \mathit{exres} \neq \keyword{fail}\\ 
     }{\evalbuvisit{\esequence{\mathit{cs}}}{v}{\sigma}{\mathit{br}}{\mathit{exres}}{\sigma'}}
\]

\[
  \inference[\textsc{EBU-No-Break-Err}]{ \esequence{v''} =
    \textrm{children}(v) &
    \evalbuvisitstar{\esequence{cs}}{\esequence{v''}}{\sigma}{\mathit{br}}{\keyword{success}
      \; \esequence{v'''}}{\sigma'} \\
\mathit{br} = \keyword{no-break} &
\reconstruct{v}{\esequence{v'''}}{\keyword{error}}}{\evalbuvisit{\esequence{\mathit{cs}}}{v}{\sigma}{\mathit{br}}{\keyword{error}}{\sigma'}}
\]

\[
  \inference[\textsc{EBU-No-Break-Exc}]{ \esequence{v''} =
    \textrm{children}(v) &
    \evalbuvisitstar{\esequence{cs}}{\esequence{v''}}{\sigma}{\mathit{br}}{\keyword{success}
      \; \esequence{v'''}}{\sigma''} \\
\mathit{br} = \keyword{no-break} &
\reconstruct{v}{\esequence{v'''}}{\keyword{success} \; v'} \\
\evalcases{\esequence{\mathit{cs}}}{v'}{\sigma''}{\mathit{exres}}{\sigma'} \\
}{\evalbuvisit{\esequence{\mathit{cs}}}{v}{\sigma}{\mathit{br}}{\mathit{exres}}{\sigma'}}
\]}

\noindent{}Evaluating a sequence of bottom-up traversals is analogous to evaluating a
sequence of top-down traversals. Each element in the sequence is evaluated and
their results is combined (\textsc{EBUS-Emp},\textsc{EBUS-More}),
stopping at the first succesful result when using the break strategy (\textsc{EBUS-Break}). Otherwise, non-\keyword{fail} exceptional results are propagated
(\textsc{EBUS-Exc1}, \textsc{EBUS-Exc2}).

{\small
\[
  \inference[\textsc{EBUS-Emp}]{}{\evalbuvisitstar{\esequence{cs}}{\varepsilon}{\sigma}{\mathit{br}}{\keyword{fail}}{\sigma}}
\]

\[
  \inference[\textsc{EBUS-Break}]{\evalbuvisit{\esequence{cs}}{v}{\sigma}{\mathit{br}}{\keyword{success}
      \; v''}{\sigma} & \mathit{br} =
    \keyword{break}}{\evalbuvisitstar{\esequence{cs}}{v,\esequence{v'}}{\sigma}{\mathit{br}}{\keyword{success}
      \; v'', \esequence{v'}}{\sigma'}}
\]

\[
  \inference[\textsc{EBUS-More}]{\evalbuvisit{\esequence{cs}}{v}{\sigma}{\mathit{br}}{\mathit{vfres}}{\sigma''}
    & \mathit{br} \neq \keyword{break} \vee \mathit{vfres} = \keyword{fail} \\
    \evalbuvisitstar{\esequence{cs}}{\esequence{v'}}{\sigma''}{\mathit{br}}{\mathit{vfres}{\star}'}{\sigma'}}{\evalbuvisitstar{\esequence{cs}}{v,\esequence{v'}}{\sigma}{\mathit{br}}{\textrm{vcombine}(\mathit{vfres},
      \mathit{vfres}{\star}', v, \esequence{v'})}{\sigma'}}
\]

\[
  \inference[\textsc{EBUS-Exc1}]{\evalbuvisit{\esequence{cs}}{v}{\sigma}{\mathit{br}}{\mathit{exres}}{\sigma'}
    & \mathit{exres} \neq \keyword{fail}}{\evalbuvisitstar{\esequence{cs}}{v,\esequence{v'}}{\sigma}{\mathit{br}}{\mathit{exres}}{\sigma'}}
\]

\[
  \inference[\textsc{EBUS-Exc2}]{\evalbuvisit{\esequence{cs}}{v}{\sigma}{\mathit{br}}{\mathit{vfres}}{\sigma''}
    & \mathit{br} \neq \keyword{break} \vee \mathit{vfres} = \keyword{fail} \\
    \evalbuvisitstar{\esequence{cs}}{\esequence{v'}}{\sigma''}{\mathit{br}}{\mathit{exres}}{\sigma'}
    & \mathit{exres} \neq \keyword{fail}}{\evalbuvisitstar{\esequence{cs}}{v,\esequence{v'}}{\sigma}{\mathit{br}}{\mathit{exres}}{\sigma'}}
\]}

\paragraph{Auxiliary}
The $\textrm{children}$ function extracts the directly contained values of the
given input value.
{\small\begin{alignat*}{4}
  \textrm{children}(\mathit{vb}) &= \varepsilon &
  \textrm{children}(k(\esequence{v})) &= \esequence{v} \\
  \textrm{children}([\esequence{v}]) &= \esequence{v} &
  \textrm{children}(\{\esequence{v}\}) &= \esequence{v} \\
  \textrm{children}((\esequence{v : v'})) &= \esequence{v}, \esequence{v'} &
  \textrm{children}(\blacksquare) &= \varepsilon
\end{alignat*}}
{\small
\begin{align*}
  \mathit{vfres} &\Coloneqq \keyword{success} \; v \mid \keyword{fail}
\end{align*}}

\noindent{}The $\textrm{if-fail}$ function will return a provided default value if the
first argument is $\keyword{fail}$, and otherwise it will use the provided value
in the first argument.
{\small\begin{alignat*}{4}
  \textrm{if-fail}(\keyword{fail}, v) &= v &\quad
  \textrm{if-fail}(\keyword{success} \; v', v) &= v'
\end{alignat*}}

\noindent{}The $\textrm{vcombine}$ function will combine $\keyword{success}$ and
$\keyword{fail}$ results from visitor, producing $\keyword{fail}$ if both result
arguments are $\keyword{fail}$ otherwise producing $\keyword{success}$ result, possibly using default values

{\small
\begin{align*}
  \textrm{vcombine}(\mathit{vfres}, \mathit{vfres}{\star}', v, \esequence{v'}) =
  \begin{cases}
    \keyword{fail} & \textbf{if }
    \begin{gathered}
      \mathit{vfres} = \keyword{fail} \wedge{} \\ \mathit{vfres}' = \keyword{fail}
    \end{gathered}
      \\
      \keyword{success} \; \left(
        \begin{gathered}
          \textrm{if-fail}(\mathit{vfres}, v), \\ \textrm{if-fail}(\mathit{vfres}{\star}', \esequence{v}')  
        \end{gathered}
\right)
        & \keyword{otherwise}
  \end{cases}
\end{align*}}

\noindent{}The reconstruction relation
\highlightbox{$\reconstruct{v}{\esequence{v'}}{\mathit{rcres}}$} tries to update
the elements of a value, checking whether provided values are type correct and
defined (not $\blacksquare$), otherwise producing an error.
{\small\[
  \mathit{rcres} \Coloneqq \keyword{success} \; v \mid \keyword{error}
\]}

{\small
\[
  \inference[\textsc{RC-Val-Sucs}]{}{\reconstruct{\mathit{vb}}{\varepsilon}{\keyword{success} \; \mathit{vb}}}
\]

\[
  \inference[\textsc{RC-Val-Err}]{}{\reconstruct{\mathit{vb}}{v', \esequence{v''}}{\keyword{error}}}
\]

\[
  \inference[\textsc{RC-Cons-Sucs}]{\keyword{data} \; \mathit{at} = \dots ~|~
    k(\esequence{t}) ~|~ \dots & \esequence{\typing{v'}{t'}} & \esequence{v' \neq \blacksquare}  & \esequence{\subtyping{t'}{t}}}{\reconstruct{k(\esequence{v})}{\esequence{v'}}{\keyword{success}
      \; k(\esequence{v'})}}
\]

\[
  \inference[\textsc{RC-Cons-Err}]{\keyword{data} \; \mathit{at} = \dots ~|~
    k(\esequence{t}) ~|~ \dots &
    \esequence{\typing{v'}{t'}} & v_i = \blacksquare \vee \notsubtyping{t'_i}{t_i}}{\reconstruct{k(\esequence{v})}{\esequence{v'}}{\keyword{error}}}
\]

\[
  \inference[\textsc{RC-List-Sucs}]{\esequence{v' \neq \blacksquare}}{\reconstruct{[\esequence{v}]}{\esequence{v'}}{\keyword{success}
      \; [\esequence{v'}]}}
\]

\[
  \inference[\textsc{RC-List-Err}]{v'_i = \blacksquare}{\reconstruct{[\esequence{v}]}{\esequence{v'}}{\keyword{error}}}
\]

\[
  \inference[\textsc{RC-Set-Sucs}]{\esequence{v' = \blacksquare}}{\reconstruct{\{\esequence{v}\}}{\esequence{v'}}{\keyword{success}
      \; \{\esequence{v'}\}}}
\]

\[
  \inference[\textsc{RC-Set-Err}]{v'_i = \blacksquare}{\reconstruct{\{\esequence{v}\}}{\esequence{v'}}{\keyword{error}}}
\]

\[
  \inference[\textsc{RC-Map-Sucs}]{\esequence{v'' \neq \blacksquare} &
    \esequence{v''' \neq \blacksquare}}{\reconstruct{(\esequence{v : v'})}{
      \esequence{v''}, \esequence{v'''}
        }{\keyword{success}
      \; (\esequence{v'' : v'''})}}
\]

\[
  \inference[\textsc{RC-Map-Err}]{v''_i = \blacksquare \vee v'''_i = \blacksquare}{\reconstruct{(\esequence{v : v'})}{
      \esequence{v''}, \esequence{v'''}
    }{\keyword{error}}}
\]

\[
  \inference[\textsc{RC-Bot-Sucs}]{}{\reconstruct{\blacksquare}{\varepsilon}{\keyword{success}
      \; \blacksquare}}
\]

\[
  \inference[\textsc{RC-Bot-Err}]{}{\reconstruct{\blacksquare}{v', \esequence{v''}}{\keyword{error}}}
\]}

\paragraph{Enumeration}
The enumeration relation
\highlightbox{$\evaleach{e}{\esequence{\rho}}{\sigma}{\mathit{vres}}{\sigma'}$}
iterates over all provided bindings (\textsc{EE-More-Sucs}) until there are
none left (\textsc{EE-Emp}) or the result is neither an ordinary value or $\keyword{continue}$
from one of the iterations (\textsc{EE-More-Exc}); in case the result is $\keyword{break}$ the
evaluation will terminate early with a succesful result (\textsc{EE-More-Break}).

{\small\[
  \inference[\textsc{EE-Emp}]{}{\evaleach{e}{\varepsilon}{\sigma}{\keyword{success} \; \blacksquare}{\sigma}}
\]

\[
  \inference[\textsc{EE-More-Sucs}]{\evalexpr{e}{\sigma
      \rho}{\mathit{vres}}{\sigma''} & \mathit{vres} = \keyword{success} \; v \vee
    \mathit{vres} = \keyword{continue} \\ \evaleach{e}{\esequence{\rho'}}{\left(
        \sigma'' \setminus \mathrm{dom} \; \rho \right)}{\mathit{vres}'}{\sigma'}}{\evaleach{e}{\rho, \esequence{\rho'}}{\sigma}{\mathit{vres}'}{\sigma'}}
\]

\[
  \inference[\textsc{EE-More-Break}]{\evalexpr{e}{\sigma
      \rho}{\keyword{break}}{\sigma'}}{\evaleach{e}{\rho,
      \esequence{\rho'}}{\sigma}{\keyword{success} \; \blacksquare}{\left( \sigma'
        \setminus \mathrm{dom} \; \rho \right)}}
\]

\[
  \inference[\textsc{EE-More-Exc}]{\evalexpr{e}{\sigma
      \rho}{\mathit{exres}}{\sigma'} & \mathit{exres} \in \{\keyword{throw} \;
    v, \keyword{return} \; v, \keyword{fail},
    \keyword{error}\}}{\evaleach{e}{\rho,
      \esequence{\rho'}}{\sigma}{\mathit{exres}}{\left( \sigma' \setminus
        \mathrm{dom} \; \rho \right)}}
\]}

\paragraph{Generator expressions}
\sloppypar
The evaluation relation for
generator expressions has form \highlightbox{$\evalgenexpr{g}{\sigma}{\mathit{envres}}{\sigma'}$}.
For matching assignments the target right-hand side expression is evaluated
first---propagating possible exceptional results (\textsc{G-Pat-Exc}) and then the
value is matched against target pattern (\textsc{G-Pat-Sucs}).
{\small
\begin{alignat*}{3}
  \mathit{envres} &\Coloneqq \keyword{success} \; \vec{\rho} \mid \mathit{exres}
\end{alignat*}
\[
  \inference[\textsc{G-Pat-Sucs}]{\evalexpr{e}{\sigma}{\keyword{success} \;
      v}{\sigma'} & \match{p}{v}{\sigma'}{\vec{\rho}}}{\evalgenexpr{p \coloneqq
      e}{\sigma}{\keyword{success} \; \vec{\rho}}{\sigma'}}
\]
\[
  \inference[\textsc{G-Pat-Exc}]{\evalexpr{e}{\sigma}{\mathit{exres}}{\sigma'}}{\evalexpr{p \coloneqq e}{\sigma}{\mathit{exres}}{\sigma'}}
\]}

\noindent{}For enumerating assignments, each possible value in a collection is provided as a
possible binding to the range variable in the output
(\textsc{G-Enum-List}, \textsc{G-Enum-Set}); for maps in particular,
only the keys are bounds (\textsc{G-Enum-Map}). An error is raised if the result value is not a
collection (\textsc{G-Enum-Err}), and exceptions are propagated as always (\textsc{G-Enum-Exc}).

{\small
\[
  \inference[\textsc{G-Enum-List}]{\evalexpr{e}{\sigma}{\keyword{success} \;
      [v_1, \dots, v_{\mathrm{n}}]}{\sigma'}}{\evalgenexpr{x <-
      e}{\sigma}{\keyword{success} \; [x \mapsto v_1], \dots, [x \mapsto v_{\mathrm{n}}]}{\sigma'}}
\]

\[
  \inference[\textsc{G-Enum-Set}]{\evalexpr{e}{\sigma}{\keyword{success} \;
      \{v_1, \dots, v_{\mathrm{n}}\}}{\sigma'}}{\evalgenexpr{x <-
      e}{\sigma}{\keyword{success} \; [x \mapsto v_1], \dots, [x \mapsto v_{\mathrm{n}}]}{\sigma'}}
\]

\[
  \inference[\textsc{G-Enum-Map}]{\evalexpr{e}{\sigma}{\keyword{success} \;
     (v_1 : v_1', \dots, v_{\mathrm{n}} : v_{\mathrm{n}}')}{\sigma'}}{\evalgenexpr{x <-
      e}{\sigma}{\keyword{success} \; [x \mapsto v_1], \dots, [x \mapsto v_{\mathrm{n}}]}{\sigma'}}
\]

\[
  \inference[\textsc{G-Enum-Err}]{\evalexpr{e}{\sigma}{\keyword{success} \;
      v }{\sigma'} & v = \mathit{vb} \vee v = k(\esequence{v'}) \vee v = \blacksquare}{\evalgenexpr{x <-
     e}{\sigma}{\keyword{error}}{\sigma'}}
\]

\[
  \inference[\textsc{G-Enum-Exc}]{\evalexpr{e}{\sigma}{\mathit{exres}}{\sigma'}}{\evalgenexpr{x <-
      e}{\sigma}{\mathit{exres}}{\sigma'}}
\]}

\subsection*{Pattern Matching}
The pattern matching relation
\highlightbox{$\match{p}{v}{\sigma}{\esequence{\rho}}$} takes as input the current
store $\sigma$, a pattern $p$ and a target value $v$ and produces a sequence of compatible
environments that represent possible bindings of variables to values.
Pattern matching a basic value against target value produces a single
environment that does not bind any variable ($[]$) if the target value is the
same basic value (\textsc{P-Val-Sucs}) and otherwise does not produce any binding environment ($\varepsilon$), (\textsc{P-Val-Fail}).

{\small
\[
  \inference[\textsc{P-Val-Sucs}]{}{\match{\mathit{vb}}{\mathit{vb}}{\sigma}{[]}}
  \quad
  \inference[\textsc{P-Val-Fail}]{v \neq \mathit{vb}}{\match{\mathit{vb}}{v}{\sigma}{\varepsilon}}
\]}

\noindent{}Pattern matching against a variable depends on whether the variable already
exists in the current store. If it is assigned in the current store, then the
target value must the assigned value to return a possible binding
(\textsc{P-Var-Uni}) and otherwise failing with no bindings
(\textsc{P-Var-Fail}); if it is not in the current store, it will simply bind
the variable to the target value (\textsc{P-Var-Bind}).

{\small
\[
  \inference[\textsc{P-Var-Uni}]{x \in \textrm{dom } \sigma & v =
    \sigma(x)}{\match{x}{v}{\sigma}{[]}}
  \quad
  \inference[\textsc{P-Var-Fail}]{x \in \textrm{dom } \sigma & v \neq \sigma(x)}{\match{x}{v}{\sigma}{\varepsilon}}
\]

\[
  \inference[\textsc{P-Var-Bind}]{x \notin \textrm{dom } \sigma}{\match{x}{v}{\sigma}{[x \mapsto v]}}
\]}

\noindent{}When pattern matching against a constructor pattern, it is first checked whether
the target value has the same constructor. If it does, then the sub-patterns are
matched against the contained values of the target value, merging their resulting environments
(\textsc{P-Cons-Sucs}), and otherwise failing with no bindings
(\textsc{P-Cons-Fail}).
The merging procedure is described formally later in this section, but it
intuitively it takes the union of all bindings that have consistent assignments
to the same variables.

{\small\[
  \inference[\textsc{P-Cons-Sucs}]{\match{p_1}{v_1}{\sigma}{\esequence{\rho_1}} & \dots &
    \match{p_{\mathrm{n}}}{v_{\mathrm{n}}}{\sigma}{\esequence{\rho_{\mathrm{n}}}}}{\match{k(\esequence{p})}{k(\esequence{v})}{\sigma}{\textrm{merge}(\esequence{\rho_1},
      \dots, \esequence{\rho_{\mathrm{n}}})}}
\]

\[
  \inference[\textsc{P-Cons-Fail}]{v \neq k(\esequence{v'})}{\match{k(\esequence{p})}{v}{\sigma}{\varepsilon}}
\]}

\noindent{}When pattern matching against a typed labeled pattern, the target value is
checked to have a compatible type---failing with no bindings otherwise (\textsc{P-Type-Fail})---
and then then the inner pattern is matched
against the same value. The result of the sub-pattern match is merged with the
environment where the target value is bound to the label variable (\textsc{P-Type-Sucs}).
{\small\[
  \inference[\textsc{P-Type-Sucs}]{\typing{v}{t'} & \subtyping{t'}{t} &
    \match{p}{v}{\sigma}{\esequence{\rho}}}{\match{t \; x:
      p}{v}{\sigma}{\textrm{merge}([x \mapsto v], \esequence{\rho})}}
\]

\[
  \inference[\textsc{P-Type-Fail}]{\typing{v}{t'} &
    \notsubtyping{t'}{t}}{\match{t \; x: p}{v}{\sigma}{\varepsilon}}
\]}

\noindent{}Pattern matching against a list pattern first checks whether the target value
is a list---otherwise failing (\textsc{P-List-Fail})---and then pattern matches
against the sub-patterns returning their result (\textsc{P-List-Sucs}). 
{\small\[
  \inference[\textsc{P-List-Sucs}]{\matchall{\esequence{{{\star}p}}}{\esequence{v}}{\sigma}{\emptyset}{[] \quad
      ,}{\esequence{\rho}}}{\match{[\esequence{{{\star}p}}]}{[\esequence{v}]}{\sigma}{\esequence{\rho}}}
  \quad
  \inference[\textsc{P-List-Fail}]{v \neq [\esequence{v'}]}{\match{[\esequence{{{\star}p}}]}{v}{\sigma}{\varepsilon}}
\]}

\noindent{}Pattern matching against a set pattern is analogous to pattern matching against
list patterns (\textsc{P-Set-Sucs}, \textsc{P-Set-Fail}).
{\small
\[
  \inference[\textsc{P-Set-Sucs}]{\matchall{\esequence{{{\star}p}}}{\esequence{v}}{\sigma}{\emptyset}{\{\} \quad \uplus}{\esequence{\rho}}}{\match{\{\esequence{{{\star}p}}\}}{\{\esequence{v}\}}{\sigma}{\esequence{\rho}}}
  \quad
  \inference[\textsc{P-Set-Fail}]{v \neq \{\esequence{v'}\}}{\match{\{\esequence{p}\}}{v}{\sigma}{\varepsilon}}
\]}

\noindent{}Negation pattern $!p$ matching succeeds with no variables bound if the
sub-pattern $p$ produces no binding environment (\textsc{P-Neg-Sucs}), and
otherwise fails (\textsc{P-Neg-Fail}).
{\small
\[
  \inference[\textsc{P-Neg-Sucs}]{\match{p}{v}{\sigma}{\varepsilon}}{\match{!p}{v}{\sigma}{[]}}
  \quad
  \inference[\textsc{P-Neg-Fail}]{\match{p}{v}{\sigma}{\rho} & \rho \neq \varepsilon}{\match{!p}{v}{\sigma}{\varepsilon}}
\]}

\noindent{}Descendant pattern matching applies the sub-pattern against target value, and keeps
applying the deep matching pattern against the children values, concatenating
their results (\textsc{P-Deep}).
{\small
\[
  \inference[\textsc{P-Deep}]{\match{p}{v}{\sigma}{\esequence{\rho}} &
    v'_{1},\dots,v'_{\mathrm{n}} = \textrm{children}(v) \\
    \match{/p}{v'_1}{\sigma}{\esequence{\rho_1'}} & \dots &
    \match{/p}{v'_{\mathrm{n}}}{\sigma}{\esequence{\rho_{\mathrm{n}}'}}}{\match{/p}{v}{\sigma}{\esequence{\rho},
      \esequence{\rho_1'}, \dots, \esequence{\rho_{\mathrm{n}}'}}}
\]}

\noindent{}The star pattern matching relation has form
\highlightbox{$\matchall{\esequence{{{\star}p}}}{\esequence{v}}{\sigma}{\mathbb{V}}{\ttuple{}
    \quad \otimes}{\esequence{\rho}}$} which tries to match the sequence of
patterns $\esequence{{{\star}p}}$ on the left-hand side with sequence of
  provided values $\esequence{v}$ on the
right-hand side. The relations is parameterized over the construction function
($\ttuple{}$) and the partition relation ($\otimes$), and because matching
is non-deterministic we keep track of a set of value sequences $\mathbb{V}$ that
have already been tried for each element in the pattern sequence.

If both the pattern sequence and the value sequence is empty, then we have
successfully finished matching (\textsc{PL-Emp-Both}).
Otherwise if the pattern sequence finishes while there are still values left to match
(\textsc{PL-Emp-Pat}) we fail, producing no bindings.
{\small
\[
  \inference[\textsc{PL-Emp-Both}]{}{\matchall{\varepsilon}{\varepsilon}{\sigma}{\mathbb{V}}{\ttuple{} \quad \otimes}{[]}}
  \quad
  \inference[\textsc{PL-Emp-Pat}]{}{\matchall{\varepsilon}{v,\esequence{v'}}{\sigma}{\mathbb{V}}{\ttuple{} \quad \otimes}{\varepsilon}}
\]
}

\noindent{}The \textsc{PL-More-Pat-Re} rule tries to partition the input sequence of values
into a single $v'$ element that is matched against target pattern $p$, and a
sequence of elements $\esequence{v''}$ that are matched against the rest of the
pattern sequence $\esequence{{\star}p}$. Additionally the rule will try to
backtrack, excluding the selected value $v'$ from the values tried for pattern
$p$; in practice there is only one valid partition for elements of a list, but
for elements of sets there can be many due to commutativity of partitioning.
If all valid partitions have been tried, the $\textsc{PL-More-Pat-Exh}$ rule is
used to end the backtracking for the pattern $p$.

{\small
\[
  \inference[\textsc{PL-More-Pat-Re}]{
    \esequence{v} = v' \otimes \esequence{v''} & v' \notin \mathbb{V}
    & \match{p}{v'}{\sigma}{\esequence{\rho}} \\ 
    \matchall{\esequence{{{\star}p}}}{\esequence{v''}}{\sigma}{\emptyset}{\ttuple{}
      \quad \otimes}{\esequence{\rho'}} & \matchall{p,
      \esequence{{{\star}p}}}{\esequence{v}}{\sigma}{\mathbb{V} \cup \{v'\}}{\ttuple{}
      \quad \otimes}{\esequence{\rho''}}}{\matchall{p,
      \esequence{{{\star}p}}}{\esequence{v}}{\sigma}{\mathbb{V}}{\ttuple{}
      \quad \otimes}{\textrm{merge}(\esequence{\rho}, \esequence{\rho'}), \esequence{\rho''}}}
\]}
{\small
\[
  \inference[\textsc{PL-More-Pat-Exh}]{\nexists v', \esequence{v''}.
    \esequence{v} = v' \otimes \esequence{v''} \wedge v' \notin \mathbb{V} }{\matchall{p,
      \esequence{{{\star}p}}}{\esequence{v}}{\sigma}{\mathbb{V}}{\ttuple{}
      \quad \otimes}{\varepsilon}}
\]}

\noindent{}Like ordinary variables, arbitrary matching patterns depend on whether
the binding variable already exists in the current store.
If the variable is assigned in the current store then either there must
exist a partition of values that has a matching subcollection
(\textsc{PL-More-Star-Uni}), or the matching fails producing any consistent binding environment
(\textsc{PL-More-Star-Pat-Fail}, \textsc{PL-More-Star-Val-Fail}).

{\small
\[
  \inference[\textsc{PL-More-Star-Uni}]{x \in \textrm{dom } \sigma & \sigma(x) =
    \ttuple{\esequence{v'}} & \esequence{v} = \esequence{v'}
    \otimes \esequence{v''} \\
    \matchall{\esequence{{{\star}p}}}{\esequence{v''}}{\sigma}{\emptyset}{\ttuple{}
      \quad \otimes}{\esequence{\rho}}}{\matchall{{\star}x, \esequence{{{\star}p}}}{\esequence{v}}{\sigma}{\mathbb{V}}{\ttuple{} \quad \otimes}{\esequence{\rho}}}
\]

\[
  \inference[\textsc{PL-More-Star-Pat-Fail}]{x \in \textrm{dom } \sigma & \sigma(x) =
    \ttuple{\esequence{v'}} & \nexists \esequence{v''}. \esequence{v} = \esequence{v'}
    \otimes \esequence{v''}}{\matchall{{\star}x, \esequence{{{\star}p}}}{\esequence{v}}{\sigma}{\mathbb{V}}{\ttuple{} \quad \otimes}{\varepsilon}}
\]

\[
  \inference[\textsc{PL-More-Star-Val-Fail}]{x \in \textrm{dom } \sigma &
    \sigma(x) = v & v \neq
    \ttuple{\esequence{v'}}}{\matchall{{\star}x, \esequence{{{\star}p}}}{\esequence{v''}}{\sigma}{\mathbb{V}}{\ttuple{} \quad \otimes}{\varepsilon}}
\]}

\noindent{}If the variable is not in the current store then there are two options: either
\begin{enumerate*}[label=\roman*)]
\item there still exist a partition that is possible to try, or
\item we have exhausted all
possible partitions of the value sequence
\end{enumerate*}. In the first case, an arbitrary
partition is bound to the target variable and the rest of the patterns are
matched against the rest of the values, merging their results; additionally, the
other partitions are also tried concatenating their results with the merged one (\textsc{PL-More-Star-Re}).
In the exhausted case, the pattern match produces no bindings (\textsc{PL-More-Star-Exh}).

{\small\[
  \inference[\textsc{PL-More-Star-Re}]{x \notin \textrm{dom } \sigma &
    \esequence{v} = \esequence{v'} \otimes \esequence{v''} &
    \esequence{v'} \notin \mathbb{V} \\
    \matchall{\esequence{{{\star}p}}}{\esequence{v''}}{\sigma}{\emptyset}{\ttuple{} \quad
      \otimes}{\esequence{\rho}} &
    \matchall{{\star}x,
      \esequence{{{\star}p}}}{\esequence{v}}{\sigma}{\mathbb{V} \cup
      \{\esequence{v'}\}}{\ttuple{} \quad \otimes}{\esequence{\rho'}}
  }{\matchall{{\star}x,
      \esequence{{{\star}p}}}{\esequence{v}}{\sigma}{\mathbb{V}}{\ttuple{}
      \quad \otimes}{\textrm{merge}([x \mapsto \ttuple{\esequence{v'}}], \esequence{\rho}), \esequence{\rho'}}}
\]

\[
  \inference[\textsc{PL-More-Star-Exh}]{x \notin \textrm{dom } \sigma &
    \nexists \esequence{v'}, \esequence{v''}. \esequence{v} =
    \esequence{v'} \otimes \esequence{v''} \wedge \esequence{v'}
    \notin \mathbb{V} }{\matchall{{\star}x, \esequence{{{\star}p}}}
    {\esequence{v}}{\sigma}{\mathbb{V}}{\ttuple{} \quad \otimes}{\varepsilon}}
\]}

\noindent{}Merging a sequence of possible variable bindings produces a sequence containing
consistent variable bindings from the sequence: that is, all possible
environments that assign consistent values to the same variables are merged.
{\small\begin{align*}
  \textrm{merge}(\varepsilon) &= [] \\
  \textrm{merge}(\esequence{\rho}, \esequence{\rho_1'}, \dots, \esequence{\rho_{\mathrm{n}}'}) &= \textrm{merge-pairs}(\esequence{\rho} \times \textrm{merge}(\esequence{\rho_1'}, \dots, \esequence{\rho_{\mathrm{n}}'}))
\end{align*}%
\begin{align*}
  \textrm{merge-pairs}(\ttuple{\rho_1, \rho_1'}, \dots, \ttuple{\rho_{\mathrm{n}}, \rho_{\mathrm{n}}'}) &= \textrm{merge-pair}(\rho_1, \rho_1'), \dots, \textrm{merge-pair}(\rho_{\mathrm{n}}, \rho_{\mathrm{n}}')
\end{align*}%
\begin{align*}
  \textrm{merge-pair}(\rho, \rho') &=
    \begin{cases}
    \rho \rho' & \textbf{if } \forall x \in \textrm{dom } \; \rho \cap \textrm{dom } \; \rho'. \rho(x) = \rho'(x) \\
    \varepsilon & \textbf{otherwise}
    \end{cases}
\end{align*}}

\section{Semantics Properties}
\label{sec:theorems}

Backtracking is pure in Rascal (Light) programs, and so if evaluating a set
cases produces $\keyword{fail}$ as result, the initial state is restored.

\begin{restatable}[Backtracking purity]{theorem}{backtrackingpurity}
  If $\begin{array}{c} \mathcal{C\!S} \\
        \evalcases{\esequence{\mathit{cs}}}{v}{\sigma}{\keyword{fail}}{\sigma'} \end{array}$ then $\sigma' = \sigma$
\end{restatable}

\noindent{}Strong typing is an important safety property of Rascal, which I capture by
specifying a theorems with two result properties: one about the well-typedness of state, and the other
about well-typedness of resulting values.

\begin{restatable}[Strong typing]{theorem}{strongtyping}
\label{thm:strongtyping}
Assume that semantic unary $\sem{\ominus}$ and binary operators $\sem{\oplus}$
  are strongly typed.
  If $\begin{array}{c} \mathcal{E} \\
\evalexpr{e}{\sigma}{\mathit{vres}}{\sigma'} \end{array}$ and there
    exists a type $t$ such that
  $ \begin{array}{c} \mathcal{T} \\ \typing{v}{t} \end{array} $ for each value in the input store $v \in
  \mathrm{img} \; \sigma$, then
  \begin{enumerate}
  \item There exists a type $t'$ such that
  $\typing{v'}{t'}$ for each value in the result store $v' \in \mathrm{img} \; \sigma'$.
\item If the result value $\mathit{vres}$ is either $\keyword{success} \; v''$, $\keyword{return}
  \; v''$, or $\keyword{throw}
  \; v''$, then there exists a type $t''$ such that
  $\typing{v''}{t''}$.
  \end{enumerate}
\end{restatable}

\noindent{}Consider an augmented version of the operational semantics where each execution
relation is annotated with \textit{partiality fuel}\,\citep{DBLP:conf/popl/AminR17}---represented by a
superscript natural number $n$---which specifies the maximal number of recursive premises
allowed in the derivation. The fuel is subtracted on each recursion,
resulting in a $\keyword{timeout}$ value when it reaches zero, and the rule set is
amended by congruence rules which propagate $\keyword{timeout}$ results from
the recursive premises to the conclusion.
{\small\begin{align*}
  \mathit{vtres} &\Coloneqq \mathit{vres} \mid \keyword{timeout}
\end{align*}}

\noindent{}For this version of the semantics, we can specify the property that
execution will either produce a result or it will timeout; that is, the semantics does
not get stuck.

\begin{restatable}[Partial progress]{theorem}{partialprogress}
   \label{thm:partialprogress}
   It is possible to construct a derivation
   $\evalexprfuel{e}{\sigma}{\mathit{vtres}}{\sigma'}{n}$ for
    any input expression $e$, well-typed store $\sigma$ and fuel $n$.
\end{restatable}

\noindent{}Finally, consider a subset of the expression language described by syntactic category $e_{\mathrm{fin}}$, where
$\keyword{while}$-loops, $\keyword{solve}$-loops and function calls are disallowed, and similarly with
all traversal strategies except $\keyword{bottom-up}$ and $\keyword{bottom-up-break}$.
This subset is known to be terminating:

\begin{restatable}[Terminating expressions]{theorem}{termination}
   \label{thm:termination}
  There exists $n$ such that derivation
  $\begin{array}{c} \mathcal{E}
     \\ \evalexprfuel{e_{\mathrm{fin}}}{\sigma}{\mathit{vres}}{\sigma'}{n} \end{array}$ has a result $\mathit{vres}$ which is not $\keyword{timeout}$ for expression  $e_{\mathrm{fin}}$ in the terminating subset.
\end{restatable}

\begin{remark}
  Why is the $\keyword{top-down}$ traversal strategy potentially non-terminating
  while the $\keyword{bottom-up}$ traversal strategy is terminating? The answer
  lies in that the $\keyword{top-down}$ traverses the children of the target
  expression after evaluating the required case, which makes it possible to keep
  add children ad infinitum as illustrated by the following example:
  \begin{minipage}{1.0\linewidth}
  \begin{lstlisting}[language=Rascal]
    data Nat = zero() | succ(Nat pred);

    Nat infincrement(Nat n) = 
      top-down visit(n) {
         case succ(n) => succ(succ(n))
      }
  \end{lstlisting}
  \end{minipage}
\end{remark}

\section{Formal Semantics, Types and Intented Semantics}
This formal specification of Rascal Light was largely performed when the
official type checker of
Rascal (developed by CWI Amsterdam) was only at the experimental stage. Useful extension of this
formalization include a deductive type system that includes
polymorphic aspects of Rascal to prove the consistency of the type system, and
type safety with regards to the dynamic semantics provided in this formalization.

The presented semantics was checked against the Rascal implementation in various
ways: the existing implementation was
used as a reference point, rules for individual were held up
against the documentation, and correspondence with the Rascal developers was used
to clarify remaining questions. 
A more formal way to check whether the captured semantics is the
intended one, is to construct an independent formal semantics which is
related to the natural semantics presented in this report. This could be an axiomatic semantics\,\cite{DBLP:journals/cacm/Hoare69}, which modularly specifies for each construct
the intented effects using logical formulae as pre and post-conditions, without
necessarily specifying how each construct is evaluated.
The natural semantics would then be checked to satisfy the axiomatic semantics
for each construct, which will further increase confidence that the captured
semantics is the intented one.

\section{Conclusion}
I presented the formalization of a large subset of the operational part of
Rascal\,\cite{DBLP:conf/scam/KlintSV09,Klint2011}, called Rascal Light.
The formalization was primarily based on the available open source
implementation\footnote{https://github.com/usethesource/rascal} and the
accompanying documentation\footnote{http://tutor.rascal-mpl.org/}, and personal
correspondence with the developers further clarified previous ambiguities and mismatches.

\section*{Acknowledgments}
 This report is partially funded by by Danish Council for Independent Research,
 grant no. 0602-02327B. I would like to thank Paul Klint and the Rascal
 team at CWI Amsterdam for answering my questions about Rascal and provide
 corrections to this report. I would also like to thank Aleksandar Dimovski and
 Andrzej W\k{a}sowski for
 proofreading and commenting on this report.

\bibliography{main.bib}

\appendix

\section{Semantics Properties Proofs}
\label{sec:proofs}

\backtrackingpurity*
\begin{proof}
  By induction on the derivation $\mathcal{C\!S}$:
  \begin{itemize}
  \item Case $\mathcal{C\!S} =
    {\footnotesize\inference[\textsc{ECS-Emp}]{}{\evalcases{\varepsilon}{v}{\sigma}{\keyword{fail}}{\sigma}}}$,
    so $\esequence{\mathit{cs}} = \varepsilon$ and $\sigma' = \sigma$. Holds by definition.
  \item Case $\mathcal{C\!S} =
    {\footnotesize\inference[\textsc{ECS-More-Fail}]{
    {\begin{array}{c}\mathcal{M} \\ \match{p}{v}{\sigma}{\esequence{\rho}}\end{array}}
    & {\begin{array}{c}\mathcal{C} \\ \evalcase{\esequence{\rho}}{e}{\sigma}{\keyword{fail}}{\sigma''}\end{array}} \\
    {\begin{array}{c} \mathcal{C\!S}'   \\ \evalcases{\esequence{\mathit{cs}'}}{v}{\sigma}{\keyword{fail}}{\sigma'} \end{array}}
  }{\evalcases{\keyword{case} \;
      p => e, \esequence{\mathit{cs}'}}{v}{\sigma}{\keyword{fail}}{\sigma'}}}$,
    so $\esequence{\mathit{cs}} = (\keyword{case} \; p => e,
    \esequence{\mathit{cs}'})$ and $\mathit{vres} = \keyword{fail}$. 
    \begin{itemize}
    \item By inductive hypothesis of $\mathcal{C\!S}'$ we get $\sigma' = \sigma$.
    \end{itemize}
  \item Note: the rule \textsc{ECS-More-Ord} is inapplicable since its premise
    states that the result value  $\mathit{vres} \neq \keyword{fail}$.
  \end{itemize}
\end{proof}

In order to prove \Cref{thm:strongtyping}, we need to state some helper lemmas
about sub-derivations.

We have a lemma for the auxiliary $\textrm{merge}$ function:
\begin{lemma}
  \label{lem:mergetyping}
  If we have $\esequence{\rho'} = \textrm{\textnormal{merge}}(\esequence{\rho_1}, \dots, \esequence{\rho_{\mathrm{n}}})$ and for each value $v \in \mathrm{img} \; \rho_{i,j}$ in an
  environment in the input sequence of
  environment sequences $\esequence{\rho_1}, \dots, \esequence{\rho_{\mathrm{n}}}$ there exists a type $t$ so that
  $\typing{v}{t}$, then we have that for each value $v' \in \mathrm{img} \; \rho'_k$ in an environment
  in the resulting environment sequence $\esequence{\rho'}$ we have a type $t'$
  so that $\typing{v'}{t'}$
\end{lemma}
\begin{proof}
  Follows directly from the premises by induction of the input sequence of
  environment sequences $\esequence{\rho_1}, \dots, \esequence{\rho_{\mathrm{n}}}$
\end{proof}

We have a lemma for the auxiliary $\textrm{children}$ function:
\begin{lemma}
  \label{lem:childrentyping}
  For a value $v$ such that we have $\esequence{v'} = \textrm{children}(v)$ and
  $\begin{array}{c}\mathcal{T}\\\typing{v}{t}\end{array}$, then there exists a type sequence $\esequence{t'}$ such that $\esequence{\typing{v'}{t'}}$
\end{lemma}
\begin{proof}
  By induction on the syntax of $v$:
  \begin{itemize}
  \item Cases $v = \mathit{vb}$ and $v = \blacksquare$, so $\esequence{v'} = \varepsilon$. Holds trivially.
  \item Case $v = k(\esequence{v'})$
    \begin{itemize}
    \item 
    By inversion we get
    \[
      {\footnotesize\inference[\textsc{T-Constructor}]{\keyword{data} \; \mathit{at} = \dots ~|~
    k(\esequence{t''}) ~|~ \dots &
    \esequence{{\begin{array}{c} \mathcal{T}' \\ \typing{v}{t'} \end{array}}} &
    \esequence{{\begin{array}{c} \mathcal{S\!T}' \\
                  \subtyping{t'}{t''} \end{array}}}}{\typing{k(\esequence{v'})}{\mathit{at}}}}\] so $t = \mathit{at}$
\item Here, $\esequence{\mathcal{T}'}$ exactly represents our target goal
    \end{itemize}
  \item Case $v = [\esequence{v'}]$
    \begin{itemize}
    \item By inversion we get
      $\mathcal{T} =  {\footnotesize\inference[\textsc{T-List}]{
        \esequence{{\begin{array}{c}
                    \mathcal{T}' \\
          \typing{v'}{t'}
        \end{array}}}}{\typing{[\esequence{v'}]}{\mathrm{list}\ttuple{\bigsqcup\esequence{t'}}}}}
      $, so $t = \mathrm{list}\ttuple{\bigsqcup\esequence{t'}}$
    \item Here, $\esequence{\mathcal{T}'}$ exactly represents our target goal
    \end{itemize}
  \item Case $v = \{\esequence{v'}\}$
    \begin{itemize}
    \item By inversion we get
      $\mathcal{T} =  {\footnotesize\inference[\textsc{T-Set}]{
        \esequence{{\begin{array}{c}
                    \mathcal{T}' \\
          \typing{v'}{t'}
        \end{array}}}}{\typing{\{\esequence{v'}\}}{\mathrm{set}\ttuple{\bigsqcup\esequence{t'}}}}}
      $, so $t = \mathrm{set}\ttuple{\bigsqcup\esequence{t'}}$
    \item Here, $\esequence{\mathcal{T}'}$ exactly represents our target goal
    \end{itemize}
  \item Case $v = (\esequence{v'' : v'''})$, so $\esequence{v'} = \esequence{v''},\esequence{v'''}$
    \begin{itemize}
    \item By inversion we get \[\mathcal{T} =
        {\footnotesize \inference[\textsc{T-Map}]{
          {\begin{array}{c} \mathcal{T''} \\ \esequence{\typing{v''}{t''}} \end{array}} & 
         {\begin{array}{c} \mathcal{T'''} \\ \esequence{\typing{v'''}{t'''}} \end{array}}
        }{\typing{(\esequence{v'' : v'''})}{\mathrm{map}\ttuple{\bigsqcup\esequence{t''},\bigsqcup\esequence{t'''}}}}
        }
\]
\item Here we take the concatenation of the premises $\mathcal{T''},
  \mathcal{T'''}$ to fulfill our goal.
    \end{itemize}
  \end{itemize}
\end{proof}

We have two mutually recursive lemmas on pattern matching:
\Cref{lem:patternmatchtyping} and \Cref{lem:patternmatchalltyping}.
\begin{lemma}
\label{lem:patternmatchtyping}
  If $\begin{array}{c} \mathcal{M} \\
        \match{p}{v}{\sigma}{\esequence{\rho}} \end{array}$ and there exists a
      type $t$ such that $\begin{array}{c} \mathcal{T} \\ \typing{v}{t} \end{array}$ and a type $t'$ such that
      $\typing{v'}{t'}$ for each value in the input store $v' \in \mathrm{img}
      \; \sigma$, then we have a type $t''$ for each value $v'' \in \mathrm{img}
      \; \rho_i$  in an
      environment in the output sequence $\esequence{\rho}$.
\end{lemma}
\begin{proof}
  By induction on the derivation $\mathcal{M}$:
  \begin{itemize}
  \item Cases $\mathcal{M} = \textsc{P-Val-Sucs}$,  $\mathcal{M} =
    \textsc{P-Val-Fail}$, $\mathcal{M} = \textsc{P-Var-Uni}$, $\mathcal{M} =
    \textsc{P-Var-Fail}$, $\mathcal{M} = \textsc{P-Cons-Fail}$, $\mathcal{M} =
    \textsc{P-Type-Fail}$, $\mathcal{M} =
    \textsc{P-List-Fail}$, $\mathcal{M} =
    \textsc{P-Set-Fail}$, $\mathcal{M} =
    \textsc{P-Neg-Sucs}$, $\mathcal{M} =
    \textsc{P-Neg-Fail}$ hold trivially since we have that the output
    environment sequence $\esequence{\rho} = []$ or $\esequence{\rho} =
    \varepsilon$, in both cases containing no values.
  \item Case $\mathcal{M} = \textsc{P-Var-Bind}$ holds by the premise derivation $\mathcal{T}$.
  \item Case $\mathcal{M} = {\footnotesize\inference[\textsc{P-Cons-Sucs}]{
    {\begin{array}{c} \mathcal{M}_1' \\ \match{p_1'}{v_1'}{\sigma}{\esequence{\rho_1'}} \end{array}} \;\dots\;
    {\begin{array}{c} \mathcal{M}_{\mathrm{n}}' \\
       \match{p_{\mathrm{n}}'}{v_{\mathrm{n}}'}{\sigma}{\esequence{\rho_{\mathrm{n}}'}}
     \end{array}}}{\match{k(\esequence{p'})}{k(\esequence{v'})}{\sigma}{\textrm{\textnormal{merge}}(\esequence{\rho_1'},
      \dots, \esequence{\rho_{\mathrm{n}}'})}}}$, so $p = k(\esequence{p'})$, $v
  = k(\esequence{v'})$ and $\esequence{\rho} = \textrm{\textnormal{merge}}(\esequence{\rho_1'}, \dots,
  \esequence{\rho_{\mathrm{n}}'})$
  \begin{itemize}
  \item By inversion we get \[\mathcal{T} =   {\footnotesize\inference[\textsc{T-Constructor}]{\keyword{data} \; \mathit{at} = \dots ~|~
    k(\esequence{t''}) ~|~ \dots &
    \esequence{{\begin{array}{c} \mathcal{T}' \\ \typing{v}{t'} \end{array}}} &
    \esequence{{\begin{array}{c} \mathcal{S\!T}' \\ \subtyping{t'}{t''} \end{array}}}}{\typing{k(v_1',
      \dots, v_{\mathrm{n}}')}{\mathit{at}}}}\]
so $t = \mathit{at}$
\item Now by induction hypotheses of $\esequence{\mathcal{M}'}$ using
  $\esequence{\mathcal{T}'}$, and then using \Cref{lem:mergetyping} we get that
  $\esequence{\rho}$ has well-typed values.
  \end{itemize}
  \item Case $\mathcal{M} = \textsc{P-Type-Sucs}$ holds by the premise
    derivation $\mathcal{T}$ and \Cref{lem:mergetyping}.
  \item Case $\mathcal{M} =  {\footnotesize\inference[\textsc{P-List-Sucs}]{
      {\begin{array}{c} \mathcal{M\!S}' \\ \matchall{\esequence{{{\star}p'}}}{\esequence{v'}}{\sigma}{\emptyset}{[] \quad ,}{\esequence{\rho}} \end{array}}
    }{\match{[\esequence{{{\star}p'}}]}{[\esequence{v'}]}{\sigma}{\esequence{\rho}}}}$,
    so $p = [\esequence{{{\star}p'}}]$ and $v = [\esequence{v'}]$

    \begin{itemize}
    \item By inversion we get
      $\mathcal{T} =  {\footnotesize\inference[\textsc{T-List}]{
        \esequence{{\begin{array}{c}
                    \mathcal{T}' \\
          \typing{v'}{t'}
        \end{array}}}}{\typing{[\esequence{v'}]}{\mathrm{list}\ttuple{\bigsqcup\esequence{t'}}}}}
      $, so $t = \mathrm{list}\ttuple{\bigsqcup\esequence{t'}}$
    \item Partitioning lists by splitting on concatenation $,$ preserves typing
      since we can for each value pick the corresponding type derivation in the
      sequence.
    \item By induction hypothesis (using \Cref{lem:patternmatchalltyping}) of
      $\mathcal{M\!S}'$ using $\esequence{\mathcal{T}'}$ and above fact we get that
      $\esequence{\rho}$ has well-typed values.
    \end{itemize}
  \item Case $\mathcal{M} =  {\footnotesize\inference[\textsc{P-Set-Sucs}]{
      {\begin{array}{c} \mathcal{M\!S}' \\ \matchall{\esequence{{{\star}p'}}}{\esequence{v'}}{\sigma}{\emptyset}{[] \quad ,}{\esequence{\rho}} \end{array}}
    }{\match{\{\esequence{{{\star}p'}}\}}{\{\esequence{v'}\}}{\sigma}{\esequence{\rho}}}}$,
   so $p = \{\esequence{{{\star}p'}}\}$ and $v = \{\esequence{v'}\}$
   \begin{itemize}
   \item By inversion we get
      $
      \mathcal{T} =  {\footnotesize\inference[\textsc{T-Set}]{
        \esequence{{\begin{array}{c}
                    \mathcal{T}' \\
          \typing{v'}{t'}
        \end{array}}}}{\typing{\{\esequence{v'}\}}{\mathrm{set}\ttuple{\bigsqcup\esequence{t'}}}}}
      $, so $t = \mathrm{set}\ttuple{\bigsqcup\esequence{t'}}$
    \item Partitioning sets using disjoint union $\uplus$ preserves typing
      of value sub-sequences, since we can for each value $v_i$ in a subsequence
      pick the corresponding typing derivation $\mathcal{T}'_i$
    \item By induction hypothesis (using \Cref{lem:patternmatchalltyping}) of
      $\mathcal{M\!S}'$ using $\esequence{\mathcal{T}'}$ and above fact we get that
      $\esequence{\rho}$ has well-typed values.
    \end{itemize}
  \item Case $\mathcal{M} =  {\footnotesize\inference[\textsc{P-Deep}]{
      {\begin{array}{c} \mathcal{M}' \\ \match{p'}{v}{\sigma}{\esequence{\rho'}} \end{array}} &
      v'_{1},\dots,v'_{\mathrm{n}} = \textrm{\textnormal{children}}(v) \\
   {\begin{array}{c} \mathcal{M}_1'' \\
      \match{/p'}{v'_1}{\sigma}{\esequence{\rho_1''}} \end{array}} \;\dots\;
{\begin{array}{c} \mathcal{M}_{\mathrm{n}}'' \\
   \match{/p'}{v'_{\mathrm{n}}}{\sigma}{\esequence{\rho_{\mathrm{n}}''}} \end{array}}
}{\match{/p'}{v}{\sigma}{\esequence{\rho'},
    \esequence{\rho_1''}, \dots, \esequence{\rho_{\mathrm{n}}''}}}}$
  \begin{itemize}
  \item By induction hypothesis on $\mathcal{M}'$, we get that
    $\esequence{\rho'}$ has well-typed values
  \item By using \Cref{lem:childrentyping} on $\esequence{v'}$, we get $\esequence{\begin{array}{c} \mathcal{T}' \\ \typing{v'}{t'} \end{array}}$
  \item By induction hypotheses on $\esequence{\mathcal{M}''}$ using
    $\mathcal{T}'$ we get that $\esequence{\rho_1''}, \dots,
    \esequence{\rho_{\mathrm{n}}''}$ is well-typed
  \item Now, we can show that $\esequence{\rho'}, \esequence{\rho_1''}, \dots,
    \esequence{\rho_{\mathrm{n}}''}$ is well-typed using above facts
  \end{itemize}
  \end{itemize}
\end{proof}

\begin{lemma}
\label{lem:patternmatchalltyping}
  If $\begin{array}{c} \mathcal{M\!S} \\
        \matchall{\esequence{{{\star}p}}}{\esequence{\mathit{v}}}{\sigma}{\mathbb{V}}{\ttuple{}
        \quad \otimes}{\esequence{\rho}} \end{array}$ and the following
    properties hold:
    \begin{enumerate}
    \item There exists a
      type sequence $\esequence{t}$ such that $\esequence{\typing{v}{t}}$
    \item There exists a type $t'$ such that
      $\typing{v'}{t'}$ for each value in the input store $v' \in \mathrm{img}
      \; \sigma$
    \item There exists a type $t'$ such that $\typing{v'}{t'}$ for each value
      in the visited value set $v' \in \mathbb{V}$
    \item \label{it:pmalltypprop4} If for all value sequences $\esequence{v'}, \esequence{v''},
      \esequence{v'''}$ where we have $\esequence{v'} = \esequence{v''} \otimes
      \esequence{v'''}$ and there exists a type sequence $\esequence{t'}$ so
      $\esequence{\typing{v'}{t'}}$, then we have that there exists two type
      sequences $\esequence{t''}$ and $\esequence{t'''}$ so that
      $\esequence{\typing{v''}{t''}}$ and $\esequence{\typing{v'''}{t'''}}$
    \end{enumerate}
    Then we have a type $t'$ so that $\typing{v'}{t'}$ for each value $v' \in \mathrm{img} \; \rho_i$  in an
      environment in the output sequence $\esequence{\rho}$.
\end{lemma}
\begin{proof}
  By induction on the derivation $\mathcal{M\!S}$:
  \begin{itemize}
  \item Cases $\mathcal{M\!S} = \textsc{PL-Emp-Both}$, $\mathcal{M\!S} =
    \textsc{PL-Emp-Pat}$, $\mathcal{M\!S} =
    \textsc{PL-More-Pat-Exh}$, $\mathcal{M\!S} =
    \textsc{PL-More-Star-Pat-Fail}$, $\mathcal{M\!S} =
    \textsc{PL-More-Star-Val-Fail}$, and $\mathcal{M\!S} = \textsc{PL-More-Star-Exh}$
    hold trivially since we have
    $\esequence{\rho} = []$ or $\esequence{\rho} = \varepsilon$, and there are
    no values in either $[]$ or $\varepsilon$.
  \item Case $\mathcal{M\!S} = \textsc{PL-More-Pat-Re}$ holds by induction
    hypotheses (including \Cref{lem:patternmatchtyping}), and \Cref{lem:mergetyping}.
  \item Case $\mathcal{M\!S} = {\footnotesize\inference[\textsc{PL-More-Star-Uni}]{x \in \textrm{dom } \sigma & \sigma(x) =
    \ttuple{\esequence{v'}} & \esequence{v} = \esequence{v'}
    \otimes \esequence{v''} \\
    {\begin{array}{c} \mathcal{M\!S}' \\ \matchall{\esequence{{{\star}p'}}}{\esequence{v''}}{\sigma}{\emptyset}{\ttuple{}
      \quad \otimes}{\esequence{\rho}}\end{array}}
  }{\matchall{{\star}x,
      \esequence{{{\star}p'}}}{\esequence{v}}{\sigma}{\mathbb{V}}{\ttuple{} \quad
      \otimes}{\esequence{\rho}}}}$, so $\esequence{{{\star}p}} = {\star}x,
  \esequence{{{\star}p'}}$.

  By property \ref{it:pmalltypprop4} we can derive that there exists a type sequence
  $\esequence{t''}$ such that $\esequence{\typing{v''}{t''}}$.
  Then we can apply induction hypothesis on derivation $\mathcal{M\!S}'$, and get our
  target result.
    \item Case $\mathcal{M\!S} = \textsc{PL-More-Star-Re}$ holds by induction
    hypotheses (including \Cref{lem:patternmatchtyping}), and \Cref{lem:mergetyping}.
  \end{itemize}
\end{proof}

We have a lemma on the auxiliary function $\textrm{if-fail}$:
\begin{lemma}
  \label{lem:iffailtyping}
  If we have $v' = \textrm{\textnormal{if-fail}}(\mathit{vfres}, \mathit{v''})$, and:
  \begin{enumerate}
  \item If $\mathit{vfres} = \keyword{success} \; v$, then we have
    $\typing{v}{t}$ for some type $t$
  \item We have $\typing{v''}{t''}$ for somet type $t''$
  \end{enumerate}
  Then there exists a type $t'$ such that $\typing{v'}{t'}$
\end{lemma}
\begin{proof}
  Straightforwardly by case analysis on $\mathit{vfres}$
\end{proof}

Similarly, we have a lemma on the auxiliary function $\textrm{vcombine}$
\begin{lemma}
  \label{lem:vcombinetyping}
  If we have $\mathit{vfres}{\star}'' =
  \textrm{\textnormal{vcombine}}(\mathit{vfres}, \mathit{vfres}{\star}', v, \esequence{v'})$, and:
  \begin{enumerate}
  \item If $\mathit{vfres} = \keyword{success} \; v'''$, then we have
    $\typing{v'''}{t'''}$ for some type $t$
  \item If $\mathit{vfres}{\star}' = \keyword{success} \; \esequence{v''''}$,
    then we $\esequence{\typing{v''''}{t''''}}$ for some typing sequence $\esequence{t''''}$
  \item We have $\typing{v}{t}$ for somet type $t$
  \item We have $\esequence{\typing{v'}{t'}}$ for somet type $t'$
  \end{enumerate}
  Then if $\mathit{vfres}{\star}'' = \keyword{success} \; \esequence{v''}$ there
  exists a type sequence $\esequence{t''}$ such that $\esequence{\typing{v''}{t''}}$
\end{lemma}
\begin{proof}
  Straightforwardly by case analysis on $\mathit{vfres}$ and
  $\mathit{vfres}{\star}'$ and using \Cref{lem:iffailtyping}.
\end{proof}

We have a lemma on the reconstruct derivation:
\begin{lemma}
  \label{lem:reconstructtyping}
  If we have $
  \begin{array}{c}
   \mathcal{R\!C} \\
  \reconstruct{v}{\esequence{v'}}{\mathit{rcres}} 
  \end{array}$, $\typing{v}{t}$
  for some type $t$, and \esequence{\typing{v'}{t'}} for some type sequence
  $\esequence{t'}$, then when $\mathit{rcres} = \keyword{success} \; v''$ there exists a type $t''$ such that if $\typing{v''}{t''}$
\end{lemma}
\begin{proof}
  By induction on the derivation $\mathcal{R\!C}$:
  \begin{itemize}
  \item Cases $\mathcal{R\!C} = \textsc{RC-Val-Err}$, $\mathcal{R\!C} =
    \textsc{RC-Cons-Err}$, $\mathcal{R\!C} =
    \textsc{RC-List-Err}$, $\mathcal{R\!C} =
    \textsc{RC-Set-Err}$, $\mathcal{R\!C} =
    \textsc{RC-Map-Err}$, $\mathcal{R\!C} =
    \textsc{RC-Bot-Err}$ hold trivially.
  \item Case $\mathcal{R\!C} = \textsc{RC-Val-Sucs}$ holds using the premises.
  \item Case $\mathcal{R\!C} = \textsc{RC-Cons-Sucs}$ holds using the premises
    and rule \textsc{T-Constructor}.
  \item Case $\mathcal{R\!C} = \textsc{RC-List-Sucs}$ holds using the premises
    and rule \textsc{T-List}.
  \item Case $\mathcal{R\!C} = \textsc{RC-Set-Sucs}$ holds using the premises
    and rule \textsc{T-Set}.
  \item Case $\mathcal{R\!C} = \textsc{RC-Map-Sucs}$ holds using the premises
    and rule \textsc{T-Map}.
  \item Case $\mathcal{R\!C} = \textsc{RC-Bot-Sucs}$ holds using the premises
    and rule \textsc{T-Bot}.
  \end{itemize}
\end{proof}

We now have a series of mutually inductive lemmas with our
\Cref{thm:strongtyping}, since the operational semantics rules are mutually
inductive themselves. The lemmas are \Cref{lem:evalstartyping},
\Cref{lem:evaleachtyping},  \Cref{lem:evalgentyping}, \Cref{lem:evalcasetyping},
\Cref{lem:evalcasestyping}, \Cref{lem:evalvisityping},
\Cref{lem:evaltdvisittyping}, \Cref{lem:evaltdvisitstartyping},
\Cref{lem:evalbuvisittyping}, and \Cref{lem:evalbuvisitstartyping}.

\begin{lemma}
  \label{lem:evalstartyping}
  If $\begin{array}{c} \mathcal{E\!S} \\
        \evalexprstar{\esequence{e}}{\sigma}{{\mathit{vres}\!{\star}}}{\sigma'} \end{array}$ and there
    exists a type $t$ such that
  $\typing{v}{t}$ for each value in the input store $v \in
  \mathrm{img} \; \sigma$, then
  \begin{enumerate}
  \item There exists a type $t'$ such that
  $\typing{v'}{t'}$ for each value in the result store $v' \in \mathrm{img} \; \sigma'$.
\item If the result value ${\mathit{vres}\!{\star}}$ is
  $\keyword{success} \; \esequence{v''}$, then there exists a type sequence $\esequence{t''}$ such that
  $\esequence{\typing{v''}{t''}}$.
\item If the result value  ${\mathit{vres}\!{\star}}$ is either
  $\keyword{return}
  \; v''$, or $\keyword{throw}
  \; v''$, then there exists a type $t''$ such that $\typing{v''}{t''}$
\end{enumerate}
\end{lemma}
\begin{proof}
  By induction on the derivation $\mathcal{E\!S}$:
  \begin{itemize}
  \item Case $\mathcal{E\!S} = \textsc{ES-Emp}$ holds directly using premises.
  \item Cases $\mathcal{E\!S} = \textsc{ES-More}$, $\mathcal{E\!S} =
    \textsc{ES-Exc1}$ and $\mathcal{E\!S} = \textsc{ES-Exc2}$ hold directly from
    the induction hypotheses (including the one given by \Cref{thm:strongtyping}).
  \end{itemize}
\end{proof}

\begin{lemma}
  \label{lem:evaleachtyping}
  If $\begin{array}{c} \mathcal{E\!E} \\
        \evaleach{e}{\esequence{\rho}}{\sigma}{\mathit{vres}}{\sigma'} \end{array}$,
      there
    exists a type $t$ such that
  $\typing{v}{t}$ for each value in the input store $v \in
  \mathrm{img} \; \sigma$, and there exists a type $t$ such that $\typing{v}{t}$
  for each value in an environment $v \in \mathrm{img} \; \rho_i$ in the environment sequence $\esequence{\rho}$
  then
  \begin{enumerate}
  \item There exists a type $t'$ such that
  $\typing{v'}{t'}$ for each value in the result store $v' \in \mathrm{img} \; \sigma'$.
\item If the result value  $\mathit{vres}$ is either
   $\keyword{success}
  \; v''$,
  $\keyword{return}
  \; v''$, or $\keyword{throw}
  \; v''$, then there exists a type $t''$ such that $\typing{v''}{t''}$
\end{enumerate}
\end{lemma}
\begin{proof}
  By induction on the derivation $\mathcal{E\!E}$:
  \begin{itemize}
  \item Case $\mathcal{E\!E} = \textsc{EE-Emp}$ holds directly using the premises.
  \item Cases $\mathcal{E\!E} = \textsc{EE-More-Sucs}$,  $\mathcal{E\!E} =
    \textsc{EE-More-Break}$ and  $\mathcal{E\!E} = \textsc{EE-More-Exc}$ hold
    directly from the induction hypotheses (including \Cref{thm:strongtyping}),
    and the (trivial) facts that if any store $\sigma$ is well-typed then $\sigma \setminus
    \mathbb{X}$ is well-typed for any set of variables $\mathbb{X}$ and $\sigma
    \rho$ is well-typed for any well-typed environment $\rho$.
  \end{itemize}
\end{proof}

\begin{lemma}
  \label{lem:evalgentyping}
  If $\begin{array}{c} \mathcal{G} \\
        \evalgenexpr{g}{\sigma}{\mathit{envres}}{\sigma'} \end{array}$,
      and there
    exists a type $t$ such that
  $\typing{v}{t}$ for each value in the input store $v \in
  \mathrm{img} \; \sigma$ then
  \begin{enumerate}
  \item There exists a type $t'$ such that
  $\typing{v'}{t'}$ for each value in the result store $v' \in \mathrm{img} \; \sigma'$.
\item If the result value  $\mathit{envres}$ is 
   $\keyword{success}
   \; \esequence{\rho}$, then there exists a type $t''$ for each value in
   an environment $v'' \in \rho_i$ in the environment sequence
   $\esequence{\rho}$ such that $\typing{v''}{t''}$
\item If the result value  $\mathit{envres}$ is either
  $\keyword{return}
  \; v'''$, or $\keyword{throw}
  \; v'''$, then there exists a type $t'''$ such that $\typing{v'''}{t'''}$
\end{enumerate}
\end{lemma}
\begin{proof}
  By induction on the derivation $\mathcal{G}$:
  \begin{itemize}
  \item Cases $\mathcal{G} = \textsc{G-Pat-Sucs}$ and $\mathcal{G} =
    \textsc{G-Pat-Exc}$ hold directly by induction hypothesis given by
    \Cref{thm:strongtyping}, and then using \Cref{lem:patternmatchtyping} if necessary.
  \item Cases $\mathcal{G} = \textsc{G-Enum-List}$, $\mathcal{G} =
    \textsc{G-Enum-Set}$ and  $\mathcal{G} = \textsc{G-Enum-Map}$ hold by the
    induction hypothesis given by \Cref{thm:strongtyping} and then using
    inversion on the type derivation of the result collection value to extract
    the type derivations of the contained values in the result environments
    (from rules $\textsc{T-List}$, $\textsc{T-Set}$, and $\textsc{T-Map}$ respectively).
  \item Cases $\mathcal{G} = \textsc{G-Enum-Err}$ and $\mathcal{G} =
    \textsc{G-Enum-Exc}$ hold directly by the induction hypothesis given by \Cref{thm:strongtyping}.
  \end{itemize}
\end{proof}

\begin{lemma}
    \label{lem:evalcasetyping}
  If $\begin{array}{c} \mathcal{C} \\
        \evalcase{\esequence{\rho}}{e}{\sigma}{\mathit{vres}}{\sigma'} \end{array}$, there exists a type $t$ such that $\typing{v}{t}$
  for each value in an environment $v \in \mathrm{img} \; \rho_i$ in the
  environment sequence $\esequence{\rho}$, and
      there
    exists a type $t$ such that
  $\typing{v}{t}$ for each value in the input store $v \in
  \mathrm{img} \; \sigma$, then
  \begin{enumerate}
  \item There exists a type $t'$ such that
  $\typing{v'}{t'}$ for each value in the result store $v' \in \mathrm{img} \; \sigma'$.
\item If the result value  $\mathit{vres}$ is either
   $\keyword{success}
  \; v''$,
  $\keyword{return}
  \; v''$, or $\keyword{throw}
  \; v''$, then there exists a type $t''$ such that $\typing{v''}{t''}$
\end{enumerate}
\end{lemma}
\begin{proof}
  By induction on derivation $\mathcal{C}$:
  \begin{itemize}
  \item Case $\mathcal{C} = \textsc{EC-Emp}$ holds directly using the premises.
  \item Case $\mathcal{C} = \textsc{EC-More-Fail}$ and $\mathcal{C} =
    \textsc{EC-More-Ord}$ hold directly using induction hypotheses (including
    \Cref{thm:strongtyping}) and the facts of well-typedness of store extension
    by well-typed environments and well-typedness of variable removals from stores.
  \end{itemize}
\end{proof}

\begin{lemma}
    \label{lem:evalcasestyping}
  If $\begin{array}{c} \mathcal{C\!S} \\
        \evalcases{\esequence{\mathit{cs}}}{v}{\sigma}{\mathit{vres}}{\sigma'}
      \end{array}$, there exists a type $t$ such that $\typing{v}{t}$ , and
      there
    exists a type $t'$ such that
  $\typing{v'}{t'}$ for each value in the input store $v' \in
  \mathrm{img} \; \sigma$, then
  \begin{enumerate}
  \item There exists a type $t''$ such that
  $\typing{v''}{t''}$ for each value in the result store $v'' \in \mathrm{img} \; \sigma'$.
\item If the result value  $\mathit{vres}$ is either
   $\keyword{success}
  \; v'''$,
  $\keyword{return}
  \; v'''$, or $\keyword{throw}
  \; v'''$, then there exists a type $t'''$ such that $\typing{v'''}{t'''}$
\end{enumerate}
\end{lemma}
\begin{proof}
  By induction on the derivation $\mathcal{C\!S}$:
  \begin{itemize}
  \item Case $\mathcal{C\!S} = \textsc{ECS-Emp}$ holds directly using the premises.
  \item Cases $\mathcal{C\!S} = \textsc{ECS-More-Fail}$ and $\mathcal{C\!S} =
    \textsc{ECS-More-Fail}$ hold using \Cref{lem:patternmatchtyping} and the induction hypotheses (including \Cref{lem:evalcasetyping}).
  \end{itemize}
\end{proof}

\begin{lemma}
  \label{lem:evalvisityping}
  If $\begin{array}{c} \mathcal{V} \\
        \evalvisit{\mathit{st}}{\esequence{\mathit{cs}}}{v}{\sigma}{\mathit{vres}}{\sigma'}
      \end{array}$, there exists a type $t$ such that $\typing{v}{t}$ , and
      there
    exists a type $t'$ such that
  $\typing{v'}{t'}$ for each value in the input store $v' \in
  \mathrm{img} \; \sigma$, then
  \begin{enumerate}
  \item There exists a type $t''$ such that
  $\typing{v''}{t''}$ for each value in the result store $v'' \in \mathrm{img} \; \sigma'$.
\item If the result value  $\mathit{vres}$ is either
   $\keyword{success}
  \; v'''$,
  $\keyword{return}
  \; v'''$, or $\keyword{throw}
  \; v'''$, then there exists a type $t'''$ such that $\typing{v'''}{t'''}$
\end{enumerate}
\end{lemma}
\begin{proof}
  By induction on the derivation $\mathcal{V}$:
  \begin{itemize}
  \item Cases $\mathcal{V} = \textsc{EV-TD}$, $\mathcal{V} = \textsc{EV-TDB}$,
    $\mathcal{V} = \textsc{EV-OM-Eq}$, $\mathcal{V} = \textsc{EV-OM-Neq}$ and
    $\mathcal{V} = \textsc{EV-OM-Exc}$ hold by induction hypotheses
    (including \Cref{lem:evaltdvisittyping}).
  \item Cases $\mathcal{V} = \textsc{EV-BU}$, $\mathcal{V} = \textsc{EV-BUB}$,
    $\mathcal{V} = \textsc{EV-IM-Eq}$, $\mathcal{V} = \textsc{EV-IM-Neq}$ and
    $\mathcal{V} = \textsc{EV-IM-Exc}$ hold by induction hypotheses
    (including \Cref{lem:evalbuvisittyping}).   
  \end{itemize}
\end{proof}

\begin{lemma}
    \label{lem:evaltdvisittyping}
  If $\begin{array}{c} \mathcal{V\!T} \\
        \evaltdvisit{\esequence{\mathit{cs}}}{v}{\sigma}{\mathit{br}}{\mathit{vres}}{\sigma'}
      \end{array}$, there exists a type $t$ such that $\typing{v}{t}$ , and
      there
    exists a type $t'$ such that
  $\typing{v'}{t'}$ for each value in the input store $v' \in
  \mathrm{img} \; \sigma$, then
  \begin{enumerate}
  \item There exists a type $t''$ such that
  $\typing{v''}{t''}$ for each value in the result store $v'' \in \mathrm{img} \; \sigma'$.
\item If the result value  $\mathit{vres}$ is either
   $\keyword{success}
  \; v'''$,
  $\keyword{return}
  \; v'''$, or $\keyword{throw}
  \; v'''$, then there exists a type $t'''$ such that $\typing{v'''}{t'''}$
\end{enumerate}
\end{lemma}
\begin{proof}
  By induction on the derivation $\mathcal{V\!T}$:
  \begin{itemize}
  \item All cases ($\mathcal{V\!T} = \textsc{ETV-Break-Sucs}$, $\mathcal{V\!T}
    = \textsc{ETV-Ord-Sucs1}$,  $\mathcal{V\!T}
    = \textsc{ETV-Ord-Sucs2}$,  $\mathcal{V\!T}
    = \textsc{ETV-Exc1}$ and $\mathcal{V\!T}
    = \textsc{ETV-Exc2}$) hold by induction hypotheses (including
    \Cref{lem:evalcasetyping} and \Cref{lem:evaltdvisitstartyping}),
    \Cref{lem:iffailtyping}, \Cref{lem:childrentyping} and \Cref{lem:reconstructtyping}.
  \end{itemize}
\end{proof}

\begin{lemma}
      \label{lem:evaltdvisitstartyping}
  If $\begin{array}{c} \mathcal{V\!T\!S} \\
        \evaltdvisitstar{\esequence{\mathit{cs}}}{\esequence{v}}{\sigma}{\mathit{br}}{\mathit{vres}{\star}}{\sigma'}
      \end{array}$, there exists a type sequence $\esequence{t}$ such that $\esequence{\typing{v}{t}}$ , and
      there
    exists a type $t'$ such that
  $\typing{v'}{t'}$ for each value in the input store $v' \in
  \mathrm{img} \; \sigma$, then
  \begin{enumerate}
  \item There exists a type $t''$ such that
  $\typing{v''}{t''}$ for each value in the result store $v'' \in \mathrm{img} \; \sigma'$.
\item If the result value  $\mathit{vres}{\star}$ is $\keyword{success} \;
  \esequence{v'''}$ then there exists a type sequence $\esequence{t'''}$ such
  that $\esequence{\typing{v'''}{t'''}}$
\item If the result value  $\mathit{vres}{\star}$ is either
  $\keyword{return}
  \; v'''$, or $\keyword{throw}
  \; v'''$, then there exists a type $t'''$ such that $\typing{v'''}{t'''}$
\end{enumerate}
\end{lemma}
\begin{proof}
  By induction on the derivation $\mathcal{V\!T\!S}$
  \begin{itemize}
  \item Case $\mathcal{V\!T\!S} = \textsc{ETVS-Emp}$ holds using the premises.
  \item Cases  $\mathcal{V\!T\!S} = \textsc{ETVS-Break}$, $\mathcal{V\!T\!S} = \textsc{ETVS-More}$, $\mathcal{V\!T\!S} =
    \textsc{ETVS-Exc1}$ and $\mathcal{V\!T\!S} = \textsc{ETVS-Exc2}$ holds using
    the induction hypotheses (including \Cref{lem:evaltdvisittyping}) and \Cref{lem:vcombinetyping}.
  \end{itemize}
\end{proof}

\begin{lemma}
      \label{lem:evalbuvisittyping}
  If $\begin{array}{c} \mathcal{V\!B} \\
        \evalbuvisit{\esequence{\mathit{cs}}}{v}{\sigma}{\mathit{br}}{\mathit{vres}}{\sigma'}
      \end{array}$, there exists a type $t$ such that $\typing{v}{t}$ , and
      there
    exists a type $t'$ such that
  $\typing{v'}{t'}$ for each value in the input store $v' \in
  \mathrm{img} \; \sigma$, then
  \begin{enumerate}
  \item There exists a type $t''$ such that
  $\typing{v''}{t''}$ for each value in the result store $v'' \in \mathrm{img} \; \sigma'$.
\item If the result value  $\mathit{vres}$ is either
   $\keyword{success}
  \; v'''$,
  $\keyword{return}
  \; v'''$, or $\keyword{throw}
  \; v'''$, then there exists a type $t'''$ such that $\typing{v'''}{t'''}$
\end{enumerate}
\end{lemma}
\begin{proof}
  By induction on the derivation $\mathcal{V\!B}$:
  \begin{itemize}
   \item All cases ($\mathcal{V\!B} = \textsc{EBU-No-Break-Sucs}$, $\mathcal{V\!B}
    = \textsc{EBU-Break-Sucs}$,  $\mathcal{V\!B}
    = \textsc{EBU-Fail-Sucs}$,  $\mathcal{V\!B}
    = \textsc{EBU-No-Break-Exc}$, $\mathcal{V\!B}
    = \textsc{EBU-Exc}$, and $\mathcal{V\!B}
    = \textsc{EBU-No-BreakErr})$ hold by induction hypotheses (including
    \Cref{lem:evalcasetyping} and \Cref{lem:evalbuvisitstartyping}),
    \Cref{lem:iffailtyping}, \Cref{lem:childrentyping} and \Cref{lem:reconstructtyping}.
  \end{itemize}
\end{proof}

\begin{lemma}
  \label{lem:evalbuvisitstartyping}
  If $\begin{array}{c} \mathcal{V\!B\!S} \\
        \evalbuvisitstar{\esequence{\mathit{cs}}}{\esequence{v}}{\sigma}{\mathit{br}}{\mathit{vres}{\star}}{\sigma'}
      \end{array}$, there exists a type sequence $\esequence{t}$ such that $\esequence{\typing{v}{t}}$ , and
      there
    exists a type $t'$ such that
  $\typing{v'}{t'}$ for each value in the input store $v' \in
  \mathrm{img} \; \sigma$, then
\begin{enumerate}
  \item There exists a type $t''$ such that
  $\typing{v''}{t''}$ for each value in the result store $v'' \in \mathrm{img} \; \sigma'$.
\item If the result value  $\mathit{vres}{\star}$ is $\keyword{success} \;
  \esequence{v'''}$ then there exists a type sequence $\esequence{t'''}$ such
  that $\esequence{\typing{v'''}{t'''}}$
\item If the result value  $\mathit{vres}{\star}$ is either
  $\keyword{return}
  \; v'''$, or $\keyword{throw}
  \; v'''$, then there exists a type $t'''$ such that $\typing{v'''}{t'''}$
\end{enumerate}

\end{lemma}
\begin{proof}
  By induction on the derivation $\mathcal{V\!B\!S}$
  \begin{itemize}
  \item Case $\mathcal{V\!B\!S} = \textsc{EBUS-Emp}$ holds using the premises.
  \item Cases $\mathcal{V\!B\!S} = \textsc{EBUS-Break}$, $\mathcal{V\!B\!S} = \textsc{EBUS-More}$, $\mathcal{V\!B\!S} =
    \textsc{EBUS-Exc1}$ and $\mathcal{V\!B\!S} = \textsc{EBUS-Exc2}$ holds using
    the induction hypotheses (including \Cref{lem:evalbuvisittyping}) and \Cref{lem:vcombinetyping}.
  \end{itemize}
\end{proof}

\strongtyping*
\begin{proof}
  By induction on the derivation $\mathcal{E}$:
  \begin{itemize}
  \item Cases $\mathcal{E} = \textsc{E-Val}$, $\mathcal{E} =
    \textsc{E-Var-Sucs}$, $\mathcal{E} = \textsc{E-Var-Err}$, $\mathcal{E} =
    \textsc{E-Break}$,
     $\mathcal{E} = \textsc{E-Continue}$ and  $\mathcal{E} = \textsc{E-Fail}$  hold directly
    using the premises.
  \item Cases $\mathcal{E} = \textsc{E-Un-Exc1}$, $\mathcal{E} =
    \textsc{E-Bin-Exc1}$, $\mathcal{E} =
    \textsc{E-Bin-Exc2}$, $\mathcal{E} = \textsc{E-Cons-Err}$,
    $\mathcal{E} = \textsc{E-Cons-Exc}$, $\mathcal{E} = \textsc{E-List-Err}$,
    $\mathcal{E} = \textsc{E-List-Exc}$, $\mathcal{E} = \textsc{E-Set-Err}$,
    $\mathcal{E} = \textsc{E-Set-Exc}$, $\mathcal{E} = \textsc{E-Map-Err}$,
    $\mathcal{E} = \textsc{E-Map-Exc}$, $\mathcal{E} = \textsc{E-Lookup-Err}$,
    $\mathcal{E} = \textsc{E-Lookup-Exc1}$, $\mathcal{E} =
    \textsc{E-Lookup-Exc2}$, $\mathcal{E} = \textsc{E-Update-Err1}$,
    $\mathcal{E} = \textsc{E-Update-Err2}$, $\mathcal{E} =
    \textsc{E-Update-Exc1}$, $\mathcal{E} = \textsc{E-Update-Exc2}$,
    $\mathcal{E} = \textsc{E-Update-Exc3}$, $\mathcal{E} =
    \textsc{E-Call-Arg-Err}$, $\mathcal{E} = \textsc{E-Call-Arg-Exc}$,
    $\mathcal{E} = \textsc{E-Ret-Sucs}$, $\mathcal{E} = \textsc{E-Ret-Exc}$,
    $\mathcal{E} = \textsc{E-Asgn-Err}$, $\mathcal{E} = \textsc{E-Asgn-Exc}$,
    $\mathcal{E} = \textsc{E-If-True}$, $\mathcal{E} = \textsc{E-If-False}$,
    $\mathcal{E} = \textsc{E-If-Err}$, $\mathcal{E} = \textsc{E-If-Exc}$,
    $\mathcal{E} = \textsc{E-Switch-Sucs}$, $\mathcal{E} =
    \textsc{E-Switch-Exc1}$, $\mathcal{E} = \textsc{E-Switch-Exc2}$,
     $\mathcal{E} = \textsc{E-Visit-Sucs}$, $\mathcal{E} =
     \textsc{E-Visit-Fail}$,
     $\mathcal{E} = \textsc{E-Visit-Exc1}$, and $\mathcal{E} =
     \textsc{E-Visit-Exc2}$, $\mathcal{E} = \textsc{E-For-Sucs}$, $\mathcal{E} =
     \textsc{E-For-Exc}$, $\mathcal{E} =
     \textsc{E-While-True-Sucs}$, $\mathcal{E} =
     \textsc{E-While-Exc1}$, $\mathcal{E} =
     \textsc{E-While-Exc2}$, $\mathcal{E} =
     \textsc{E-While-Err}$, $\mathcal{E} =
     \textsc{E-Solve-Eq}$, $\mathcal{E} =
     \textsc{E-Solve-Neq}$, $\mathcal{E} =
     \textsc{E-Solve-Exc}$, $\mathcal{E} =
     \textsc{E-Solve-Err}$, , $\mathcal{E} =
     \textsc{E-Thr-Sucs}$, $\mathcal{E} =
     \textsc{E-Thr-Err}$,  $\mathcal{E} =
     \textsc{E-Fin-Sucs}$, $\mathcal{E} =
     \textsc{E-Fin-Exc}$, and $\mathcal{E} =
     \textsc{E-Try-Ord}$
     hold by induction
    hypotheses (including \Cref{lem:evalstartyping}, \Cref{lem:evalgentyping}, \Cref{lem:evalcasestyping}
    and \Cref{lem:evalvisityping}).
  \item Case $\mathcal{E} = \textsc{E-Un-Sucs}$ holds by its induction hypothesis
    and strong typing of $\sem{\ominus}$.
  \item Case $\mathcal{E} = \textsc{E-Bin-Sucs}$ holds by its induction hypothesis
    and strong typing of $\sem{\oplus}$.
  \item Case $\mathcal{E} = \textsc{E-Cons-Sucs}$ holds by its induction
    hypothesis given by \Cref{lem:evalstartyping}, and then by using
    $\textsc{T-Constructor}$ with the provided typing premises.
  \item Case $\mathcal{E} = \textsc{E-List-Sucs}$ holds by its induction
    hypothesis given by \Cref{lem:evalstartyping}, and then by using
    $\textsc{T-List}$.
  \item Case $\mathcal{E} = \textsc{E-Set-Sucs}$ holds by its induction
    hypothesis given by \Cref{lem:evalstartyping}, and then by using
    $\textsc{T-Set}$.
  \item Case $\mathcal{E} = \textsc{E-Map-Sucs}$ holds by its induction
    hypothesis given by \Cref{lem:evalstartyping}, and then by using
    $\textsc{T-Map}$.
  \item Case $\mathcal{E} = \textsc{E-Lookup-Sucs}$ holds by its induction
    hypotheses and inversion of $\mathcal{T}$ to $\textsc{T-Map}$.
  \item Case $\mathcal{E} = \textsc{E-Lookup-NoKey}$ holds by its induction
    hypotheses and using \textsc{T-Cons} on the definition of $\textrm{NoKey}$.
  \item Case $\mathcal{E} = \textsc{E-Update-Sucs}$ holds by its induction
    hypotheses and by using $\textsc{T-Map}$ to reconstruct the type derivation
    for the result.
  \item Cases $\mathcal{E} = \textsc{E-Call-Sucs}$, $\mathcal{E} =
    \textsc{E-Call-Res-Exc}$, $\mathcal{E} = \textsc{E-Call-Res-Err1}$, and
    $\mathcal{E} = \textsc{E-Call-Res-Err2}$ hold by induction hypotheses (including
    \Cref{lem:evalstartyping}) and by using the fact that extracting variables from
    and extending well-typed stores produces well-typed stores.
  \item Cases $\mathcal{E} = \textsc{E-Asgn-Sucs}$ and $\mathcal{E} = \textsc{E-Try-Catch}$ hold by induction
    hypotheses and the fact that extending well-typed stores with well-typed
    environments preserves well-typedness.
  \item Cases $\mathcal{E} = \textsc{E-Switch-Fail}$, $\mathcal{E} =
    \textsc{E-While-False}$, and $\mathcal{E} = \textsc{E-While-True-Break}$
    hold by induction hypotheses and using $\textsc{T-Void}$.
  \item Cases $\mathcal{E} = \textsc{E-Block-Sucs}$, $\mathcal{E} =
    \textsc{E-Block-Exc}$ hold by the induction hypothesis
    given by \Cref{lem:evalstartyping} and the fact that removing variables from
    a well-typed store produces a well-typed store.
  \end{itemize}
\end{proof}

We will for \Cref{thm:partialprogress}, also need to state some lemmas. We will
for the proof focus mainly on cases where the induction hypothesis does not
$\keyword{timeout}$, since if it does it is trivially possible to construct a
$\keyword{timeout}$ derivation for the result syntax.

First, we need a lemma that specifies that the sequence of values produced by
the $\textrm{children}$ is strictly smaller than the input value. Let $v \prec v'$
denote the relation that $v$ is syntactically contained in $v'$, then our target property
is specified in \Cref{lem:childrenordering}.
\begin{lemma}
  \label{lem:childrenordering}
  If $\esequence{v'} = \textrm{\textnormal{children}}(v)$ then $v'_i \prec v$
  for all $i$.
\end{lemma}
\begin{proof}
  Directly by induction on $v$.
\end{proof}

We have a lemma for progress on reconstruction:
\begin{lemma}
  \label{lem:reconsprogress}
  It is possible to construct a derivation
  $\reconstruct{v}{\esequence{v'}}{\mathit{rcres}}$ for any well-typed value $v$
  and well-typed value sequence $\esequence{v'}$.
\end{lemma}
\begin{proof}
  Straightforwardly by case analysis on $v$.
\end{proof}

We have two mutually recursive lemmas for pattern matching:
\Cref{lem:matchprogress} and \Cref{lem:matchallprogress}.
\begin{lemma}
  \label{lem:matchprogress}
  It is possible to construct a
  derivation $\match{p}{v}{\sigma}{\esequence{\rho}}$ for any pattern
  $p$, well-typed value $v$, well-typed store $\sigma$.
\end{lemma}
\begin{proof}
  By induction on syntax of $p$:
  \begin{itemize}
  \item Case $p = \mathit{vb}$:
    We proceed by testing whether $v$ is equal to $\mathit{vb}$:
    \begin{itemize}
    \item Case $v = \mathit{vb}$ then use $\textsc{P-Val-Sucs}$.
    \item Case $v \neq \mathit{vb}$ then use $\textsc{P-Val-Fail}$.
    \end{itemize}
  \item Case $p = x$:
    We proceed by testing whether $x$ is in $\mathrm{dom} \; \sigma$
    \begin{itemize}
    \item Case $x \in \mathrm{dom} \; \sigma$: we proceed to test whether
      v is equal to $\sigma(x)$
      \begin{itemize}
      \item Case $v = \sigma(x)$ then use $\textsc{P-Var-Uni}$.
      \item Case $v \neq \sigma(x)$ then use $\textsc{P-Var-Fail}$.
      \end{itemize}
    \item Case $x \notin \mathrm{dom} \; \sigma$ then use $\textsc{P-Var-Bind}$.
    \end{itemize}
  \item Case $p = k(\esequence{p'})$:
    We proceed to test whether $v$ is equal to $k(\esequence{v'})$ for some \esequence{v'}
    \begin{itemize}
    \item Case $v = k(v')$ then use $\textsc{P-Cons-Sucs}$ using the induction
      hypotheses of $\esequence{p'}$ with $\esequence{v'}$.
    \item Case $v \neq k(v')$ then use $\textsc{P-Cons-Fail}$.
    \end{itemize}
  \item Case $p = t \; x : p'$:
    From our premise we know that $v$ is well-typed, i.e. that there exists a
    $t'$ such that $\typing{v}{t'}$. We proceed to test whether
    $\subtyping{t'}{t}$.
    \begin{itemize}
    \item Case $\subtyping{t'}{t}$ then use $\textsc{P-Type-Sucs}$ using the
      induction hypothesis on $p'$ with v.
    \item Case $\notsubtyping{t'}{t}$ then use $\textsc{P-Type-Fail}$.
    \end{itemize}
  \item Case $p = [\esequence{{{\star}p'}}]$:
   
    We proceed to test whether $v = [\esequence{v'}]$ for some value sequence
    $\esequence{v'}$.
    \begin{itemize}
    \item Case $v = [\esequence{v'}]$ then use $\textsc{P-List-Sucs}$ induction hypothesis given by \Cref{lem:matchallprogress} on $\esequence{{{\star}p'}}$ with $\esequence{v'}$
    \item Case $v \neq  [\esequence{v'}]$ then use $\textsc{P-List-Fail}$
    \end{itemize}
  \item Case $p = \{\esequence{{{\star}p'}}\}$:

     We proceed to test whether $v = \{\esequence{v'}\}$ for some value sequence
    $\esequence{v'}$.
    \begin{itemize}
    \item Case $v = \{\esequence{v'}\}$ then use $\textsc{P-Set-Sucs}$ induction hypothesis given by \Cref{lem:matchallprogress} on $\esequence{{{\star}p'}}$ with $\esequence{v'}$
    \item Case $v \neq  \{\esequence{v'}\}$ then use $\textsc{P-Set-Fail}$
    \end{itemize}
  \item Case $p = /p'$:

    \begin{itemize}
    \item Using the induction hypothesis on $p'$ with $v$, we get $\match{p'}{v}{\sigma}{\esequence{\rho}}$.
      \item Now, let $\esequence{v'} = \textrm{\textnormal{children}}(v)$.
    In order to handle the self-recursive calls using $/p'$, we proceed by inner
    well-founded induction on the relation $\prec$ using value sequence $\esequence{v'}$:
    \begin{itemize}
    \item Using the inner induction hypothesis we get derivations
      $\match{/p'}{v'_i}{\sigma}{\esequence{\rho_i}}$ for all $i$
    \end{itemize}
    \item Finally, we use $\textsc{P-Deep}$ on the derivations we got from the
      outer and inner induction hypotheses.
  \end{itemize}
  \end{itemize}
\end{proof}

\begin{lemma}
  \label{lem:matchallprogress}
  It is possible to construct a
  derivation
  $\matchall{\esequence{{{\star}p}}}{v}{\sigma}{\mathbb{V}}{\ttuple{} \quad \otimes}{\esequence{\rho}}$ for any pattern
  $p$, well-typed value $v$, well-typed store $\sigma$, well-typed visited value
  set $\mathbb{V}$, type-preserving partition operator $\otimes$.
\end{lemma}
\begin{proof}
  By induction on the syntax of $\esequence{{{\star}p}}$:
  \begin{itemize}
  \item Case $\esequence{{{\star}p}} = \varepsilon$

    By case analysis on $\esequence{v}$:
    \begin{itemize}
    \item Case $\esequence{v} = \varepsilon$ then use $\textsc{PL-Emp-Both}$
    \item Case $\esequence{v} = v', \esequence{v''}$ then use $\textsc{PL-Emp-Pat}$
    \end{itemize}
  \item Case $\esequence{{{\star}p}} = p', \esequence{{{\star}p''}}$:

    Because of backtracking we need to do an inner induction to handle the
    cases which recurse to the same pattern sequence. Let
    $\mathbb{V}_{\mathrm{all}} = \left\{ v' \middle| \exists \esequence{v''}. \esequence{v} =
      v' \otimes \esequence{v''} \right\}$, then we proceed by inner
    induction on $|\mathbb{V}_{\mathrm{all}}| - |\mathbb{V}|$:
    \begin{itemize}
    \item Case $|\mathbb{V}_{\mathrm{all}}| - |\mathbb{V}| = 0$ then we have
      $|\mathbb{V}_{\mathrm{all}}| = |\mathbb{V}|$ and thus we have to use $\textsc{PL-More-Pat-Exh}$.
    \item Case $|\mathbb{V}_{\mathrm{all}}| - |\mathbb{V}| > 0$:
      Then there exists a $\mathbb{V}' \neq \emptyset$ such that
      $\mathbb{V}_{\mathrm{all}} = \mathbb{V}' \uplus \mathbb{V}$, and we can
      pick a $v' \in \mathbb{V}'$ such that $\esequence{v} = v' \otimes
      \esequence{v''}$ for some $v'$.
      We can finally use $\textsc{PL-More-Pat-Re}$ using
      \Cref{lem:matchprogress} with $p$, $v'$ and $\sigma$, the outer induction
      hypothesis with $\esequence{{\star}p''}$, $\esequence{v''}$ and $\sigma$,
      and the inner induction hypothesis on $\mathbb{V} \cup \{v'\}$
    \end{itemize}
  \item Case $\esequence{{{\star}p}} = {\star}x, \esequence{{{\star}p''}}$:

    We proceed to test whether $x$ is in $\mathrm{dom} \; \sigma$
    \begin{itemize}
    \item Case $x \in \mathrm{dom} \; \sigma$:
      By case analysis on $\sigma(x) = v'$:
      \begin{itemize}
      \item Case $v' = \ttuple{\esequence{v''}}$:

        By checking whether there exists $\esequence{v'''}$ such that
        $\esequence{v} = \esequence{v''} \otimes \esequence{v'''}$
        \begin{itemize}
        \item Case $\exists \esequence{v'''}. \esequence{v} = \esequence{v''}
          \otimes \esequence{v'''}$ then use $\textsc{PL-More-Star-Uni}$ with
          induction hypothesis applied to $\esequence{{\star}p''}$ and $\esequence{v'''}$.
        \item Case $\nexists \esequence{v'''}. \esequence{v} = \esequence{v''}
          \otimes \esequence{v'''}$: then use rule $\textsc{PL-More-Star-Pat-Fail}$.
        \end{itemize}
      \item Case $v' \neq \ttuple{\esequence{v''}}$: then use $\textsc{PL-More-Star-Val-Fail}$.
      \end{itemize}
    \item Case $x \notin \mathrm{dom} \; \sigma$:
      
    Similarly to the ordinary pattern case, we need to do an inner induction to handle the
    cases which recurse to the same star pattern sequence. Let
    $\mathbb{V}_{\mathrm{all}} = \left\{ \esequence{v}' \middle| \exists \esequence{v''}. \esequence{v} =
      \esequence{v'} \otimes \esequence{v''} \right\}$, then we proceed by inner
    induction on $|\mathbb{V}_{\mathrm{all}}| - |\mathbb{V}|$
    \begin{itemize}
    \item Case $|\mathbb{V}_{\mathrm{all}}| - |\mathbb{V}| = 0$:

      Then we have $\mathbb{V} = \mathbb{V}_{\mathrm{all}}$ and so the only
      applicable rule is $\textsc{PL-More-Star-Exh}$ since
      $\mathbb{V}_{\mathrm{all}}$ covers all partitions and thus $\nexists \esequence{v'}, \esequence{v''}. \esequence{v} =
    \esequence{v'} \otimes \esequence{v''} \wedge \esequence{v'}
    \notin \mathbb{V}$ holds.
    \item Case $|\mathbb{V}_{\mathrm{all}}| - |\mathbb{V}| > 0$:

      Then there exists $\mathbb{V}' \neq \emptyset$ such that
      $\mathbb{V}_{\mathrm{all}} = \mathbb{V}' \uplus \mathbb{V}$ and we can
      pick $\esequence{v'} \in \mathbb{V}'$ such that $\esequence{v} =
      \esequence{v}' \otimes \esequence{v}''$ for some $\esequence{v}''$.
      
      We can now use $\textsc{PL-More-Star-Re}$ by using the outer hypothesis on
      $\esequence{{{\star}p}}$ with $\esequence{v''}$ and $\emptyset$, and the
      inner hypothesis on $\mathbb{V} \cup \{\esequence{v'}\}$ (since
      the size decreases by 1)
    \end{itemize}
    \end{itemize}
  \end{itemize}
\end{proof}

We now have a series of lemmas that are mutually recursive with
\Cref{thm:partialprogress}:
\Cref{lem:evalstarprogress}, \Cref{lem:evaleachprogress},
\Cref{lem:evalgenprogress},
\Cref{lem:evalcaseprogress}, \Cref{lem:evalcasesprogress},
\Cref{lem:evalvisitprogress}, \Cref{lem:evaltdvisitprogress},
\Cref{lem:evaltdvisitstarprogress}, \Cref{lem:evalbuvisitprogress}, and
\Cref{lem:evalbuvisitstarprogress}.

\begin{lemma}
  \label{lem:evalstarprogress}
  It is possible to construct a derivation
  $\evalexprstarfuel{\esequence{e}}{\sigma}{\mathit{vres}{\star}}{\sigma'}{n}$, for any
  expression sequence $\esequence{e}$,
  well-typed store $\sigma$ and fuel $n$.
\end{lemma}
\begin{proof}
  By induction on $n$:
  \begin{itemize}
  \item Case $n = 0$ then use corresponding $\keyword{timeout}$-rule.
  \item Case $n > 0$, then by case analysis on the expression sequence
    $\esequence{e}$:
    \begin{itemize}
    \item Case $\esequence{e} = \varepsilon$ then use $\textsc{ES-Emp}$.
    \item Case $\esequence{e} = e', \esequence{e''}$:
      \begin{itemize}
      \item Using induction hypothesis given by \Cref{thm:partialprogress} on
        $n-1$ with $e'$ and $\sigma$ we get a derivation
        $\evalexprfuel{e'}{\sigma}{\mathit{vres}'}{\sigma''}{n-1}$
        
        By case analysis on $\mathit{vres}'$:
        \begin{itemize}
        \item Case $\mathit{vres}' = \keyword{success} \; v'$:

        By \Cref{thm:strongtyping} we get that $\sigma''$ is well-typed.

        By induction hypothesis on $n-1$ with $\esequence{e''}$ and $\sigma''$
        we get a derivation
        $\evalexprstarfuel{\esequence{e''}}{\sigma''}{\mathit{vres}{\star}''}{\sigma'}{n-1}$
        
        By case analysis on $\mathit{vres}{\star}''$:
        \begin{itemize}
        \item Case $\mathit{vres}{\star}'' = \keyword{success} \;
          \esequence{v''}$ then use $\textsc{ES-More}$
        \item Case $\mathit{vres}{\star}'' = \mathit{exres}$ then use $\textsc{ES-Exc2}$
        \end{itemize}
        \item Case $\mathit{vres}' = \mathit{exres}$ then use \textsc{ES-Exc1}
        \end{itemize}
      \end{itemize}
    \end{itemize}
  \end{itemize}
\end{proof}

\begin{lemma}
  \label{lem:evaleachprogress}
  It is possible to construct a derivation
  $\evaleachfuel{e}{\esequence{\rho}}{\sigma}{\mathit{vres}}{\sigma'}{n}$ for
  any expression $\esequence{e}$, well-typed environment sequence
  $\esequence{\rho}$, well-typed store $\sigma$ and fuel $n$.
\end{lemma}
\begin{proof}
  By induction on $n$:

  \begin{itemize}
  \item Case $n = 0$ then use the corresponding $\keyword{timeout}$-derivation.
  \item Case $n > 0$:

   By case analysis on the environment sequence $\esequence{\rho}$:
  \begin{itemize}
  \item Case $\esequence{\rho} = \varepsilon$ then use $\textsc{EE-Emp}$
  \item Case $\esequence{\rho} = \rho', \esequence{\rho''}$:

    By induction hypothesis given by \Cref{thm:partialprogress} on $n - 1$ with
    $e$ and $\sigma \rho$ we get a derivation $\evalexprfuel{e}{\sigma
      \rho}{\mathit{vres}'}{\sigma''}{n - 1}$.

    By case analysis on $\mathit{vres}'$:
    \begin{itemize}
    \item Cases $\mathit{vres}' = \keyword{success} \; v'$ and $\mathit{vres}' =
      \keyword{continue}$ then use $\textsc{EE-More-Sucs}$ with above derivation
      and the induction hypothesis on  $n-1$ with $\esequence{\rho''}$, $e$,
      and $\sigma''' = \sigma'' \setminus \mathrm{dom} \; \rho'$.
    \item Cases $\mathit{vres}' = \keyword{break}$ then use $\textsc{EE-More-Break}$
    \item Cases $\mathit{vres}' = \keyword{throw} \; v'$, $\mathit{vres}' =
      \keyword{return} \; v'$, $\mathit{vres}' = \keyword{fail}$, and
      $\mathit{vres}' = \keyword{error}$ then use $\textsc{EE-More-Exc}$
    \end{itemize}
  \end{itemize}
  \end{itemize}
\end{proof}

\begin{lemma}
  \label{lem:evalgenprogress}
  It is possible to construct a derivation
  $\evalgenexprfuel{g}{\sigma}{\mathit{envres}}{\sigma'}{n}$ for any generator
  expression $g$, well-typed store $\sigma$ and fuel $n$.
\end{lemma}
\begin{proof}
  By induction on $n$:

  \begin{itemize}
  \item Case $n = 0$  then use the corresponding $\keyword{timeout}$-derivation.
  \item Case $n > 0$:

    By case analysis on $g$:
    \begin{itemize}
    \item Case $g = p \coloneqq e$:

      By induction hypothesis given by \Cref{thm:partialprogress} on $n - 1$ with
      $e$ and $\sigma$ we get a derivation
      $\evalexprfuel{e}{\sigma}{\mathit{vres}'}{\sigma'}{n - 1}$

      By case analysis on $\mathit{vres}'$:
      \begin{itemize}
      \item Case $\mathit{vres}' = \keyword{success} \; v'$ then use
        $\textsc{G-Pat-Sucs}$ with above derivation and the derivation from
        applying \Cref{lem:matchprogress} on $p$ with $\sigma'$ and $v'$.
      \item Case $\mathit{vres}' = \mathit{exres}$ then use $\textsc{G-Pat-Exc}$:
      \end{itemize}
    \item Case $g = x <- e$:
      
      By induction hypothesis given by \Cref{thm:partialprogress} on $n - 1$ with
      $e$ and $\sigma$ we get a derivation
      $\evalexprfuel{e}{\sigma}{\mathit{vres}'}{\sigma'}{n - 1}$
      
      By case analysis on $\mathit{vres}'$
      \begin{itemize}
      \item Case $\mathit{vres}' = \keyword{success} \; v'$  
        
      By case analysis on $v'$:
      \begin{itemize}
      \item Case $v' = [\esequence{v''}]$ then use $G-Enum-List$
      \item Case $v' = \{\esequence{v''}\}$ then use $G-Enum-Set$
      \item Case $v' = (\esequence{v'' : v'''})$ then use $G-Enum-Map$
      \item Cases $v' = \mathit{vb}$, $v' = k(\esequence{v''})$ and $v' = \blacksquare$ then use $\textsc{G-Enum-Err}$
      \end{itemize}
      \item Case $\mathit{vres}' = \mathit{exres}$ then use $\textsc{G-Enum-Exc}$
      \end{itemize}
    \end{itemize}
  \end{itemize}
\end{proof}

\begin{lemma}
  \label{lem:evalcaseprogress}
  It is possible to construct a derivation
  $\evalcasefuel{\esequence{\rho}}{e}{\sigma}{\mathit{vres}}{\sigma'}{n}$ for
  any well-typed environment sequence $\esequence{\rho}$, expression $e$,
  well-typed store $\sigma$ and fuel $n$.
\end{lemma}
\begin{proof}
  By induction on $n$:
  \begin{itemize}
  \item Case $n = 0$ then use the corresponding $\keyword{timeout}$-derivation.
  \item Case $n > 0$:
    
    By case analysis on the environment sequence $\esequence{\rho}$
    \begin{itemize}
    \item Case $\esequence{\rho} = 0$ then use $\textsc{EC-Emp}$.
    \item Case $\esequence{\rho} = \rho', \esequence{\rho''}$:

      Using the induction hypothesis given by \Cref{thm:partialprogress} on $n -
      1$ with $e$ and $\sigma$ we get a derivation
      $\evalexprfuel{e}{\sigma}{\mathit{vres}'}{\sigma'}{n - 1}$.

      By case analysis on $\mathit{vres}'$
      \begin{itemize}
      \item Case $\mathit{vres}' = \keyword{fail}$ then use
        $\textsc{EC-More-Fail}$ on above derivation and the derivation from
        applying the induction hypothesis on $n - 1$ with $\esequence{\rho''}$,
        $e$ and $\sigma$.
      \item Case $\mathit{vres}' \neq \keyword{fail}$ then use $\textsc{EC-More-Ord}$.
      \end{itemize}
    \end{itemize}
  \end{itemize}
\end{proof}

\begin{lemma}
  \label{lem:evalcasesprogress}
  It is possible to construct a derivation
  $\evalcasesfuel{\esequence{\mathit{cs}}}{v}{\sigma}{\mathit{vres}}{\sigma'}{n}$ for
  any case sequence $\esequence{\mathit{cs}}$, well-typed value $v$,
  well-typed store $\sigma$ and fuel $n$.
\end{lemma}
\begin{proof}
  By induction on $n$:
  \begin{itemize}
  \item Case $n = 0$ then use the corresponding $\keyword{timeout}$-derivation.
  \item Case $n > 0$:

  By case analysis on the case sequence $\esequence{\mathit{cs}}$
  \begin{itemize}
  \item Case $\esequence{\mathit{cs}} = \varepsilon$ then use \textsc{ECS-Emp}
  \item Case $\esequence{\mathit{cs}} = \keyword{case} \; p => e,
    \esequence{\mathit{cs}}'$:

    Using \Cref{lem:matchprogress} on $p$ with $v$ and $\sigma$ gives us a
    derivation $\match{p}{v}{\sigma}{\esequence{\rho}}$. By
    \Cref{lem:patternmatchtyping} we know that $\esequence{\rho}$ is well-typed.

    Using induction hypothesis given by \Cref{lem:evalcaseprogress} on $n-1$
    with $\esequence{\rho}$, e and $\sigma$ we get a derivation
    $\evalcasefuel{\esequence{\rho}}{e}{\sigma}{\mathit{vres}'}{\sigma'}{n -1}$.
    
    By case analysis on $\mathit{vres}'$:
    \begin{itemize}
    \item Case $\mathit{vres}' = \keyword{fail}$ then use
      $\textsc{ECS-More-Fail}$ using above derivation and the derivation given
      by the induction
      hypothesis on $n - 1$ with $\esequence{\mathit{cs}}'$, $v$ and $\sigma$.
    \item Case $\mathit{vres}' \neq \keyword{fail}$ then use $\textsc{ECS-More-Ord}$
    \end{itemize}
  \end{itemize}
  \end{itemize}
\end{proof}

\begin{lemma}
  \label{lem:evalvisitprogress}
  It is possible to construct a derivation 
  $\evalvisitfuel{\mathit{st}}{\esequence{\mathit{cs}}}{v}{\sigma}{\mathit{vres}}{\sigma'}{n}$
  for any case sequence $\esequence{\mathit{cs}}$, well-typed value $v$,
  well-typed store $\sigma$, traversal strategy $\mathit{st}$ and fuel $n$.
\end{lemma}
\begin{proof}
By induction on $n$:
\begin{itemize}
\item Case $n = 0$ then use the corresponding $\keyword{timeout}$-derivation.
\item Case $n > 0$:

  By case analysis on syntax of $\mathit{st}$
  \begin{itemize}
  \item Case $\mathit{st} = \keyword{top-down}$ then use $\textsc{EV-TD}$ with
    the derivation from the induction hypothesis given by \Cref{lem:evaltdvisitprogress}
    on $n - 1$ with $\esequence{\mathit{cs}}$, $v$, $\sigma$ and $\keyword{no-break}$.
  \item Case $\mathit{st} = \keyword{top-down-break}$ then use $\textsc{EV-TDB}$
    with the derivation from the induction hypothesis given by \Cref{lem:evaltdvisitprogress}
    on $n - 1$ with $\esequence{\mathit{cs}}$, $v$, $\sigma$ and $\keyword{break}$.
  \item Case $\mathit{st} = \keyword{bottom-up}$ then use $\textsc{EV-BU}$ with
    the derivation from the induction hypothesis given by \Cref{lem:evalbuvisitprogress}
    on $n - 1$ with $\esequence{\mathit{cs}}$, $v$, $\sigma$ and $\keyword{no-break}$.
  \item Case $\mathit{st} = \keyword{bottom-up-break}$ then use
    $\textsc{EV-BUB}$ with
    the derivation from the induction hypothesis given by \Cref{lem:evalbuvisitprogress}
    on $n - 1$ with $\esequence{\mathit{cs}}$, $v$, $\sigma$ and $\keyword{break}$.
  \item Case $\mathit{st} = \keyword{outermost}$:

    Using the induction hypothesis by \Cref{lem:evaltdvisitprogress} on $n - 1$
    with $\esequence{\mathit{cs}}$, $v$, $\sigma$ and $\keyword{no-break}$ we get a
    derivation
    $\evaltdvisitfuel{\esequence{\mathit{cs}}}{v}{\sigma}{\keyword{no-break}}{\mathit{vres}'}{\sigma''}{n-1}$.
    We know the well-typedness of the output components from \Cref{lem:evaltdvisittyping}.

    By case analysis on $\mathit{vres}'$:
    \begin{itemize}
    \item Case $\mathit{vres}' = \keyword{success} \; v'$:
      
      We proceed by checking whether $v = v'$:
      \begin{itemize}
      \item Case $v = v'$ then use $\textsc{EV-OM-Eq}$ using above derivation.
      \item Case $v \neq v'$ then use $\textsc{EV-OM-Neq}$ using above
        derivation and the derivation from the induction hypothesis given by
        \Cref{lem:evaltdvisitprogress} on $n - 1$ with
        $\esequence{\mathit{cs}}$, $\mathit{v'}$, $\sigma''$ and $\keyword{no-break}$.
      \end{itemize}
      \item Case $\mathit{vres}' = \mathit{exres}$ then use $\textsc{EV-OM-Exc}$
        using above derivation.
      \end{itemize}
    \item Case $\mathit{st} = \keyword{innermost}$:
      Using the induction hypothesis by \Cref{lem:evalbuvisitprogress} on $n - 1$
      with $\esequence{\mathit{cs}}$, $v$, $\sigma$ and $\keyword{no-break}$ we get a
      derivation
      $\evalbuvisitfuel{\esequence{\mathit{cs}}}{v}{\sigma}{\keyword{no-break}}{\mathit{vres}'}{\sigma''}{n-1}$.
      We know the well-typedness of the output components from \Cref{lem:evalbuvisittyping}.

      By case analysis on $\mathit{vres}'$:
      \begin{itemize}
      \item Case $\mathit{vres}' = \keyword{success} \; v'$:

        We proceed by checking whether $v = v'$:
        \begin{itemize}
        \item Case $v = v'$ then use $\textsc{EV-IM-Eq}$ using above derivation.
        \item Case $v \neq v'$ then use $\textsc{EV-IM-Neq}$ using above
          derivation and the derivation from the induction hypothesis given by
          \Cref{lem:evalbuvisitprogress} on $n - 1$ with
          $\esequence{\mathit{cs}}$, $\mathit{v'}$, $\sigma''$ and $\keyword{no-break}$.
        \end{itemize}
      \item Case $\mathit{vres}' = \mathit{exres}$ then use $\textsc{EV-IM-Exc}$
        using above derivation.
      \end{itemize}
    \end{itemize}
  \end{itemize}
  \end{proof}

  \begin{lemma}
    \label{lem:evaltdvisitprogress}
    It is possible to construct a derivation
    $\evaltdvisitfuel{\esequence{\mathit{cs}}}{v}{\sigma}{\mathit{br}}{\mathit{vres}}{\sigma'}{n}$ for any $\esequence{\mathit{cs}}$, well-typed value $v$,
    well-typed store $\sigma$, breaking strategy $\mathit{br}$ and fuel $n$.
  \end{lemma}
  \begin{proof}
    By induction on $n$:

    \begin{itemize}
    \item Case $n = 0$ then use the corresponding
      $\keyword{timeout}$-derivation.
    \item Case $n > 0$:

      Applying the induction hypothesis given by \Cref{lem:evalcasesprogress} on
      $n - 1$ with $\mathit{cs}$, $v$ and $sigma$ we get a derivation
      $\begin{array}{c} \mathcal{C\!S} \\ \evalcasesfuel{\esequence{\mathit{cs}}}{v}{\sigma}{\mathit{vres}'}{\sigma''}{n-1} \end{array}$

      We proceed to check whether $\mathit{vres}'$ has syntax $\mathit{vfres}'$
      for some $\mathit{vfres}'$:
      \begin{itemize}
      \item Case $\mathit{vres}' = \mathit{vfres}'$:

        We proceed to check whether $\mathit{vfres}' = \keyword{success} \; v'$
        and $\mathit{br} = \keyword{break}$
        
        \begin{itemize}
        \item Case $\mathit{vfres}' = \keyword{success} \; v'$ and $\mathit{br} =
          \keyword{break}$ then use $\textsc{ETV-Break-Sucs}$ with  $\mathcal{C\!S}$.
        \item Case $\mathit{vfres}' \neq \keyword{success} \; v'$ or $\mathit{br} \neq \keyword{break}$:
         
          By Boolean logic, we have that $\mathit{br} \neq
          \keyword{break} => \mathit{vfres}' = \keyword{fail}$ is satisfied.

          Let $v'' = \textrm{if-fail}(\mathit{vfres}, v)$ and $\esequence{v'''}
          = \textrm{children}(v'')$ (all well-typed from
          \Cref{lem:iffailtyping} and \Cref{lem:childrentyping}).

          By induction hypothesis given by \Cref{lem:evaltdvisitstartyping} on
          $n - 1$ with $\esequence{\mathit{cs}}$, $\esequence{v'''}$ and
          $\sigma''$, we get a derivation  $\begin{array}{c} \mathcal{C\!S}' \\
                                              \evaltdvisitstarfuel{\esequence{\mathit{cs}}}{\esequence{v'''}}{\sigma''}{\mathit{br}}{\mathit{vres}''}{\sigma'}{n
                                              - 1} \end{array}$.
                                         
        We know that the results are well typed from
        \Cref{lem:evaltdvisitstartyping}.

        By case analysis on $\mathit{vres}''$:
        \begin{itemize}
        \item Case $\mathit{vfres}'' = \keyword{success} \; \esequence{v''''}$
          then use $\textsc{ETV-Ord-Sucs2}$ with $\mathcal{C\!S}$,
          $\mathcal{C\!S}'$ and the result derivation from applying \Cref{lem:reconsprogress}.
        \item Case $\mathit{vfres}'' = \keyword{fail}$ then use
          $\textsc{ETV-Ord-Sucs1}$ with $\mathcal{C\!S}$,
          $\mathcal{C\!S}'$.
        \item Case $\mathit{vfres}'' = \mathit{exres}$ where $\mathit{exres}
          \neq \keyword{fail}$ then use   $\textsc{ETV-Exc2}$ with $\mathcal{C\!S}$,
          $\mathcal{C\!S}'$.
        \end{itemize}
        \end{itemize}
      \item Case $\mathit{vres}' \neq \mathit{vfres}'$:
        Necessarily, we then have that $\mathit{vres}'$ = $\mathit{exres}$ where
        $\mathit{exres} \neq \keyword{fail}$ and so we use $\textsc{ETV-Exc1}$
        with $\mathcal{C\!S}$.
      \end{itemize}
    \end{itemize}
  \end{proof}

  \begin{lemma}
    \label{lem:evaltdvisitstarprogress}
    It is possible to construct a derivation
    $\evaltdvisitstarfuel{\esequence{\mathit{cs}}}{\esequence{v}}{\sigma}{\mathit{br}}{\mathit{vres}{\star}}{\sigma'}{n}$
    for any $\esequence{\mathit{cs}}$, well-typed value sequence $\esequence{v}$,
    well-typed store $\sigma$, breaking strategy $\mathit{br}$ and fuel $n$.
  \end{lemma}
  \begin{proof}
    By induction on $n$:
    \begin{itemize}
    \item Case $n = 0$ then use the corresponding
      $\keyword{timeout}$-derivation.
    \item Case $n > 0$:
      By case analysis on $\esequence{v}$:

      \begin{itemize}
      \item Case $\esequence{v} = \varepsilon$ then use $\textsc{ETVS-Emp}$.
      \item Case $\esequence{v} = v', \esequence{v''}$:

      By induction hypothesis given by \Cref{lem:evaltdvisitprogress} on $n - 1$
      with $\esequence{\mathit{cs}}$, $v$ and $\sigma$ we get a derivation
      $\begin{array}{c} \mathcal{V\!T} \\
         \evaltdvisitfuel{\esequence{\mathit{cs}}}{v}{\sigma}{\mathit{br}}{\mathit{vres}''}{\sigma''}{n-1}
         \end{array}$.

      We proceed by checking whether $\mathit{vres}''$ is syntactically a
      $\mathit{vfres}''$ for some $\mathit{vfres}''$:

      \begin{itemize}
      \item Case  $\mathit{vres}'' = \mathit{vfres}''$:
        
        We proceed by checking whether $\mathit{vfres}'' = \keyword{success} \;
        v'''$ and $\mathit{br} = \keyword{break}$:
        \begin{itemize}
        \item Case $\mathit{vfres}'' = \keyword{success} \; v'''$ and
          $\mathit{br} = \keyword{break}$ then use $\textsc{ETVS-Break}$ with
          $\mathcal{V\!T}$.
        \item Case $\mathit{vfres}'' \neq \keyword{success} \; v'''$ or
          $\mathit{br} \neq \keyword{break}$:

          By Boolean logic, we have that $\mathit{br} \neq
            \keyword{break} => \mathit{vfres}'' = \keyword{fail}$ is satisfied.

            By induction hypothesis on $n - 1$ with $\esequence{\mathit{cs}}$,
            $\esequence{v''}$ and $\sigma''$ we get a derivation

        $\begin{array}{c} \mathcal{V\!T\!S} \\
           \evaltdvisitstarfuel{\esequence{\mathit{cs}}}{\esequence{v''}}{\sigma''}{\mathit{br}}{\mathit{vres}{\star}'''}{\sigma'}{n
           - 1} \end{array}$.

       By case analysis on $\mathit{vres}{\star}'''$:
       
       \begin{itemize}
       \item Case $\mathit{vres}{\star}''' = \mathit{vfres}{\star}'''$ then use
         $\textsc{ETVS-More}$ with $\mathcal{V\!T}$ and $\mathcal{V\!T\!S}$.
       \item Case $\mathit{vres}{\star}''' = \mathit{exres}$ where
         $\mathit{exres} \neq \keyword{fail}$ then use $\textsc{ETVS-Exc2}$  with $\mathcal{V\!T}$ and $\mathcal{V\!T\!S}$.
       \end{itemize}
        \end{itemize}
      \item Case  $\mathit{vres}'' \neq \mathit{vfres}''$:
        
        Here we then have that $\mathit{vres}'' = \mathit{exres}$, where
        $\mathit{exres} \neq \keyword{fail}$ and so we use $\textsc{ETVS-Exc1}$.
      \end{itemize}
      \end{itemize}
    \end{itemize}
  \end{proof}

  \begin{lemma}
    \label{lem:evalbuvisitprogress}
    It is possible to construct a derivation
    $\evalbuvisitfuel{\esequence{\mathit{cs}}}{v}{\sigma}{\mathit{br}}{\mathit{vres}}{\sigma'}{n}$ for any $\esequence{\mathit{cs}}$, well-typed value $v$,
    well-typed store $\sigma$, breaking strategy $\mathit{br}$ and fuel $n$.
  \end{lemma}
  \begin{proof}
    By case analysis on $n$:
    \begin{itemize}
    \item Case $n = 0$ then use the corresponding
      $\keyword{timeout}$-derivation.
    \item Case $n > 0$:

      Let $\esequence{v''} = \textrm{children}(v)$. We know that
      $\esequence{v''}$ is well-typed from \Cref{lem:childrentyping}.
      
      By induction hypothesis given by \Cref{lem:evalbuvisitstarprogress} on $n
      - 1$ with $\esequence{\mathit{cs}}$, $\esequence{v''}$ and $\sigma$ we get
      a derivation $\begin{array}{c} \mathcal{V\!B\!S} \\
                      \evalbuvisitstarfuel{\esequence{\mathit{cs}}}{\esequence{v''}}{\sigma}{\mathit{br}}{\mathit{vres}''}{\sigma''}{n
                      - 1} \end{array}$. Recall that the output is well-typed from \Cref{lem:evalbuvisitstartyping}.

    By case analysis on $\mathit{vres}{\star}''$:
    \begin{itemize}
    \item Case $\mathit{vres}{\star}'' = \mathit{vfres}''$:
      
      We proceed by case analysis on $\mathit{vfres}''$
      \begin{itemize}
      \item Case  $\mathit{vfres}'' = \keyword{success} \;
        \esequence{v'''}$:
        We proceed by case analysis on $\mathit{br}$:

        \begin{itemize}
        \item Case $\mathit{br} = \keyword{break}$ then use
        $\textsc{EBU-Break-Sucs}$ with $\mathcal{V\!B\!S}$ and the derivation
        from applying \Cref{lem:reconsprogress} on $v$ and $\esequence{v'''}$.
        \item Case $\mathit{br} = \keyword{no-break}$:

        By applying \Cref{lem:reconsprogress} on $v$ and $\esequence{v'''}$ we
        get a derivation $\begin{array}{c} \mathcal{R\!C} \\
                            \reconstruct{v}{\esequence{v'''}}{\mathit{rcres}} \end{array}$.
         
        By case analysis on $\mathit{rcres}$:
        \begin{itemize}
        \item Case $\mathit{rcres} = \keyword{success} \; v'$:

          By induction hypothesis given by \Cref{lem:evalcasesprogress} on $n -
          1$ with $\esequence{\mathit{cs}}$, $v'$ and $\sigma''$ we get a
          derivation $\begin{array}{c} \mathcal{C\!S} \\
                        \evalcasesfuel{\esequence{\mathit{cs}}}{v'}{\sigma''}{\mathit{vres}'}{\sigma'}{n-1} \end{array}$

          By case analysis on $\mathit{vres}'$:
          \begin{itemize}
          \item Case $\mathit{vres}' = \mathit{vfres}'$ then use
            $\textsc{EBU-No-Break-Sucs}$ with $\mathcal{V\!B\!S}$,
            $\mathcal{R\!C}$ and $\mathcal{C\!S}$ .
          \item Case $\mathit{exres}$ then use $\textsc{EBU-No-Break-Exc}$  with $\mathcal{V\!B\!S}$,
            $\mathcal{R\!C}$ and $\mathcal{C\!S}$ .
          \end{itemize}
        \item Case $\mathit{rcres} = \keyword{error}$ then use $\textsc{EBU-No-Break-Err}$.
        \end{itemize}
        \end{itemize} 
      \item Case  $\mathit{vfres}'' = \keyword{fail}$ then use
        $\textsc{EBU-Fail-Sucs}$ with the induction hypothesis given by
        \Cref{lem:evalcasesprogress} on $n - 1$ with $\esequence{\mathit{cs}}$,
        $v$ and $\sigma''$.
      \end{itemize}
    \item Case $\mathit{vres}{\star}'' = \mathit{exres}$ where $\mathit{exres} \neq
      \keyword{fail}$ then use $\textsc{EBU-Exc}$ with $\mathcal{V\!B\!S}$.
    \end{itemize}
    \end{itemize}
  \end{proof}

  \begin{lemma}
    \label{lem:evalbuvisitstarprogress}
    It is possible to construct a derivation
    $\evalbuvisitstarfuel{\esequence{\mathit{cs}}}{\esequence{v}}{\sigma}{\mathit{br}}{\mathit{vres}{\star}}{\sigma'}{n}$
    for any $\esequence{\mathit{cs}}$, well-typed value sequence $\esequence{v}$,
    well-typed store $\sigma$, breaking strategy $\mathit{br}$ and fuel $n$.
  \end{lemma}
  \begin{proof}
    By induction on $n$:
    \begin{itemize}
    \item Case $n = 0$ then use the corresponding
      $\keyword{timeout}$-derivation.
    \item Case $n > 0$:

      By case analysis on $\esequence{v}$:
      \begin{itemize}
      \item Case $\esequence{v} = \varepsilon$ then use $\textsc{EBUS-Emp}$.
      \item Case $\esequence{v} = v', \esequence{v''}$:
        By induction hypothesis given by \Cref{lem:evalbuvisitprogress} on $n-1$
        with $\esequence{\mathit{cs}}$, $v'$ and $\sigma$ we get a derivation
        $\begin{array}{c} \mathcal{V\!B} \\
           \evalbuvisitfuel{\esequence{\mathit{cs}}}{v'}{\sigma}{\mathit{br}}{\mathit{vres}''}{\sigma''}{n-1} \end{array}$.

         By case analysis on $\mathit{vres}''$:
         \begin{itemize}
         \item Case $\mathit{vres}'' = \mathit{vfres}''$:

           By case analysis on $\mathit{vfres}''$ and $\mathit{br}$:
           \begin{itemize}
           \item Case $\mathit{vfres}'' = \keyword{success} \; v'''$ and
             $\mathit{br} = \keyword{break}$: then use $\textsc{EBUS-Break}$
             with $\mathcal{V\!B}$.
           \item  Case $\mathit{vfres}'' \neq \keyword{success} \; v'''$ or
             $\mathit{br} \neq \keyword{break}$:
             
             By Boolean logic we have $\mathit{br} = \keyword{break} =>
             \mathit{vfres} = \keyword{fail}$.

             By induction hypothesis on $n - 1$ with $\esequence{\mathit{cs}}$,
             $\esequence{v''}$ and $\sigma$ we get a derivation  $\begin{array}{c} \mathcal{V\!B\!S} \\
                                                                    \evalbuvisitstarfuel{\esequence{\mathit{cs}}}{\esequence{v''}}{\sigma''}{\mathit{br}}{\mathit{vres}{\star}'}{\sigma'}{n-1} \end{array}$.
                                                                  
             By case analysis on $\mathit{vres}{\star}'$:
             \begin{itemize}
             \item Case $\mathit{vres}{\star}' = \mathit{vfres}{\star}'$ then use
               $\textsc{EBUS-More}$ with  $\mathcal{V\!B}$ and  $\mathcal{V\!B\!S}$
             \item Case $\mathit{vres}{\star}' = \mathit{exres}$ then use
               $\textsc{EBUS-Exc2}$ with  $\mathcal{V\!B}$ and  $\mathcal{V\!B\!S}$
             \end{itemize}
           \end{itemize}
         \item Case $\mathit{vres}'' = \mathit{exres}$ then use
           $\textsc{EBUS-Exc1}$ with $\mathcal{V\!B}$.
         \end{itemize}
      \end{itemize}
    \end{itemize}
  \end{proof}

  \partialprogress*
  \begin{proof}
    By induction on $n$:
    \begin{itemize}
    \item Case $n = 0$ then use the corresponding
      $\keyword{timeout}$-derivation.
    \item Case $n > 0$:

      By case analysis on syntax $e$:
      \begin{itemize}
      \item Case $e = \mathit{vb}$ then use $\textsc{T-Basic}$.
      \item Case $e = x$

        We proceed by checking $x \in \mathrm{dom} \; \sigma$:
        \begin{itemize}
        \item Case $x \in \mathrm{dom} \; \sigma$ then use $\textsc{E-Var-Sucs}$.
        \item Case $x \notin \mathrm{dom} \; \sigma$ then use $\textsc{E-Var-Err}$.
        \end{itemize}
      \item Case $e = \ominus \; e'$:
        
        By induction hypothesis on $n - 1$ with $e'$ and $\sigma$ we get a
        derivation $\evalexprfuel{e'}{\sigma}{\mathit{vres}'}{\sigma'}{n - 1}$.

        By case analysis on $\mathit{vres}'$:
        \begin{itemize}
        \item Case $\mathit{vres}' = \keyword{success} v$ then use
          $\textsc{E-Un-Sucs}$ with above derivation.
        \item Case $\mathit{vres}' = \mathit{exres}$ then use
          $\textsc{E-Un-Exc}$ with above derivation.
        \end{itemize}
      \item Case $e = e_1 \oplus e_2$:
        By induction hypothesis on $n - 1$ with $e_1$ and $\sigma$ we get a
        derivation $\begin{array}{c} \mathcal{E}_1 \\
          \evalexprfuel{e_1}{\sigma}{\mathit{vres}_1}{\sigma''}{n - 1} \end{array}$.

      By case analysis on $\mathit{vres}_1$:
      \begin{itemize}
      \item Case  $\mathit{vres}_1 = \keyword{success} \; v_1$:
        By induction hypothesis on $n - 1$ with $e_2$ and $\sigma''$ we get a
        derivation $\begin{array}{c} \mathcal{E}_2 \\
          \evalexprfuel{e_2}{\sigma''}{\mathit{vres}_2}{\sigma'}{n - 1} \end{array}$.
        
        By case analysis on $\mathit{vres}_2$:
        \begin{itemize}
        \item Case  $\mathit{vres}_2 = \keyword{success} \; v_2$ then use
          $\textsc{E-Bin-Sucs}$ with  $\mathcal{E}_1$ and  $\mathcal{E}_2$.
        \item Case $\mathit{vres}_2 = \mathit{exres}$ then use
          $\textsc{E-Bin-Exc2}$ with  $\mathcal{E}_1$ and  $\mathcal{E}_2$.
        \end{itemize}
        
      \item Case  $\mathit{vres}_1 = \mathit{exres}$ then use
        $\textsc{E-Bin-Exc1}$ with $\mathcal{E}_1$.
      \end{itemize}

      \item Case $e = k(\esequence{e'})$:

        Recall that all derivations in our paper are assumed to be well-scoped
        so there must exist a corresponding  data-type $\mathit{at}$ that has
        $k(\esequence{t})$ as a constructor.

        By using induction hypothesis given by \Cref{lem:evalstarprogress} on $n
        - 1$ with $\esequence{e'}$ and $\sigma$ we get a derivation
        $\begin{array}{c} \mathcal{E\!S} \\
           \evalexprstarfuel{\esequence{e'}}{\sigma}{\mathit{vres}{\star}'}{\sigma'}{n - 1}  \end{array}$.

         By case analysis on $\mathit{vres}{\star}'$:
       \begin{itemize}
       \item Case $\mathit{vres}{\star}' = \keyword{success} \; \esequence{v'}$:

         By \Cref{lem:evalstartyping} we know that $\esequence{v'}$ is
         well-typed, i.e. that we have $\esequence{\typing{v'}{t'}}$ for some
         type sequence $t'$.
         We proceed to check whether all values in the sequence are
         non-$\blacksquare$ and each have a type $t'_i$ that is a subtype of the
         target type $t_i$.

         \begin{itemize}
         \item Case $\esequence{v \neq \blacksquare}$ and
           $\esequence{\subtyping{t'}{t}}$ then use $\textsc{E-Cons-Sucs}$ with
           $\mathcal{E\!S}$ and the required typing derivations.
         \item Case $v_i = \blacksquare$ or
           $\subtyping{t'_i}{t_i}$ for some $i$ then use $\textsc{E-Cons-Err}$
           with  $\mathcal{E\!S}$.
         \end{itemize}
       \item Case $\mathit{vres}{\star}' = \mathit{exres}$ then use
         $\textsc{E-Cons-Exc}$ with $\mathcal{E\!S}$
       \end{itemize}
      \item Case $e = [\esequence{e'}]$:

        By induction hypothesis given by \Cref{lem:evalstarprogress} on $n - 1$
        with $\esequence{e'}$ and $\sigma$ we get a derivation
        $\evalexprstarfuel{\esequence{e'}}{\sigma}{\mathit{vres}{\star}'}{\sigma'}{n-1}$.

        By case analysis on $\mathit{vres}{\star}'$:
        \begin{itemize}
        \item Case $\mathit{vres}{\star}' = \keyword{success} \; \esequence{v}$:
          
          We proceed by checking whether all values $\esequence{v}$ are
          non-$\blacksquare$:
          \begin{itemize}
          \item Case $\esequence{v \neq \blacksquare}$ then use
            $\textsc{E-List-Sucs}$ with above derivation.
          \item Case $v_i \neq \blacksquare$ for some $i$ then use
            $\textsc{E-List-Err}$ with above derivation.
          \end{itemize}
        \item Case $\mathit{vres}{\star}' = \mathit{exres}$ then use
          $\textsc{E-List-Exc}$ with above derivation.
        \end{itemize}
      \item Case $e = \{\esequence{e'}\}$

            By induction hypothesis given by \Cref{lem:evalstarprogress} on $n - 1$
        with $\esequence{e'}$ and $\sigma$ we get a derivation
        $\evalexprstarfuel{\esequence{e'}}{\sigma}{\mathit{vres}{\star}'}{\sigma'}{n-1}$.

        By case analysis on $\mathit{vres}{\star}'$:
        \begin{itemize}
        \item Case $\mathit{vres}{\star}' = \keyword{success} \; \esequence{v}$:
          
          We proceed by checking whether all values $\esequence{v}$ are
          non-$\blacksquare$:
          \begin{itemize}
          \item Case $\esequence{v \neq \blacksquare}$ then use
            $\textsc{E-Set-Sucs}$ with above derivation.
          \item Case $v_i \neq \blacksquare$ for some $i$ then use
            $\textsc{E-Set-Err}$ with above derivation.
          \end{itemize}
        \item Case $\mathit{vres}{\star}' = \mathit{exres}$ then use
          $\textsc{E-Set-Exc}$ with above derivation.
        \end{itemize}
      \item Case $e = (\esequence{e' : e''})$

            By induction hypothesis given by \Cref{lem:evalstarprogress} on $n - 1$
            with $\esequence{e',e''}$ and $\sigma$ we get a derivation
        $\evalexprstarfuel{\esequence{e',e''}}{\sigma}{\mathit{vres}{\star}'}{\sigma'}{n-1}$.

        By case analysis on $\mathit{vres}{\star}'$:
        \begin{itemize}
        \item Case $\mathit{vres}{\star}' = \keyword{success} \; \esequence{v,v'}$:
          
          We proceed by checking whether all values $\esequence{v}$ and $\esequence{v'}$ are
          non-$\blacksquare$:
          \begin{itemize}
          \item Case $\esequence{v \neq \blacksquare}$ and $\esequence{v' \neq \blacksquare}$ then use
            $\textsc{E-Map-Sucs}$ with above derivation.
          \item Case $v_i \neq \blacksquare$ or $v_i' \neq \blacksquare$ for some $i$ then use
            $\textsc{E-Map-Err}$ with above derivation.
          \end{itemize}
        \item Case $\mathit{vres}{\star}' = \mathit{exres}$ then use
          $\textsc{E-Map-Exc}$ with above derivation.
        \end{itemize}
      \item Case $e = e_1[e_2]$
        
        By induction hypothesis on $n - 1$ with $e_1$ and $\sigma$ we get a
        derivation $\begin{array}{c} \mathcal{E} \\
                      \evalexprfuel{e_1}{\sigma}{\mathit{vres}_1}{\sigma''}{n-1} \end{array}$.
        
          By case analysis on $\mathit{vres}_1$:
          \begin{itemize}
          \item Case $\mathit{vres}_1 = \keyword{success} \; v_1$:
            
            By case analysis on $v_1$:
            \begin{itemize}
            \item Case $v_1 = (\esequence{v' : v''})$:
              
              By induction hypothesis on $n - 1$ with $e_2$ and $\sigma''$ we
              get a derivation
              $\begin{array}{c} \mathcal{E}' \\
                      \evalexprfuel{e_2}{\sigma''}{\mathit{vres}_2}{\sigma'}{n-1} \end{array}$.
                    
                    By case analysis on $\mathit{vres}_2$:
                    \begin{itemize}
                    \item Case $\mathit{vres}_2 = \keyword{success} \; v_2$:

                      We proceed to check whether $\exists i . v'_i = v_2$:
                      \begin{itemize}
                      \item Case $v'_i = v_2$ then use $\textsc{E-Lookup-Sucs}$
                        with $\mathcal{E}$ and $\mathcal{E}'$.
                      \item Case $\nexists i . v'_i = v_2$ then use
                        $\textsc{E-Lookup-NoKey}$ with $\mathcal{E}$ and $\mathcal{E}'$.
                      \end{itemize}
                    \item Case $\mathit{vres}_2 = \mathit{exres}$ then use
                      $\textsc{E-Lookup-Exc2}$ with $\mathcal{E}$ and $\mathcal{E}'$.
                    \end{itemize}
            \item Case $v_1 \neq (\esequence{v' : v''})$ then use
              $\textsc{E-Lookup-Err}$ with $\mathcal{E}$.
            \end{itemize}
          \item Case $\mathit{vres}_1 = \mathit{exres}$ then use
            $\textsc{E-Lookup-Exc1}$ with $\mathcal{E}$.
          \end{itemize}
      \item Case $e = e_1[e_2 = e_3]$
        
        By induction hypothesis on $n - 1$ with $e_1$ and $\sigma$ we get a
        derivation $\begin{array}{c} \mathcal{E} \\
                      \evalexprfuel{e_1}{\sigma}{\mathit{vres}_1}{\sigma'''}{n - 1}
                    \end{array}$.
        
        By case analysis on $\mathit{vres}_1$:
        \begin{itemize}
        \item Case $\mathit{vres}_1 = \keyword{success} \; v_1$
          
          By case analysis on $v_1$:
          \begin{itemize}
          \item Case $v_1 = (\esequence{v' : v''})$:
            
            By induction hypothesis on $n - 1$ with $e_2$ and $\sigma'''$ we get a
            derivation $\begin{array}{c} \mathcal{E}' \\
                          \evalexprfuel{e_2}{\sigma'''}{\mathit{vres}_2}{\sigma''}{n - 1}
                        \end{array}$.
                        
            By case analysis on $\mathit{vres}_2$:
            \begin{itemize}
            \item Case $\mathit{vres}_2 = \keyword{success} \; v_2$:

                    By induction hypothesis on $n - 1$ with $e_3$ and $\sigma''$ we get a
            derivation $\begin{array}{c} \mathcal{E}'' \\
                          \evalexprfuel{e_3}{\sigma''}{\mathit{vres}_3}{\sigma'}{n - 1}
                        \end{array}$.

              By case analysis on $\mathit{vres}_3$:
              \begin{itemize}
              \item Case $\mathit{vres}_3 = \keyword{success} \; v_3$:

                We proceed to check whether $v_2$ and $v_3$ are
                non-$\blacksquare$:
                \begin{itemize}
                \item Case $v_2 \neq \blacksquare$ and $v_3 \neq \blacksquare$
                  then use $\textsc{E-Update-Sucs}$ with $\mathcal{E}$,
                  $\mathcal{E}'$ and $\mathcal{E}''$.
                \item Case $v_2 = \blacksquare$ or $v_3 = \blacksquare$ then use
                  $\textsc{E-Update-Err2}$ with $\mathcal{E}$,
                  $\mathcal{E}'$ and $\mathcal{E}''$.
                \end{itemize}
              \item Case $\mathit{vres}_3 = \mathit{exres}$ then use
                $\textsc{E-Update-Exc3}$ with $\mathcal{E}$ and $\mathcal{E}'$.
              \end{itemize}
            \item Case $\mathit{vres}_2 = \mathit{exres}$ then use
              $\textsc{E-Update-Exc2}$ on $\mathcal{E}$ and $\mathcal{E}'$.
            \end{itemize}
         
          \item Case $v_1 \neq (\esequence{v' : v''})$ then use
            $\textsc{E-Update-Err1}$ with $\mathcal{E}$.
          \end{itemize}
        \item Case $\mathit{vres}_1 = \mathit{exres}$ then use
          $\textsc{E-Update-Exc1}$ with $\mathcal{E}$.
        \end{itemize}
                    
      \item Case $e = f(\esequence{e'})$

        By induction hypothesis given by \Cref{lem:evalstarprogress} on $n - 1$
        with $\esequence{e'}$ and $\sigma$ we get a derivation
        $\begin{array}{c} \mathcal{E\!S} \\
           \evalexprstarfuel{\esequence{e'}}{\sigma}{\mathit{vres}{\star}''}{\sigma''}{n-1} \end{array}$
         
         By case analysis on $\mathit{vres}{\star}''$:
         \begin{itemize}
         \item Case $\mathit{vres}{\star}'' = \keyword{success} \; \esequence{v''}$

           Recall that we assume that our function calls are well-scoped and so
           there must exist a corresponding function definition $\keyword{fun}
           \; t' \; f(\esequence{t \; x}) = e''$.
           By \Cref{lem:evalstartyping} we know that we have
           $\esequence{\typing{v''}{t''}}$ for some type sequence
           $\esequence{t''}$.
           
           We proceed to check whether $\esequence{\subtyping{t''}{t}}$:
           \begin{itemize}
           \item Case $\esequence{\subtyping{t''}{t}}$

             Let $\esequence{\keyword{global} \; t_y \; y}$ represent all global
             variable definitions.

             By induction hypothesis on $n - 1$ with $e''$ and $[\esequence{y
               \mapsto \sigma''(y)}, \esequence{x \mapsto v''}]$ we get a
             derivation
             $\begin{array}{c} \mathcal{E} \\
              \evalexprfuel{e'}{[\esequence{y
                \mapsto \sigma''(y)}, \esequence{x \mapsto
                v''}]}{\mathit{vres}'}{\sigma'''}{n-1} \end{array}$.

             By case analysis on $\mathit{vres}'$:
             \begin{itemize}
             \item Case $\mathit{vres}' = \keyword{success} \; v'$ or
               $\mathit{vres}' = \keyword{return} \; v'$.

               By \Cref{thm:strongtyping} we know that $\typing{v'}{t'''}$ for
               some type $t'''$. 
               We proceed to check whether $\subtyping{t'''}{t'}$:
               \begin{itemize}
               \item Case $\subtyping{t'''}{t'}$ then use  $\textsc{E-Call-Res-Sucs}$ with $\mathcal{E\!S}$ and $\mathcal{E}$.
               \item Case $\notsubtyping{t'''}{t'}$ then use  $\textsc{E-Call-Res-Err1}$ with $\mathcal{E\!S}$ and $\mathcal{E}$.
               \end{itemize}
             \item Case $\mathit{vres}' = \keyword{throw} \; v'$ then use
               $\textsc{E-Call-Res-Exc}$ with $\mathcal{E\!S}$ and $\mathcal{E}$
             \item Case $\mathit{vres}' \in \{\keyword{break},
               \keyword{continue}, \keyword{fail}, \keyword{error}\}$ then use
               $\textsc{E-Call-Res-Err2}$ with $\mathcal{E\!S}$ and $\mathcal{E}$.
             \end{itemize}
           \item Case $\notsubtyping{t''_i}{t_i}$ for some $i$ then use
             $\textsc{E-Call-Arg-Err}$ with $\mathcal{E\!S}$.
           \end{itemize}
         \item Case $\mathit{vres}{\star}'' = \mathit{exres}$ then use
           $\textsc{E-Call-Arg-Exc}$ with $\mathcal{E\!S}$.
         \end{itemize}
      \item Case $e = \keyword{return} \; e'$

        By induction hypothesis on $n - 1$ with $e'$ and $\sigma$ we get a
        derivation $\evalexprfuel{e'}{\sigma}{\mathit{vres}'}{\sigma'}{n - 1}$.

        \begin{itemize}
        \item Case $\mathit{vres}' = \keyword{success} \; v$ then use
          $\textsc{E-Ret-Sucs}$ with above derivation.
        \item Case  $\mathit{vres}' = \mathit{exres}$ then use
          $\textsc{E-Ret-Exc}$ with above derivation.
        \end{itemize}
      \item Case $e = (x = e')$
        
        Recall that our definitions are assumed to be well-scoped and so there
        must exists either a $\keyword{local} \; t \; x$ or $\keyword{global} \;
        t \; x$ declaration for the variable (with no overshadowing).

        By induction hypothesis on $n - 1$ with $e'$ and $\sigma$ we get a
        derivation $\evalexprfuel{e'}{\sigma}{\mathit{vres}'}{\sigma''}{n - 1}$.
        
        By case analysis on $\mathit{vres}'$:

        \begin{itemize}
        \item Case $\mathit{vres}' = \keyword{success} \; v$:
          By \Cref{thm:strongtyping} we know that there exists a $t'$ such that
          $\typing{v}{t'}$.

          We proceed by checking whether $\subtyping{t'}{t}$:

          \begin{itemize}
          \item Case $\subtyping{t'}{t}$  then use $\textsc{E-Asgn-Sucs}$ with
            above evaluation and typing derivations.
          \item Case $\notsubtyping{t'}{t}$ then use $\textsc{E-Asgn-Err}$ with
            above evaluation and typing derivations.
          \end{itemize}
        \item Case $\mathit{vres}' = \mathit{exres}$ then use
          $\textsc{E-Asgn-Exc}$ with above derivation.
        \end{itemize}
      \item Case $e = \keyword{if} \; e_1 \; \keyword{then} \; e_2 \; \keyword{else}
        \; e_3$

        By induction hypothesis on $n - 1$ with $e_1$ and $\sigma$ we get a
        derivation  $\begin{array}{c} \mathcal{E} \\
                       \evalexprfuel{e_1}{\sigma}{\mathit{vres}''}{\sigma''}{n-1} \end{array}$.

        By case analysis on $\mathit{vres}''$:
        \begin{itemize}
        \item Case $\mathit{vres}'' = \keyword{success} \; v''$

          By case analysis on $v''$:
          \begin{itemize}
          \item Case $v'' = \textrm{false}()$ then use $\textsc{E-If-False}$
            with $\mathcal{E}$ and the derivation from the induction hypothesis
            on $n - 1$ with $e_2$ and $\sigma''$.
          \item Case $v'' = \textrm{true}()$ then use $\textsc{E-If-True}$
            with $\mathcal{E}$ and the derivation from the induction hypothesis
            on $n - 1$ with $e_3$ and $\sigma''$.
          \item Case $v'' \neq \textrm{true}()$ and $v'' \neq \textrm{false}()$ then use $\textsc{E-If-Err}$
            with $\mathcal{E}$.
          \end{itemize}
        \item Case $\mathit{vres}'' = \mathit{exres}$ then use
          $\textsc{E-If-Exc}$ with $\mathcal{E}$.
        \end{itemize}
      \item Case $e = \keyword{switch} \; e' \; \keyword{do} \; \esequence{\mathit{cs}}$

        By induction hypothesis on $n - 1$ with $e'$ and $\sigma$ we get a
        derivation  $\begin{array}{c} \mathcal{E} \\
                       \evalexprfuel{e'}{\sigma}{\mathit{vres}''}{\sigma''}{n-1} \end{array}$.

      By case analysis on $\mathit{vres}''$:
      \begin{itemize}
      \item Case $\mathit{vres}'' = \keyword{success} \; v''$
        
        By induction hypothesis given by \Cref{lem:evalcasesprogress} on $n - 1$
        with $\esequence{\mathit{cs}}$, $v''$ and $\sigma''$ we get a derivation
        $\begin{array}{c} \mathcal{C\!S} \\ \evalcasesfuel{\mathit{\esequence{cs}}}{v''}{\sigma''}{\mathit{vres}'}{\sigma'}{n - 1} \end{array}$.

        By case analysis on $\mathit{vres}'$:
        \begin{itemize}
        \item Case $\mathit{vres}' = \keyword{success} \; v'$ then use $\textsc{E-Switch-Sucs}$ with $\mathcal{E}$
          and $\mathcal{C\!S}$.
        \item Case $\mathit{vres}' = \keyword{fail}$ then use $\textsc{E-Switch-Fail}$ with $\mathcal{E}$
          and $\mathcal{C\!S}$.
        \item Case $\mathit{vres}' = \mathit{exres}$ where $\mathit{exres} \neq
          \keyword{fail}$ then use $\textsc{E-Switch-Exc2}$ with $\mathcal{E}$
          and $\mathcal{C\!S}$.
        \end{itemize}
              \item Case $\mathit{vres}'' = \mathit{exres}$ then use
        $\textsc{E-Switch-Exc1}$ with $\mathcal{E}$.
        \end{itemize}
      \item Case $e = \mathit{st} \;\keyword{visit} \; e' \; \keyword{do} \;
        \esequence{\mathit{cs}}$:

        By induction hypothesis on $n - 1$ with $e'$ and $\sigma$ we get a
        derivation  $\begin{array}{c} \mathcal{E} \\
                       \evalexprfuel{e'}{\sigma}{\mathit{vres}''}{\sigma''}{n-1} \end{array}$.

      By case analysis on $\mathit{vres}''$:
      \begin{itemize}
      \item Case $\mathit{vres}'' = \keyword{success} \; v''$
        
        By induction hypothesis given by \Cref{lem:evalvisitprogress} on $n - 1$
        with $\esequence{\mathit{cs}}$, $v''$ and $\sigma''$ we get a derivation
        $\begin{array}{c} \mathcal{V} \\ \evalvisitfuel{\mathit{st}}{\mathit{\esequence{cs}}}{v''}{\sigma''}{\mathit{vres}'}{\sigma'}{n
        - 1} \end{array}$.

        By case analysis on $\mathit{vres}'$:
        \begin{itemize}
        \item Case $\mathit{vres}' = \keyword{success} \; v'$ then use $\textsc{E-Visit-Sucs}$ with $\mathcal{E}$
          and $\mathcal{V}$.
        \item Case $\mathit{vres}' = \keyword{fail}$ then use $\textsc{E-Visit-Fail}$ with $\mathcal{E}$
          and $\mathcal{V}$.
        \item Case $\mathit{vres}' = \mathit{exres}$ where $\mathit{exres} \neq
          \keyword{fail}$ then use $\textsc{E-Visit-Exc2}$ with $\mathcal{E}$
          and $\mathcal{V}$.
        \end{itemize}
      \item Case $\mathit{vres}'' = \mathit{exres}$ then use
        $\textsc{E-Visit-Exc1}$ with $\mathcal{E}$.
      \end{itemize}
      \item Case $e = \keyword{break}$ then use $\textsc{E-Break}$.
      \item Case $e = \keyword{continue}$ then use $\textsc{E-Continue}$.
      \item Case $e = \keyword{fail}$ then use $\textsc{E-Fail}$.
      \item Case $e = \keyword{local} \; \esequence{t \; x} \; \keyword{in} \;
        \esequence{e'} \; \keyword{end}$ then use either $\textsc{E-Block-Sucs}$
        (when it produces a successful result) or $\textsc{E-Block-Exc}$ (otherwise) with the
        derivation of the induction hypothesis given by
        \Cref{lem:evalstarprogress} on $n - 1$ with $\esequence{e'}$ and $\sigma$.
      \item Case $e = \keyword{for} \; g \; e'$:

        By induction hypothesis given by \Cref{lem:evalgenprogress} on $n - 1$
        with $g$ and $\sigma$ we get a derivation
        $\evalgenexprfuel{g}{\sigma}{\mathit{envres}}{\sigma''}{n-1}$.

        By case analysis on $\mathit{envres}$:
        \begin{itemize}
        \item Case $\mathit{envres} = \keyword{success} \; \esequence{\rho}$
          then use $\textsc{E-For-Sucs}$ with above derivation and the
          derivation from the induction hypothesis given by
          \Cref{lem:evaleachprogress} on $n - 1$ with $e'$, $\esequence{\rho}$ and $\sigma''$.
        \item Case $\mathit{envres} = \mathit{exres}$ then use
          $\textsc{E-For-Exc}$ with above derivation.
        \end{itemize}
      \item Case $e = \keyword{while} \; e_1 \; e_2$

        By induction hypothesis on $n - 1$ with $e_1$ and $\sigma$ we get a
        derivation: $\begin{array}{c} \mathcal{E} \\
                       \evalexprfuel{e_1}{\sigma}{\mathit{vres}''}{\sigma''}{n-1} \end{array}$.

        By case analysis on $\mathit{vres}''$:
        \begin{itemize}
        \item Case $\mathit{vres}'' = \keyword{success} \; v''$
          
          By case analysis on $v''$:
          \begin{itemize}
          \item Case $v'' = \textrm{false}()$ then use $\textsc{E-While-False}$
            with $\mathcal{E}$.
          \item Case $v'' = \textrm{true}()$:

            By induction hypothesis on $n - 1$ with $e_2$ and $\sigma''$ we get
            a derivation  $\begin{array}{c} \mathcal{E}' \\
                       \evalexprfuel{e_2}{\sigma''}{\mathit{vres}'''}{\sigma'''}{n-1} \end{array}$.
           
                     By case analysis on $\mathit{vres}'''$:
                     
                     \begin{itemize}
                     \item Case $\mathit{vres}''' = \keyword{success} \; v'''$
                       or $\mathit{vres}''' = \keyword{continue}$ then use
                       $\textsc{E-While-True-Sucs}$ with $\mathcal{E}$,
                       $\mathcal{E}'$ and the derivation from the induction
                       hypothesis on $n - 1$ with $\keyword{while} \; e_1 \;
                       e_2$ and $\sigma'''$.
                     \item Case $\mathit{vres}''' = \keyword{break}$ then use
                       $\textsc{E-While-True-Break}$ with $\mathcal{E}$ and $\mathcal{E'}$.
                     \item Case $\mathit{vres}''' = \mathit{exres} \in \{\keyword{throw} \; v''',
                       \keyword{return} \; v''', \keyword{fail},
                       \keyword{error}\}$  then use $\textsc{E-While-Exc2}$ with
                       $\mathcal{E}$ and $\mathcal{E}'$.
                     \end{itemize}
          \item Case $v'' \neq \textrm{true}()$ and $v'' \neq \textrm{false}()$ then use $\textsc{E-While-Err}$
          \end{itemize}
        \item Case $\mathit{vres}'' = \mathit{exres}$ then use
          $\textsc{E-While-Exc1}$ with $\mathcal{E}$.
        \end{itemize}
      \item Case $e = \keyword{solve} \; \esequence{x} \; e'$:
        
        By induction hypothesis on $n - 1$ with $e'$ and $\sigma$ we get a
        derivation $\begin{array}{c} \mathcal{E} \\ \evalexprfuel{e'}{\sigma}{\mathit{vres}''}{\sigma''}{n-1} \end{array}$.
        
        By case analysis on $\mathit{vres}''$
        \begin{itemize}
        \item Case $\mathit{vres}'' = \keyword{success} \; v''$:

          We proceed by checking whether the variable sequence $\esequence{x}$
          is both in old and new stores.
          \begin{itemize}
          \item Case $\esequence{x} \subseteq \mathrm{dom} \; \sigma \cap
            \mathrm{dom} \; \sigma''$:
            
            We then check whether any value in $\esequence{x}$ has changed:
            \begin{itemize}
            \item Case $\esequence{\sigma(x) = \sigma''(x)}$ then use
              $\textsc{E-Solve-Eq}$ with $\mathcal{E}$.
            \item Case $\sigma(x_i) \neq \sigma''(x_i)$ for some $i$ then use
              $\textsc{E-Solve-Neq}$ with $\mathcal{E}$ and the derivation from
              the induction hypothesis on $n-1$ with $\keyword{solve} \;
              \esequence{x} \; e'$ and $\sigma''$.
            \end{itemize}
          \item Case $x_i \notin \mathrm{dom} \; \sigma \cap \mathrm{dom} \;
            \sigma''$ for some $i$ then use $\textsc{E-Solve-Err}$ with $\mathcal{E}$.
          \end{itemize}
        \item Case $\mathit{vres}'' = \mathit{exres}$ then use
          $\textsc{E-Solve-Exc}$ with $\mathcal{E}$.
        \end{itemize}
      \item Case $e = \keyword{throw} \;  e'$:
        
        By induction hypothesis on $n - 1$ with $e'$ and $\sigma$ we get a
        derivation $\evalexprfuel{e'}{\sigma}{\mathit{vres}'}{\sigma'}{n - 1}$.

        \begin{itemize}
        \item Case $\mathit{vres}' = \keyword{success} \; v$ then use
          $\textsc{E-Thr-Sucs}$ with above derivation.
        \item Case  $\mathit{vres}' = \mathit{exres}$ then use
          $\textsc{E-Thr-Exc}$ with above derivation.
        \end{itemize}
      \item Case $e = \keyword{try} \;  e_1 \; \keyword{catch} \; x => e_2$:

        By induction hypothesis on $n - 1$ with $e_1$ and $\sigma$ we get a
        derivation $\begin{array}{c} \mathcal{E} \\
                      \evalexprfuel{e_1}{\sigma}{\mathit{vres}_1}{\sigma''}{n - 1} \end{array}$.

        By case analysis on $\mathit{vres}_1$:
        \begin{itemize}
        \item Case $\mathit{vres}_1 = \keyword{throw} \; v_1$ then use
          $\textsc{E-Try-Catch}$ with $\mathcal{E}$ and the derivation from the induction hypothesis
          on $n - 1$ with $e_2$ and $\sigma''[x \mapsto v_1]$.
        \item Case $\mathit{vres}_1 \neq \keyword{throw} \; v_1$ then use
          $\textsc{E-Try-Ord}$ with $\mathcal{E}$.
        \end{itemize}
      \item Case $e = \keyword{try} \;  e_1 \; \keyword{finally} \; e_2$
        
        By induction hypothesis on $n - 1$ with $e_1$ and $\sigma$ we get a
        derivation $\begin{array}{c} \mathcal{E} \\
                      \evalexprfuel{e_1}{\sigma}{\mathit{vres}_1}{\sigma''}{n - 1} \end{array}$.

        By induction hypothesis on $n - 1$ with $e_2$ and $\sigma''$ we get a
        derivation $\begin{array}{c} \mathcal{E}' \\
                      \evalexprfuel{e_2}{\sigma''}{\mathit{vres}_2}{\sigma'}{n - 1} \end{array}$.

        By case analysis on $\mathit{vres}_2$:
        \begin{itemize}
        \item Case $\mathit{vres}_2 = \keyword{success} \; v_2$ then use
          $\textsc{E-Fin-Sucs}$ with $\mathcal{E}$ and $\mathcal{E}'$
        \item Case $\mathit{vres}_2 = \mathit{exres}$ then use
          $\textsc{E-Fin-Exc}$ with $\mathcal{E}$ and $\mathcal{E}'$.
        \end{itemize}
      \end{itemize}
    \end{itemize}
  \end{proof}
  \termination*
  \begin{proof}
    The proof proceeds similarly to \Cref{thm:partialprogress}, except that instead of
    doing the induction on $n$---which we know need to provide---we do the
    induction on the relevant syntactic element starting with the
    $e_{\mathrm{fin}}$ for this theorem. The only major complication is that for
    the $\keyword{bottom-up}$ visit rules, we need to do an inner well-founded
    induction on the $\prec$ relation on values when traversing the children in order to terminate.

    The result $n$ is simply taking to be
    $n' = 1 + n$ where $n$ is the maximal fuel used in a sub-term.
  \end{proof}
\end{document}